\documentclass[12pt]{article}

\usepackage[margin=1in]{geometry}
\usepackage{
    amsmath, 
    amsthm, 
    amssymb, 
    amsfonts 
    }
\usepackage{thmtools}
\usepackage{thm-restate}
\usepackage{bbm}
\usepackage{relsize}
\usepackage{mathtools}
\usepackage[pdfstartview=FitH,colorlinks,linkcolor=blue,filecolor=blue,citecolor=blue,urlcolor=blue,pagebackref=true]{hyperref}  
\usepackage{graphicx}
\usepackage{cancel}
\usepackage{array}
\usepackage{mdframed} 
\usepackage{setspace}
\usepackage[ruled, vlined]{algorithm2e}
\usepackage{tablefootnote}
\usepackage{float}
\usepackage{cite}
\usepackage[utf8]{inputenc} 
\usepackage[T1]{fontenc}    
\usepackage{url}            
\usepackage{booktabs}       
\usepackage{microtype}      
\usepackage{tikz}           
\usetikzlibrary{quantikz2}
\usepackage{nicefrac}       
\usepackage{caption} 
\graphicspath{ {./images/} }
\usepackage{xcolor}
\usepackage{tcolorbox}

\makeatother
\usepackage{calligra}
\usepackage{braket}
\usepackage{calrsfs}
\usepackage[mathscr]{euscript}

\def\authorNote{}
\ifdefined\authorNote
  \newcommand{\yc}[1]{{\color{magenta}[YC: #1]}}
  \newcommand{\nai}[1]{{\color{brown}[Nai: #1]}}
  \newcommand{\cy}[1]{{\color{blue}[CY: #1]}}
  \newcommand{\ki}[1]{{\color{purple}[Ki: #1]}}
\else
  \newcommand{\yc}[1]{}
  \newcommand{\nai}[1]{}
  \newcommand{\cy}[1]{}
  \newcommand{\ki}[1]{}
\fi

\newcommand{\appendixsection}[1]{%
  \refstepcounter{section}%
  \section*{Appendix~\Alph{section}: #1}%
  \addcontentsline{toc}{section}{Appendix~\Alph{section}: #1}%
}

\newcommand{\email}[1]{\href{mailto:#1}{#1}}

\newtheorem{theorem}{Theorem}[section]

\newtheorem{lemma}[theorem]{Lemma}
\newtheorem*{restatedlemma}{Lemma}
\newtheorem*{restatedcorollary}{Corollary}

\newtheorem{corollary}[theorem]{Corollary}
\newtheorem{claim}[theorem]{Claim}

\newtheorem{remark}[theorem]{Remark}

\theoremstyle{definition}
\newtheorem{definition}[theorem]{Definition}
\usepackage[capitalise,nameinlink]{cleveref}
\Crefname{lemma}{Lemma}{Lemmas}
\Crefname{fact}{Fact}{Facts}
\Crefname{theorem}{Theorem}{Theorems}
\Crefname{corollary}{Corollary}{Corollaries}
\Crefname{claim}{Claim}{Claims}
\Crefname{example}{Example}{Examples}
\Crefname{problem}{Problem}{Problems}
\Crefname{definition}{Definition}{Definitions}
\Crefname{notation}{Notation}{Notations}
\Crefname{assumption}{Assumption}{Assumptions}
\Crefname{subsection}{Section}{Sections}
\Crefname{section}{Section}{Sections}
\Crefname{equation}{Equation}{Equations}
\Crefname{algorithm}{Algorithm}{Algorithm}
\Crefname{remark}{Remark}{Remark}
\DeclareMathOperator{\Tr}{\Tr}
\DeclareMathOperator{\poly}{poly}
\DeclareMathOperator{\spanv}{span}
\newcommand{\ketbra}[2]{\lvert#1\rangle\langle#2\rvert}
\title{
    Shadow Quantum Singular Value Transformation with Shallow Quantum Circuits
}

\author{
Nai-Hui Chia\thanks{Computer Science Department, Ken Kennedy Institute,  and Smalley-Curl Institute, Rice University, USA. Email: \email{nc67@rice.edu}.} \and
Hyunseong Kim\thanks{Applied Physics Graduate Program, Smelly-Curl Institute and Department of Computer Science, Rice University, USA. Email: \email{hk81@rice.edu}}\and
Chia-Ying Lin\thanks{Department of Physics and Astronomy, Rice University, USA. Email: \email{cl207@rice.edu}.}\and
Yu-Ching Shen\thanks{Department of Computer Science, Rice University, USA. Email: \email{ycshen@rice.edu}.}
}
\date{}
\begin{document}
\maketitle

\begin{abstract} 

We introduce \emph{shadow quantum singular value transformation} (S-QSVT): given an initial state $|{\psi}\rangle$, a Hermitian matrix $H$, a polynomial $f$, and a set of observables $\{O_1,\dots,O_m\}$, the goal is to estimate $\langle{\psi}|f(H)^{\dagger}O_j f(H)|{\psi}\rangle$ for all $j\in\{1,\dots,m\}$. S-QSVT provides a systematic route to reduce the quantum resources required by standard QSVT, which constructs a unitary block-encoding of $f(H)$. It uses structure in the input state and observables, together with the fact that many applications require only observable estimates rather than synthesizing the full unitary.

We present three algorithms that exploit structure in the initial state and observables to achieve quantum circuit depths that can be substantially smaller than the degree of $f$. First, we develop a state-aware QSVT algorithm that prepares the target state with low circuit depth when the Krylov subspace associated with $H$ and $|\psi\rangle$ is low-dimensional or admits an accurate low-dimensional approximation. Second, we introduce Classical Shadow S-QSVT, which constructs a classical representation from $H$ and $|\psi\rangle$ without prior knowledge of the observables or explicit preparation of the target state proportional to $f(H)|\psi\rangle$. This representation supports subsequent estimation of the target quantities for observables specified after the quantum computation. Finally, we develop an observable-aware S-QSVT algorithm that combines a new observable-aware Krylov subspace with history states to further reduce circuit depth and gate complexity. The first two algorithms require quantum circuit depth $O(r)$, while the third requires depth $O(s)$ with $s\leq r$, where $r$ and $s$ denote the dimensions of the Krylov subspace and the observable-aware Krylov subspace, respectively.

\end{abstract}

\tableofcontents
\newpage

\section{Introduction}

Quantum algorithms have the potential to outperform their classical
counterparts across a broad range of computational tasks.
Representative examples include quantum linear systems
algorithms~\cite{HHL09,childs2015QLinSysExpPrec,ambainis2010VTAA,subacsi2019quantum,lin2020optimal,costa2022optimal,dalzell2024shortcut,low2024quantum, morales2024quantum}
and Hamiltonian simulation~\cite{berry2007efficient,berry2015taylor,low2017optimal,low2019qubitization}.
Quantum singular value transformation
(QSVT)~\cite{gilyen2018QSingValTransf} provides a unifying framework for these and many other quantum algorithms. Given a block-encoding of a Hermitian matrix $H$ and a suitable polynomial $f$, QSVT constructs a unitary that block-encodes $f(H)$. By choosing polynomials that approximate the desired spectral transformations, this framework supports tasks such as Hamiltonian simulation, which implements $e^{-iHt}$, and quantum linear systems solving, which prepares a state proportional to $H^{-1}\ket{\psi}$ for a given input state $\ket{\psi}$.

Although QSVT underlies many provable quantum advantages, translating these provable advantages into end-to-end advantages in practical applications remains challenging under realistic resource constraints. In particular, early fault-tolerant quantum computers (EFTQCs) \cite{PRXQuantum.5.020101} are expected to operate with limited logical qubits, quantum memory, and coherent circuit depth, making many QSVT-based algorithms difficult to implement.

Notably, limited coherent circuit depth is one of the major obstacles to implementing QSVT on EFTQCs. Standard QSVT implementations~\cite{gilyen2018QSingValTransf} require a sequence of block-encoding queries whose length scales linearly with the polynomial degree, $\deg(f)$. Moreover, impossibility results for parallel fast-forwarding in Hamiltonian simulation~\cite{chia2023impossibilityFF} implies a fundamental barriers to reducing circuit depth for QSVT under cryptographic assumptions. These limitations highlight the difficulty of translating the QSVT-based quantum advantages into practical quantum applications and motivate the central question of this work:
\begin{center}
    \emph{Can we use the QSVT framework to enable end-to-end quantum applications within the resource constraints of EFTQCs?}
\end{center}

To address this question, we begin with a simple observation: many existing algorithmic tools are designed to implement operators such as $e^{-iHt}$ or $f(H)$ for arbitrary input states. Their resource requirements are therefore governed primarily by worst-case parameters associated with $H$ and $f$. In many applications, however, implementing the full operator is more than what is needed: the goal is simply to estimate quantities such as $\bra{\psi}f(H)^{\dagger}Of(H)\ket{\psi}$ for a particular initial state $\ket{\psi}$ and observable $O$. Using structure in the input state and target observables may therefore yield additional resource savings. Indeed, recent pioneer works on state-aware quantum linear systems algorithms~\cite{li2026newquantumlinearalgorithm,dalzell2026faster} and shadow Hamiltonian simulation~\cite{somma2025shadow,chakraborty2026efficientalgorithmapproximateshadow} have demonstrated the potential of such instance-specific information to improve existing quantum algorithms. These observations motivate the central idea of this work:
\begin{center}
    \emph{Think more about the input state structure, and compute only what is needed from the output.}
\end{center}

\subsection{Main result}
\label{sec:mainresult}
Guided by the central idea above, we formalize the shadow quantum singular value transformation (S-QSVT) problem, with the goal of developing state- and observable-aware algorithms with small quantum circuit depth.

\begin{definition}[Shadow Quantum Singular Value Transformation (S-QSVT)]
    Given $n$ qubit quantum system, Hamiltonian $H \in \mathbb{C}^{2^n \times 2^n}$, an initial quantum state $\ket{\psi}\in\mathbb{C}^{2^n}$, a polynomial $f(x) = \sum_{k=0}^{d} a_k x^k \in \mathbb{C}[x]$, and a set of observables $\mathcal{O}=\{O_1,O_2,\dots,O_m\}$, shadow QSVT asks to estimate $\langle O_j \rangle := \bra{\psi} f(H)^\dagger O_j f(H) \ket{\psi}$ for all $j\in [m]$.
\end{definition}

In this work, we show three algorithms for S-QSVT with circuit depth strictly smaller than the dgree of $f$ using the structure in the initial state and the observables. The first algorithm is a state-aware QSVT algorithm based on Krylov subspace. Krylov subspace is defined as 
\begin{align*}
\mathcal{K}_r(H, |\psi\rangle) := \operatorname{span}\{|\psi \rangle, H|\psi\rangle, \dots H^{r-1}|\psi\rangle\}.    
\end{align*}
We call a Krylov subspace $H$-invariant if $H^r \ket{v}\in \mathcal{K}_r(H, \ket{v})$. 

\begin{theorem}[State-aware QSVT (informal), \Cref{thm:main_alg1_formal_YC}]
\label{thm:main_alg1_informal}
Given block-encoding unitary $U_H$ of $H$, an initial state $\ket{\psi}$, and a polynomial $f$, suppose that $\mathcal{K}_r(H,\ket{\psi})$ is $H$-invariant and thus there exists an ($r-1$)-degree polynomial $p$ satisfying
$p(H)\ket{\psi}=f(H)\ket{\psi}\,\mbox{and}\,
    \frac{\max_{x\in[-1,1]}|p(x)|}{\|f(H)\ket{\psi}\|_2}$ is bounded.
Then \Cref{alg:state-aware-QSVT_YC} prepares a state within
$\epsilon$ of $\frac{f(H)\ket{\psi}}{\|f(H)\ket{\psi}\|_2}$ in the Euclidean norm, using quantum circuits of depth $O(rD_H)$, where $D_H$ is the circuit depth for $U_H$. The number of copies of $\ket{\psi}$ required, and the classical
and quantum computational complexities, are bounded by
$\poly(r,\deg(f),1/\epsilon, C_H)$, where $C_H$ is the number of gates for $U_H$.
\end{theorem}

\begin{remark}
The algorithm underlying  \Cref{thm:main_alg1_informal}
addresses an implementation challenge raised by
Adhikari~\cite{adhikari2025krylovpolynomialsquantumquery}.
Adhikari showed how polynomial approximation with respect to the
spectral measure induced by $(H,\ket{\psi})$ can exploit input-state
structure, but left efficient implementation as an open question.
Our algorithm provides a constructive realization of this
state-aware approach when $\mathcal{K}_r(H,\ket{\psi})$ is
$H$-invariant, preparing the normalized target state with quantum
circuit depth $O(rD_H)$ without prior knowledge of the spectral measure.    
\end{remark}

Although \Cref{alg:state-aware-QSVT_YC} reduces quantum circuit
depth, it still synthesizes the target state and did not implement the second part of our central idea. There are several questions remain:
\begin{enumerate}
    \item Although the $H$-invariance of
    $\mathcal{K}_r(H,\ket{\psi})$ guarantees the existence of a
    low-degree polynomial $p$ satisfying
    $p(H)\ket{\psi}=f(H)\ket{\psi}$, the ratio $\frac{\max_{x\in[-1,1]}|p(x)|}
             {\|f(H)\ket{\psi}\|_2}$
    can still be large. Can we avoid the resulting QSVT
    normalization overhead?
    \item Solving S-QSVT via \cref{thm:main_alg1_formal_YC} requires repeatedly preparing the target state. This repeated state preparation can lead to substantial overhead, particularly when the number of observables is large. Can we estimate these quantities efficiently without repeatedly preparing the transformed state? 
    \item Our algorithm achieves small quantum circuit depth by using the $H$-invariance of $\mathcal{K}_r(H,\ket{\psi})$. However, this reduction depends only on the Hamiltonian and the initial state. Can we further reduce the circuit depth by exploiting the structure of the target observables?
\end{enumerate}

Our second algorithm, classical shadow S-QSVT, addresses the first two concerns by directly preparing ``the classical shadows'' instead of synthesizing the target states.  

\begin{theorem}[Classical-Shadow S-QSVT (informal), \Cref{thm:end-to-end-observable-estimation}] Let $\ket{\psi}$, $H$, $f(\cdot)$, $\mathcal{O}$ be an S-QSVT instance, and $\mathcal{O}$ is not given in advance. Suppose that
$\mathcal{K}_r(H,\ket{\psi})$ is $H$-invariant. Then there exists an algorithm that can prepare a set of
reusable classical data independently of the observables to be queried
in depth $O(rD_H)$ such that one can approximate the observable estimates with classical postprocessing. Under standard conditioning and bounded-polynomial assumptions, the number of copies of $\ket{\psi}$ and the quantum computational complexity are bounded by $\poly(r,\deg(f),1/\epsilon, \nu_O^2, \log|\mathcal{O}|)$, while the classical computational complexity is bounded by $\poly\left(
r,\deg(f),1/\epsilon, \nu_O^2,|\mathcal O|
\right).$ Here $\nu_O^2$ denotes the shadow-norm bound for the queried observables.
\end{theorem}

The algorithm underlying \cref{thm:end-to-end-observable-estimation} is particularly well suited to settings in which the observables are not known in advance and their number is large. It first collects classical shadows using shallow quantum circuits, then uses these stored data to estimate expectation values on the target state through classical postprocessing once the observables
are specified. Subject to the theorem's shadow-norm assumptions, this approach can accommodate even superpolynomially many observables without repeating the quantum data-acquisition stage. Moreover, it eliminates
the assumption that the ratio $({\max_{x\in[-1,1]}|p(x)|}/{\|f(H)\ket{\psi}\|_2})$
is bounded by a constant, as required in
\Cref{thm:main_alg1_informal}.

Our final algorithm, observable-aware S-QSVT, further address the last question by using the structure of initial state and observables to reduce the quantum depth. Briefly, we use the following observation: Only directions visible to the given observables in $\mathcal{K}_r(H, \ket{v})$ are needed.


\begin{theorem}[Observable-Aware Shadow QSVT (informal), \Cref{thm:observable-aware-end-to-end}]
\label{thm:main-observable-aware-informal} Let $\ket{\psi}$, $H$, $f(\cdot)$, $\mathcal{O}$ be an S-QSVT instance. Suppose that $\mathcal K_r(H,\ket{\psi})$ is $H$-invariant and that the dimension of the subspace visible by $\mathcal{O}$ in $\mathcal{K}_r(H, \ket{\psi})$ is $s\leq r$. Then there exists an algorithm (\Cref{alg:observable-aware-qsvt}) that can approximate all observable estimates in $\mathcal{O}$ with quantum depth $O(sD_H)$. Under standard conditioning and bounded-polynomial assumptions, the number of copies of $\ket{\psi}$ and the quantum computational complexity are bounded by $\poly(s,\deg(f),1/\epsilon,\nu_O^2,\nu_\Omega^2,\log(|\mathcal{O}|))$, while the classical computational complexity is bounded by $\poly\left(s,\deg(f),1/\epsilon,\nu_O^2,\nu_\Omega^2,|\mathcal O|)\right).$ Here $\nu_O^2$ and $\nu_\Omega^2$ denote the shadow-norm bounds entering the observable- and history-matrix estimations, respectively.
\end{theorem}

Since $s\leq r$, \Cref{alg:observable-aware-qsvt} in
\Cref{thm:main-observable-aware-informal} can further reduce the quantum circuit depth when $s< r$, compared with the bound in \Cref{thm:main_alg1_informal}. 

The use of classical shadows is not intrinsic to the observable-aware reduction. If suitable block encodings of the observables are available, the required cross moments may instead be estimated directly using Hadamard-test-type procedures. Direct estimation can be preferable for a small number of observables, whereas classical shadows allow the same quantum data to be reused across many observables, at the cost of shadow-norm overhead. See \Cref{rem:observable-aware-direct-estimation} for a detailed discussion.

The algorithm underlying
\Cref{thm:main-observable-aware-informal} operates on the
quotient of the Krylov subspace $\mathcal{K}_r(H,\ket{v})$ by its
observable-invisible subspace and uses a history-state construction
to implement the resulting dynamics. This combination of observable-aware quotient spaces and history-state constructions may also be useful in other algorithms.

\begin{remark}[Common properties of our algorithms]
\label{remark:common_properties}
Our algorithms share the following properties:
\begin{itemize}
    \item \textbf{Reusability across target functions.}
    For a fixed Hamiltonian, initial state, and, where required,
    observables, our algorithms collect classical data that can be reused across
    different target functions. Provided that the new function
    satisfies the corresponding algorithm's assumptions and the
    stored data meet its accuracy requirements, only the classical
    postprocessing and, when applicable, the output-state
    preparation need to be updated.

    \item \textbf{Adaptive selection of the Krylov dimension.}
    Prior knowledge of $r$ is unnecessary, provided that a suitable
    criterion for detecting sufficient Krylov closure is available.
    One can start with a small guess and double it until this
    criterion is met. Since successive trials are separated by
    measurements and classical processing, their coherent circuit
    depths do not add. For a depth bound linear in the guessed
    dimension, doubling incurs at most a factor-of-two overhead
    in the maximum circuit depth relative to using the smallest
    sufficient dimension directly. The additional trials contribute
    to the total sampling and computational costs.

    \item \textbf{Approximate Krylov invariance.}
    Exact $H$-invariance of $\mathcal{K}_r(H,\ket{\psi})$ can be
    relaxed, provided that the resulting truncation error is
    controlled. Writing $\Pi_r$ for the orthogonal projector onto
    this subspace, approximate invariance can be quantified by
    the leakage $\|(I-\Pi_r)H\Pi_r\|$. The resulting error must be
    propagated through the target function $f$ and included in
    the overall error budget, together with any amplification
    caused by normalizing the output state. Thus, the required
    degree of approximate invariance depends on $f$ and the
    desired accuracy.
\end{itemize}
\end{remark}

\subsection{Technical Overview}

\subsubsection{State-aware QSVT}
The algorithm underlying \Cref{thm:main_alg1_informal}
starts from the observation that, if
$\mathcal{K}_r(H,\ket{\psi})$ is $H$-invariant, the action of
any polynomial $f(H)$ on $\ket{\psi}$ can be reproduced by a
polynomial of degree at most $r-1$, regardless of the degree of $f$.
Indeed, for every $k\geq r$, there exist coefficients
$\{\lambda_j^{(k)}\}_{j=0}^{r-1}$ such that $H^k\ket{\psi}
    = \sum_{j=0}^{r-1} \lambda_j^{(k)} H^j\ket{\psi}$.
Consequently, there exists a polynomial $p$ of degree at most $r-1$ satisfying

This observation suggests a route to preparing the target state using a lower-degree polynomial. The remaining questions are how to efficiently
construct a suitable polynomial and implement it using
shallow quantum circuits.

Our algorithm proceeds in three steps. 
\begin{enumerate}
    \item Use Hadamard tests
to estimate the moments of $H$ with respect to $\ket{\psi}$ and construct the moment matrix, $M_{ij} := \bra{\psi} H^i H^j\ket{\psi}$.
\item Use $M$ to compute
a low-degree polynomial $p$ that approximates the action of $f(H)$ on $\ket{\psi}$.
\item Use QSVT to construct a unitary block-encoding of  $p(H)$, and use it to prepare a state approximating the
normalized target state $f(H)\ket{\psi}/\|f(H)\ket{\psi}\|_2$.
\end{enumerate} 

Although this algorithm prepares the target state using circuits
whose depth scales with $r$ rather than $\deg(f)$, offering a
reduction when $r\ll\deg(f)$, three limitations remain.
First, it is unclear whether the depth dependence can be reduced
further, below linear in $r$.
Second, the ratio
\[
    \frac{\max_{x\in[-1,1]}|p(x)|}{\|f(H)\ket{\psi}\|_2}
\]
may be large, making QSVT-based state preparation infeasible.
Third, the state-preparation approach requires coherently preparing
the normalized state proportional to $p(H)\ket{\psi}$ for every
measurement shot, even when the final task is only to estimate
observables.

\subsubsection{Classical-Shadow S-QSVT}

We first remove the need to prepare the transformed state while retaining
the full $r$-dimensional state-Krylov representation. Let
\[
    A_r
    :=
    \begin{bmatrix}
        \ket{\psi} &
        H\ket{\psi} &
        \cdots &
        H^{r-1}\ket{\psi}
    \end{bmatrix},
\]
and define the observable Krylov matrices
\[
    B_{r,i}
    :=
    A_r^\dagger O_i A_r,
    \qquad
    (B_{r,i})_{a,b}
    =
    \bra{\psi}H^aO_iH^b\ket{\psi}.
\]
If $\boldsymbol{\lambda}$ denotes the Krylov coefficient vector
representing the action of $f(H)$ on $\ket{\psi}$, then the target
observable values can be written as
\[
    y_i
    =
    \boldsymbol{\lambda}^\dagger
    B_{r,i}
    \boldsymbol{\lambda}.
\]
Thus the transformed state need not be prepared explicitly.

Our algorithm proceeds in three steps.
\begin{enumerate}
    \item Use the state-aware Krylov procedure to estimate
    $\boldsymbol{\lambda}$ without performing the final QSVT
    state-preparation step.

    \item For each upper-triangular Krylov pair $(a,b)$, collect a
    reusable classical-shadow dataset $\mathcal D_{a,b}$. This quantum
    data-acquisition stage depends only on $H$, $\ket{\psi}$, and the
    Krylov indices, and is independent of the identities of the
    observables.

    \item Once an observable family
    $\mathcal O=\{O_1,\ldots,O_M\}$ is specified, reuse the stored
    snapshots to estimate all entries of $B_{r,i}$ and output
    \[
        y_i
        =
        \boldsymbol{\lambda}^\dagger
        B_{r,i}
        \boldsymbol{\lambda},
        \qquad i\in[M],
    \]
    using only classical postprocessing.
\end{enumerate}

For shadow-norm bound $\nu_O^2$, entrywise precision $\epsilon_S$, and
failure probability $\delta_B$, reconstructing all observable Krylov
matrices requires
\[
    O\!\left(
        \frac{r^2\nu_O^2}{\epsilon_S^2}
        \log\!\frac{Mr^2}{\delta_B}
    \right)
\]
shadow samples. The maximum quantum circuit depth is
$O(rD_H+D_{\mathrm{sh}})$. Hence Classical-Shadow S-QSVT removes the
need for repeated preparation of the transformed state and allows the
same quantum data to be reused for observables specified after the
quantum computation.

\subsubsection{Observable-Aware S-QSVT}

Classical-Shadow S-QSVT avoids preparing the transformed state, but it
still retains the full $r$-dimensional state-Krylov representation.
When the observable family is known in advance, we can reduce this
representation further by discarding Krylov directions that can never
affect the prescribed observables. Let $\mathcal O:=\{O_1,\ldots,O_M\}$ and define the
observable-invisible subspace
\[
    K^{\mathcal O}
    :=
    \left\{
        \ket{\phi}\in\mathcal K_r :
        O_iH^\ell\ket{\phi}=0
        \text{ for all }i\in[M]\text{ and }\ell\geq0
    \right\}.
\]
The observable-relevant dynamics are described by the quotient
\[
    \mathcal Q_{\mathcal O}
    :=
    \mathcal K_r/K^{\mathcal O},
    \qquad
    s:=\dim\mathcal Q_{\mathcal O}\leq r.
\]
Thus $r$ counts all state-Krylov directions, whereas $s$ counts only
those that can affect the prescribed observables. The first $s$
quotient Krylov vectors form a basis of $\mathcal Q_{\mathcal O}$, so
every polynomial transformation admits an observable-equivalent
representative of degree at most $s-1$. Because the quotient space is abstract, we realize it through the
finite history map
\[
    F_s^{\mathcal O}(\ket{\phi})
    :=
    \sum_{i=1}^{M}
    \sum_{\ell=0}^{s-1}
        \ket{i}\ket{\ell}
        \otimes
        O_iH^\ell\ket{\phi}.
\]
Its kernel is exactly $K^{\mathcal O}$, so its image gives a concrete
$s$-dimensional representation of the observable-relevant dynamics.

The algorithm then proceeds in three steps.
\begin{enumerate}
    \item Estimate the history Gram matrices
    $G_s^{\mathcal O}$ and $\widehat G_s^{\mathcal O}$ from Krylov
    cross moments involving
    $\Omega_{\mathcal O}:=\sum_{i=1}^{M}O_i^2$.

    \item Reconstruct the induced Hamiltonian action
    \[
        J_s^{\mathcal O}
        =
        (G_s^{\mathcal O})^{-1}
        \widehat G_s^{\mathcal O},
    \]
    and evaluate the polynomial action classically through
    \[
        \boldsymbol{\lambda}^{\mathcal O}
        =
        P(J_s^{\mathcal O})\mathbf e_0.
    \]

    \item Estimate the reduced observable matrices
    \[
        (B_{s,i})_{a,b}
        =
        \bra{\psi}H^aO_iH^b\ket{\psi},
        \qquad
        a,b=0,\ldots,s-1,
    \]
    and output
    \[
        y_i
        =
        (\boldsymbol{\lambda}^{\mathcal O})^\dagger
        B_{s,i}
        \boldsymbol{\lambda}^{\mathcal O}.
    \]
\end{enumerate}

All quantum measurements now involve Krylov powers of order $O(s)$,
giving maximum circuit depth $O(sD_H+D_{\mathrm{sh}})$. Thus
Observable-Aware S-QSVT further replaces the state-Krylov dimension
$r$ by the observable-relevant dimension $s\leq r$, while still
avoiding coherent preparation of the polynomially transformed state.
The additional reduction comes at the cost of requiring the observable
family during the quantum data-acquisition stage.

\subsection{Discussion} 
The S-QSVT framework offers a promising route to reducing
quantum resource requirements when the input state is specified
and only observable estimates are needed. Our three algorithms
demonstrate this potential by exploiting instance-specific
structure to reduce quantum circuit depth. They also offer
shared and complementary features that may facilitate practical
applications (see \Cref{remark:common_properties}). Nevertheless, several questions
worth further exploration concerning the capabilities and limitations of these
algorithms and, more broadly, the S-QSVT framework.

\paragraph{Potential quantum advantage}
Our algorithms use quantum measurements to obtain Krylov moment data for subsequent classical processing. Such data need not be efficiently accessible classically merely because $H$ is sparse. In particular, Janzing and Wocjan~\cite{janzing2007simple} showed that estimating $\langle j|H^k|j\rangle$ for an efficiently accessible sparse real symmetric matrix $H \in \mathbb{R}^{N \times N}$ with $\|H\| \leq 1$ is PromiseBQP-complete when $k = \poly\log(N)$ and the required additive precision is $1/\poly\log(N)$. This result establishes the hardness of general sparse-matrix moment estimation, but does not directly establish quantum advantage under the additional structural assumptions of our algorithms. Such an advantage would require instances for which the relevant moment data can be acquired and processed efficiently on a quantum computer, while the target quantities remain difficult to estimate classically under comparable input-access assumptions.

Conversely, our framework could also admit efficient classical implementations in certain regimes. For instance, if $H$ and $\ket{\psi}$ have efficiently accessible sparse representations with sparsity polynomial in $n$, and the Krylov dimension $r$ is constant, the required Krylov vectors can be computed classically using sparse matrix--vector products. Provided that the relevant observables also admit efficient classical access, the moment matrices and subsequent reduced-matrix computations can then be evaluated classically. The final coherent state-preparation step of Algorithm~1 remains a separate quantum task.

\paragraph{Applications}
An important direction is to identify practical applications of
our algorithms. One promising example is state-dependent parallel
fast-forwarding of Hamiltonian simulation, where the goal is to
prepare $e^{-iHt}\ket{\psi}$ with quantum circuit depth $o(t)$
for a specified initial state $\ket{\psi}$. Although general
parallel fast-forwarding is impossible under standard cryptographic
assumptions~\cite{chia2023impossibilityFF}, fast-forwarding is
possible for certain structured Hamiltonians~\cite{gu2021fast}.
Our framework suggests a complementary route based on structure
in the input state and target observables: preparing the evolved
state, or estimating its observable statistics, using shallow
quantum circuits under suitable instance-dependent assumptions.
Identifying physically relevant instances that satisfy these
assumptions while keeping the total resource costs manageable
is an important next step.

\paragraph{Reducing number of gates and copies}
Our current gate- and sample-complexity bounds still depend
polynomially on $\deg(f)$. Thus, our algorithms primarily provide an approach of reducing the circuit depth of QSVT-based algorithms when the instances satisfy the required structural promises. Further reductions in number of gate and copies of the state remain an important open direction. For example, efficient access to sufficiently accurate moment matrices could yield additional
resource savings, but such an assumption would shift a major
computational burden into the input model. There seems to be an inherent barrier to control the precision required for moment estimation and the subsequent propagation of statistical errors. Achieving
stronger end-to-end improvements without relying on additional
input assumptions may therefore require new algorithmic ideas.

\subsection{Disclosure of AI usage}
The main ideas underlying the algorithm design and problem formulations originated primarily with the authors. The authors used OpenAI's
ChatGPT 5.6 Sol and 6 Astra to assist
with checking intermediate arguments and verifying and refining the error analysis. Also, these tools helped identify the applicability of the Cayley--Hamilton theorem to the analysis of \Cref{alg:krylov-shadow-sampling}.
Anthropic's Claude 5.1 Fable was used to explore references relevant to \Cref{alg:state-aware-QSVT_YC}.
The final proofs and calculations were written and independently verified by the authors, who take full responsibility for the manuscript.
\section{Preliminary}

\subsection{Krylov subspace and Lanczos basis}

The basic structure of our algorithms is approximating smaller subspace where our target state or observable is. 
Since given $f$ is polynomial of Hermit matrix $H$, it is natural to represent such a subspace with span of lower power of $H^k$. This space is called by \textit{Krylov subspace}.

\begin{definition}[Krylov subspace]\label{def:Krylov}
    For $n,r\in\mathbb{N}$ satisfying $r\le n$, a matrix $H\in\mathbb{C}^{n\times n}$, and a vector $\ket{v}\in\mathbb{C}^{n}$, the $r$-dimensional Krylov subspace of $(H,\ket{v})$ is defined by
    \begin{equation}\label{eq:Krylov_subspace_def}
        \mathcal{K}_r(H, \ket{v}) := \spanv\{\ket{v}, H\ket{v}, \dots, H^{r-1}\ket{v}\}.    
    \end{equation}
    when $\{H^j\ket{\psi}\}_{j=0}^{r-1}$ is a linear independent set.
\end{definition}

\begin{definition}[$H$-invariant Krylov subspace]\label{def:H_invariant_Krylov_subspace}
    For $n,r\in\mathbb{N}$ satisfying $r\le n$, a matrix $H\in\mathbb{C}^{n\times n}$, and a vector $\ket{v}\in\mathbb{C}^{n}$, let $\mathcal{K}_r(H, \ket{v})$ be the $r$-dimensional Krylov subspace of $(H,\ket{v})$.
    We say $\mathcal{K}_r(H, \ket{v})$ is $H$-invariant if $H^r \ket{v}\in \mathcal{K}_r(H, \ket{v})$. 
\end{definition}

For any $n \times n$ matrix, we can get a orthonormal basis from Lanczos algorithm.
\begin{lemma}[Lanczos basis \cite{Lanczos50}]\label{lem:Lanczos}
    There exists an algorithm such that for any $n, r'\in\mathbb{N}$ satisfying $r'\le n$, for any Hermitian matrix $H\in\mathbb{C}^{n\times n}$, and for any vector $\ket{v}\in\mathbb{C}^{n}$, the algorithm runs in $\poly(n)$ time and outputs a set of orthonormal vectors $\{\ket{v_j}\in\mathbb{C}^n\}_{j=0}^{r-1}$ where $r\le r'$ such that $\ket{v_0}=\ket{v}$ and for all $j=0,\dots, r-2$ it holds that
    \begin{equation}\label{eq:Lanczos_propagation}
        H\ket{v_j} = \beta_j\ket{v_{j-1}} + \alpha_j\ket{v_j} + \beta_{j+1}\ket{v_{j+1}},
    \end{equation}
    where $\alpha_j, \beta_j\in\mathbb{C}$ and $\beta_0=0$.
\end{lemma}

We call the set $\{\ket{v_j}\}_{j=0}^{r-1}$ in \Cref{lem:Lanczos} $H$'s Lanczos basis. 
For a $n\times n$ matrix $H$, we can define a $r\times n$ matrix $V_r$ from $H$'s Lanczos basis.
The $j$-th column of $V_r$ is $\ket{v_j}$ for all $j=0,\dots, r-1$. That is, 
\begin{equation}\label{eq:Vr}
    V_r:=\sum_{j=0}^{r-1}\ketbra{v_j}{e_j}=\begin{bmatrix}
        \ket{v_0}\vert\cdots\vert\ket{v_{r-1}}
    \end{bmatrix},
\end{equation}
where $\ket{e_j}$ is the $j$-th standard basis vector and $\bra{e_j}$ is its transpose.

Then, we define a $r\times r$ matrix $J_m$ follows.
\begin{equation}\label{eq:Jr}
        J_r := V_r^\dagger H V_r =\begin{bmatrix}
        \alpha_0& \beta_1 &          &        &      &\\
        \beta_1 & \alpha_1& \beta_2  &        &      &\\
                & \beta_2 & \alpha_2 & \ddots &      &\\
                &         & \ddots   & \ddots &      \beta_{r-1}&\\
                &         &          & \beta_{r-1} & \alpha_{r-1}&\\
    \end{bmatrix}.
\end{equation}
It holds that
\begin{equation}\label{eq:Mat_version_Lanczos_prop}
         HV_r = V_r J_r.
\end{equation}
In addition, we can define $n\times r$ matrix $A_r$ of which the $j$-th column is $H^j \ket{\psi}$ as
\begin{equation}
    A_r: = \sum_{j=0}^{r-1} H^j \ket{\psi} \bra{e_j} = \begin{bmatrix}
        \ket{\psi} & H\ket{\psi} & H^2\ket{\psi}, \cdots H^{r-1}\ket{\psi}
    \end{bmatrix}.
\end{equation}
With this matrix, we define a $r\times r$ gram matrix $G_r$ and $\hat{G}_r$ follows.
\begin{equation}\label{eq:gram_gr_def}
    G_r := A_r^\dagger A_r, \quad (G_r)_{ij} = \bra{\psi} H^{i+j} \ket{\psi}
\end{equation}
\begin{equation}
    \hat{G}_r := A_r^\dagger H A_r, \quad (\hat{G}_r)_{ij} = \bra{\psi} H^{i+j+1} \ket{\psi}
\end{equation}

\subsection{Block-encoding}  

\begin{definition}[Block-encoding \cite{gilyen2018QSingValTransf}]
\label{def:block_encoding}
Let $H$ be an operator acting on $s$ qubits, and let $U_H$ be a unitary acting on $s+a$ qubits. For $\alpha>0$ and $\epsilon\geq0$, we say that $U_H$ is an $(\alpha,a,\epsilon)$-block-encoding of $H$ if
\begin{equation}
    \left\|H- \alpha(\langle 0|^{\otimes a}\otimes\mathbbm{1})U_H(|0\rangle^{\otimes a}\otimes\mathbbm{1})\right\|_s \leq\epsilon.
\end{equation}
where $\|\cdot\|_s$ denotes the spectral norm and $\alpha$ is the normalization factor. For an exact block-encoding ($\epsilon=0$), necessarily $\alpha\geq\|H\|_s$, and $U_H$ can be written as
\begin{equation}
    U_H =\begin{bmatrix}H/\alpha & * \\ * & *\end{bmatrix},
\end{equation}
with respect to the ancilla subspace $|0\rangle^{\otimes a}$. We abbreviate $(\alpha,a,\epsilon)$-block-encoding as $(\alpha,a,\epsilon)$-BE.
\end{definition}

Through out this paper, we assume that $\|H\|_s \leq 1$ and it is given by a unitary $U_H$ that is a $(1, l, \epsilon)$-BE of $H$.
\begin{definition}[State preparation\cite{gilyen2018QSingValTransf}]
    For a given $\vec{y} \in \mathbb{C}^m$ satysfying $\| \vec{y} \|_1 \leq \beta$, we say a pair of $b$-qubit unitary $(P_L, P_R)$ is a $(\beta, b, \epsilon)$-state preparation of $\vec{y}$ where $\epsilon\in\mathbb{R}$ and $b\ge\log m$if $P_L \ket{0}^{\otimes b} = \sum_{j=0}^{2^b -1} c_j \ket{j}$ and $P_R \ket{0}^{\otimes b} = \sum_{j=0}^{2^b -1} d_j |j\rangle$ such that $\sum_{j=0}^{m-1} \left| c_j^\ast d_j - \frac{y_j}{\beta} \right| \leq \epsilon$ and $\left| c_j^\ast d_j - \frac{y_j}{\beta} \right| =0$ for $j \in m, \dots 2^b-1$.
\end{definition}

The following two lemmas show the composition of block-encodings. 

\begin{lemma}[Linear combination of block-encoded matrices, Lemma~52 of \cite{gilyen2018QSingValTransf}]

    Let $A = \sum_{j=1}^m y_j A_j$ be a $s$-qubit quantum operator. 
    Let $(P_L, P_R)$ be a $(\beta, b, \epsilon_s)$-state preparation of $\vec{y}=(y_1, \dots, y_m)$ and $U_j$ be a $(\alpha, a, \epsilon_A)$-block encoding of $A_j$ for all $j=1,\dots,m$.
    Define $W: = \sum_{j=1}^m \ketbra{j}{j} \otimes U_j+\sum_{j=m+1}^{2^b-1}\ketbra{j}{j}\otimes\mathbbm{1}_{a+s}$ that is a $(s+a+b)$-qubit unitary. 
    We can implement a $(\alpha \beta, a+b, \alpha \epsilon_s + \alpha \beta \epsilon_A)$-block-encoding of $A$ by a single use of $P_L$, $P_R$, and $W$.
    \label{thm:linear_c_be}
\end{lemma}

\begin{lemma}[Product of block-encoded matrices, Lemma~53 of \cite{gilyen2018QSingValTransf}]
    \label{thm:product_be}
    For two $s$-qubit operators $A$ and $B$, let $U_A$ be a $(\alpha, a, \epsilon_A)$-BE of $A$ and $U_B$ be a $(\beta, b, \epsilon_B)$-BE of B.
    Define the $(s+a+b)$-qubit operator $U_{AB}:= (\mathbbm{1}_b \otimes U_A) (\mathbbm{1}_a \otimes U_B)$.

    It holds that $U_{AB}$ is a $(\alpha \beta, a+b, \alpha \epsilon_B + \beta \epsilon_A)$-BE of $AB$.

\end{lemma}
Therefore, the circuit depth of $U_{AB}$ is  the same as the sum of the depth of $U_A$ and $U_B$.
Let $D_{M}$ denote the circuit depth of $U_{M}$ that is a block-encoding of $M$. 
It holds that
\begin{equation}
\label{eq:circuit_comp_be}
    D_{AB} \leq D_A + D_B.
\end{equation}

By \Cref{thm:linear_c_be} and \Cref{thm:product_be}, we can implement a block-encoding of polynomial of $H$ from $H$'s block-encoding.

\begin{lemma}
\label{lemma:poly_block_encoding}
    Let $P_r(H) = \sum_{j=0}^{r-1} \lambda_j H^j$ be a degree $r-1$ polynomial of the $s$-qubit operator $H$.
    Let $U_H$ be a $(1, l, \delta)$-BE of $H$ and $(P_L,P_R)$ be a $(\beta,  b, \epsilon_s)$-state preparation of $\vec{\lambda}=(\lambda_0,\dots,\lambda_{r-1})$ where $b\ge \log r$.
    If $\delta < \epsilon_B/(r-1)$, we can implement a $(\beta, (r-1)l +b, \epsilon_s + \beta \epsilon_B)$-BE of $P_r(H)$, denoted by $U_{P_r(H)}$ by $O(r^2)$ execution of control-$U_H$ and a single use of $P_L$ and $P_R$.
\end{lemma}
\begin{proof}
    Applying \Cref{thm:product_be} $j$ times, we have $U_H^j$ is a $(1, jl, j\delta)$-BE of $H^j$. 
    Then, define $W:=  \sum_{j=0}^{r-1} \ketbra{j}{j} \otimes U_H^{j}+\sum_{j=r}^{2^b-1}\ketbra{j}{j}\otimes\mathbbm{1}_{a+s}$. There are $O(r^2)$ execution of control-$U_H$ in $W$. 
    By \Cref{thm:linear_c_be}, we can implement $U_{P_r(H)}$ that is a $(\beta, (r-1)l +b, \epsilon_s + \beta (r-1)\delta)$-BE of $\sum_{j=0}^{r-1} \lambda_j H^j$ with a single use of $P_L$ ,$P_R$, and $W$.
    The number of execution of control-$U_H$ in $W$ is $O(r^2)$.
    Because $\delta < \epsilon_B/(r-1)$, we have $U_{P_r(H)}$ is a $(\beta, (r-1)l +b, \epsilon_s + \beta \epsilon_B)$-BE of $P_r(H)$.
\end{proof}

\subsection{Quantum Singular Value Transformation}
\begin{theorem}[Quantum Singular Value Transformation (QSVT),  Theorem 56 of \cite{gilyen2018QSingValTransf}] 
    \label{thm:QSVT}
    Let $\epsilon \in \mathbb{R}$ and $\epsilon >0$.
    Let $U_H$ be an $(1, l, \epsilon)$-BE of a Hermitian matrix $H$.
    Let $P_r \in \mathbb{C}[x]$ be a degree-$(r-1)$ polynomial satisfying that
    $|P_{r}(x)| \leq \frac{1}{4}$ for all $x \in [-1, 1]$.
    Then, for all $\delta>0$, there exists a quantum circuit $\tilde{U}$, which is an $(1, l+2, 4(r-1)\sqrt{\epsilon} + \delta)$-encoding of $P_{r}(H)$, and consists of $(r-1)$ applications of $U_H$ and $U_H^\dagger$ gates, a single application of controlled-$U_H$ and $O(lr)$ other one- and two-qubit gates. Moreover we can compute a description of such a circuit with a classical computer in time $O(\operatorname{{poly}(d, \log(1/\delta))}).$
\end{theorem}

Applying QSVT circuit to $\ket{\psi}$ yields 
\begin{equation}
    \ket{0}{P}_r(H)\ket{\psi} + \ket{0^\perp}\ket{other}.
\end{equation}
To obtain the $P_r(H)\ket{\psi}$ state, we have to obtain $\ket{0}$ from the first register. 
Define the success probability to measure $\ket{0}$ at first branch as $p_{QSVT}$, and let the implementation error as $\epsilon_{QSVT}$ 
\begin{equation}
    \|\tilde{P}_r(H)\ket{\psi} - P_r(H)\ket{\psi}\|_2 \leq \epsilon_{QSVT} 
\end{equation}
then we can find the upper and lower bound of the success probability.
\begin{equation}
    \max(0, {\|P_r(H)\ket{\psi}\|_2} - \epsilon_{QSVT}) \leq \sqrt{p_{QSVT}} \leq {\|P_r(H)\ket{\psi}\|_2} + \epsilon_{QSVT}
\end{equation}
If $\epsilon_{QSVT}$ is small enough, then 
\begin{equation}
    p_{QSVT} \approx \|P_r(H)\ket{\psi}\|_2^2.
\end{equation}

Let $\mathbbm p$ denote the expected number of repetitions required to successfully prepare the normalized state \(P_r(H)\ket{\psi}/\|P_r(H)\ket{\psi}\|_2\) via QSVT postselection. It is an inverse of $p_{QSVT}$ 
\begin{equation}\label{eq:bbm_p_definition}
    \mathbbm{p} := \frac{1}{p_{QSVT}}
\end{equation}
then, we can show that the post-selection sample complexity of QSVT method as 
\begin{equation}
    \label{eq:qsvt_success_prob}
    O\left(\ln\Big(\frac{1}{\delta}\Big)\mathbbm{p}\right) = O\left( \frac{1}{ \|P_r(H)\ket{\psi}\|_2^2} \right)
\end{equation}
to achieve success probability $1-\delta$ by Chernoff bound.

Let $D_H$ and $C_H$ be circuit depth and gate complexity of $(1, (r-1)l, \epsilon_B)-\operatorname{BE}$ for $\epsilon_B >0$. In QSVT implementation of the $r-1$ degree polynomial, $U_{\Phi}$ algorithm requires $r-1$ queries on $U_{H}$ or $U_H^\dagger$ and one controlled $U_H$.

Consequently, under the assumption of constant-overhead controlled access to the block-encoding, the circuit has depth 
\begin{equation}
    O(rD_H + r)
    \label{eq:qsvt_circuit_depth}
\end{equation}
and gate complexity is
\begin{equation}
    O(r C_H  + r).
\end{equation}
If $D_H$ and $C_H$ are bigger than $O(1)$, then we can simply write down them as $O(rD_H)$ and $O(rC_H)$.

\subsection{Hadamard test}

All algorithms in this paper requires to measure gram matrices. 
To estimate these matrices, the algorithms use Hadamard test.
The Hadamard test is a quantum procedure for estimating the real or imaginary part of the expectation value $\bra{\psi}U\ket{\psi}$ for a unitary operator $U$.
\begin{lemma}[Hadamard test]\label{lem:Hadamard_test}
For any $n$-qubit unitary operator $U$ and any $n$-qubit quantum state $\ket{\psi}$, consider the quantum circuit below and let 
    $p_0$ be the probability of getting the measurement outcome 0.
    
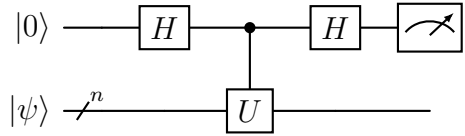
\begin{center}
        \begin{quantikz}
            \lstick{$\ket{0}$}&              & \gate{H}       & \ctrl{1}  & \gate{H} &\meter{} \\
            \lstick{$\ket{\psi}$}      & \qwbundle{n}&                & \gate{U}  &          &  
        \end{quantikz}
    \end{center}
    It holds that 
    \begin{equation}
        Re(\langle \psi | U | \psi \rangle) = 2 p_0 -1.
    \end{equation}
\end{lemma}
Considering the success probability and estimation error of the test, we can define a $(\epsilon, \delta)$-Hadamard test estimator follows.
\begin{definition}[$(\epsilon, \delta)$-Hadamard test estimator]
    \label{def:hdamard_estimator}
    Let $\epsilon, \delta \in \mathbb{R}_{+}$. We say an algorithm $\operatorname{HT}_{\epsilon, \delta}$ is an $(\epsilon,\delta)$-estimator if for any unitary operator $U$ and quantum state $\ket{\psi}$ have the same dimension, $\operatorname{HT}_{\epsilon, \delta}$ takes $(U, \ket{\psi})$ as an input and outputs  $\hat{\mu}\in\mathbb{R}$ such that
    \begin{equation*}
        \Pr[|\hat{\mu} - \mu| \leq \epsilon] \geq 1- \delta,
    \end{equation*}
    where $\mu: = \mbox{Re}(\bra{\psi}U\ket{\psi})$.
\end{definition}

\begin{lemma}
\label{lemma:HT_sample_complexity}
    An $(\epsilon,\delta)$-estimator $\operatorname{HT}_{\epsilon, \delta}(U, \ket{\psi})$ can be implemented by independently executing $O\left(\frac{1}{\epsilon^2} \log\left( \frac{2}{\delta}\right)\right)$ number of Hadamard test for $(U ,\ket{\psi})$.
\end{lemma}
\begin{proof}
Let $p_0$ be the probability of obtaining output $0$, $N$ be the number of execution, and $X_i\in\{0,1\}$ be the measurement outcome of the $i$-th execution of Hadamard test. 
Define $\hat{p}_0:=\frac{1}{N}\sum_i X_i$.

By Hoeffding inequality, it holds that
\begin{equation*}\label{eq:Hoeffding}
    \Pr[|\hat{p}_0 -p_0|\ge\frac{\epsilon}{2}]\le 2 e^{-N\epsilon^2/2}.
\end{equation*}
Taking $N=\frac{2}{\epsilon^2}\ln(\frac{2}{\delta})$, it holds that
\begin{equation*}\label{eq:Hoeffding_2}
    \Pr[|\hat{p}_0 -p_0|\ge\frac{\epsilon}{2}]\le \delta.
\end{equation*}
To implement $(\epsilon, \delta)$-Hadamard, we independently execute Hardmard test $N$ times on the input $(U, \ket{\psi})$ and get $\hat{p}_0$. Finally we output $2\hat{p}_0-1$, which  satisfies the requirement. 
\end{proof}

The Hadamard test estimates an expectation value of a unitary operator $U$ for given state $\ket{\psi}$. To estimate it, we therefore apply the test to $U_H$ block-encoding of $H$ to $\ket{0}^l\ket{\psi}$. 

\begin{lemma}[Hadamard-test error with block-encoding]
\label{lem:hadamard_be_error}
Let $\operatorname{HT}_{\epsilon,\delta}$ be an $(\epsilon,\delta)$-Hadamard test estimator. 
For any $n\in\mathbb{N}$, for any Hermitian operator $H\in\mathbb{C}^{n\times n}$, and quantum state $\ket{\psi}$, 
it holds that
\begin{equation*}
\Pr\left[
    \left|
    \operatorname{HT}_{\epsilon,\delta}(U_H,\ket{\psi})
    -\operatorname{Re}\langle\psi|H|\psi\rangle
    \right|
    \leq \epsilon+\epsilon_B
\right]
\geq 1-\delta,
\end{equation*}
where$U_H$ be a $(1,l,\epsilon_B)$-BE of $H$. 
\end{lemma}

\begin{proof}
Define $\tilde H:=(\bra{0^l}\otimes I)U_H(\ket{0^l}\otimes I)$.
Since $U_H$ is a $(1,l,\epsilon_B)$-block-encoding of $H$,
we have
\[
    \|\widetilde H-H\|\leq\epsilon_B.
\]
Therefore for any $H, \ket{\psi}$, we have 
\[
\begin{aligned}
\left|
\operatorname{Re}\langle\psi|\widetilde H|\psi\rangle -\operatorname{Re}\langle\psi|H|\psi\rangle
\right| \leq
\left|
\langle\psi|(\widetilde H-H)|\psi\rangle
\right| \leq
\|\widetilde H-H\|
\leq\epsilon_B.
\end{aligned}
\]
By the definition of
$\operatorname{HT}_{\epsilon,\delta}$,
with probability at least $1-\delta$,
\[
\left|
\operatorname{HT}_{\epsilon,\delta}(U_H,\ket{\psi})
-\operatorname{Re}\langle\psi|\widetilde H|\psi\rangle
\right|
\leq\epsilon.
\]
Hence, by the triangle inequality,
\[
\left|
\operatorname{HT}_{\epsilon,\delta}(U_H,\ket{\psi})
-\operatorname{Re}\langle\psi|H|\psi\rangle
\right|
\leq\epsilon+\epsilon_B,
\]
with probability at least $1-\delta$.
\end{proof}
Therefore, with $\epsilon_B$ block-encoding error, $\operatorname{HT}_{\epsilon, \delta}(U_H, \ket{0}\ket{\psi})$ becomes $(\epsilon+\epsilon_B, \delta)$-estimator of $\bra{\psi}H\ket{\psi}$.

\section{State-aware QSVT with shallow quantum circuit}\label{sec:section3}
In this section, we provide a shallow quantum circuit algorithm that takes a Hermitian matrix $H\in\mathbb{C}^{n \times n}$, a quantum state $\ket{\psi}\in\mathbb{C}^{n}$, and a polynomial function $f(\cdot)$  as inputs and outputs $f(H)\ket{\psi}/\|f(H)\ket{\psi}\|_2$ as long as the Krylov subspace $\mathcal{K}_r(H,\ket{\psi})$ is $H$-invariant.

\paragraph{Block-encoding normalization.}
Throughout this and the following sections, we state the algorithms using an exact $(1,l,0)$-block encoding. This is a normalization convention rather than a restriction. Given an exact $(\alpha,l,0)$-block encoding $U_H$ of $H$, define
\begin{equation}
    \overline H:=\frac{H}{\alpha},
    \qquad
    \overline f(x):=f(\alpha x).
\end{equation}
Then the same unitary $U_H$ is an exact $(1,l,0)$-block encoding of $\overline H$, and 
\begin{equation}
    \overline f(\overline H)=f(H).
\end{equation}
Moreover, $\mathcal K_r(\overline H,\ket{\psi}) =\mathcal K_r(H,\ket{\psi})$, so the Krylov dimension and the invariance assumption are unchanged. The degree of the target polynomial $f$ is also unchanged. The block-encoding query structure and circuit-depth scalings are therefore unchanged, while quantities depending on the normalization of the Krylov basis, such as Gram-matrix conditioning and function- or polynomial-dependent bounds, are evaluated for the rescaled instance and may depend on $\alpha$. For notational simplicity, we henceforth drop the bars and denote the rescaled pair $(\overline H,\overline f)$ again by $(H,f)$.

\subsection{State-Aware QSVT Algorithm}\label{sec:State-Aware QSVT-Algorithm}
We first sketch the idea of the algorithm. 
The goal is to prepare the quantum state $f(H)\ket{\psi}$. 
If the $r$-dimensional Krylov subspace $\mathcal{K}_{r}(H,\ket{\psi})$ is $H$-invariant, we have $f(H)\ket{\psi}\in\mathcal{K}_{r}(H,\ket{\psi})$. In other words,
\begin{equation}\label{eq:decompose_in_Krylov}
    f(H)\ket{\psi}=\sum_{i=0}^{r-1}\lambda_{i}H^i\ket{\psi}
\end{equation}
for some $\lambda_0, \dots, \lambda_{r-1}\in\mathbb{C}$.
The problem is reduced to implementing the QSVT of a $(r-1)$-degree polynomial $P_r(H):=\sum_{i=0}^{r-1}\lambda_{i} H^{i}$ on the input state $\ket{\psi}$. 

Then we explain how to find $\lambda_0,\dots, \lambda_{r-1}$ from $H$, $\ket{\psi}$, and $f(\cdot)$.
Let $\big\{\ket{v_j}\big\}_{j=0}^{r-1}$ be the Lanczos basis of $H$ with dimension $r$ satisfying $\ket{v_0}=\ket{\psi}$.
Because $\big\{\ket{v_j}\big\}_{j=0}^{r-1}$ and $\mathcal{K}_{r}(H,\ket{\psi})$ span the same vector space, we have
\begin{equation}\label{eq:decomposed_in_Lanczos}
    f(H)\ket{\psi}=\sum_{i=0}^{r-1}c_j\ket{v_j}
\end{equation}
for some $c_0,\dots,c_{r-1}$.
Define $\boldsymbol{\lambda}:=\begin{bmatrix}\lambda_{0} & \lambda_{1} & \cdots & \lambda_{r-1}\end{bmatrix}^T\in\mathbb{C}^{r}$, $\boldsymbol{c}:=\begin{bmatrix}c_{0} & c_{1} & \cdots & c_{r-1}\end{bmatrix}^T\in\mathbb{C}^r$,
\begin{equation}
    \label{eq:Ar_definition_YC}
    A_r := 
    \begin{bmatrix}
        \ket{\psi} & H\ket{\psi} & \cdots H^{r-1}\ket{\psi} 
    \end{bmatrix}\in\mathbb{C}^{n \times r},
\end{equation}
and
\begin{equation}
    \label{eq:Vr_definition_YC}
    V_r := 
    \begin{bmatrix}
        \ket{v_0} & \ket{v_1} & \cdots \ket{v_{r-1}} 
    \end{bmatrix}\in\mathbb{C}^{n \times r}.
\end{equation}
We can rewrite \Cref{eq:decompose_in_Krylov} as
\begin{equation}\label{eq:rewrite_Krylov}
    f(H)\ket{\psi}=A_r\boldsymbol{\lambda},
\end{equation}
and rewrite \Cref{eq:decomposed_in_Lanczos} as
\begin{equation}\label{eq:rewrite_Lanczos}
    f(H)\ket{\psi}=V_r\boldsymbol{c}.
\end{equation}
Let $U_r\in \mathbb{C}^{r \times r}$ be a matrix such that
\begin{equation}\label{eq:ArVrUr}
    A_r=V_r U_r.
\end{equation}
Because $H^k\ket{\psi}\in\spanv\{\ket{v_0}\dots,\ket{v_{k}}\}$ for any $k=0,\dots,r-1$, we have that $U_r$ is an upper-triangular matrix. 

Since $\big\{\ket{v_j}\big\}_{j=0}^{r-1}$ is an orthonormal basis, we have 
\begin{equation}\label{eq:identity_of_Vr}
    V_r^\dagger V_r=\mathbbm{1}_{r}
\end{equation}.
Combining \Cref{eq:rewrite_Krylov}, (\ref{eq:rewrite_Lanczos}), (\ref{eq:ArVrUr}), and (\ref{eq:identity_of_Vr}) , we have
\begin{equation}\label{eq:laambda_expreesed_by_c}
    \boldsymbol{\lambda}= U_r^{-1}\boldsymbol{c}.
\end{equation}

The remaining task is to compute the $r\times r$ matrix $U_r$ and the $r$-dimensional vector $\boldsymbol{c}$.
We first show how to find $U_r$.
Let $G\in\mathbb{C}^{r \times r}$ be the Gram matrix of the Krylov subspace $\mathcal{K}_r(H, \ket{\psi})$. That is,
\begin{equation}\label{eq:def_of_Gram}
    G_r:=A_r^\dagger A_r.
\end{equation}
Combining \Cref{eq:ArVrUr}, (\ref{eq:identity_of_Vr}), and (\ref{eq:def_of_Gram}), we have
\begin{equation}\label{eq:Cholesky_decomposition_of_G}
    G_r=U_r^\dagger U_r.
\end{equation}
Because $U_r$ is an upper-triangular matrix, \Cref{eq:Cholesky_decomposition_of_G} is a Cholesky decomposition of $G_r$. 
Hence, we can find $U_r$ by finding the Cholesky decomposition of the $r\times r$ matrix $G_r$.

We still need to find the description of $G_r$. 
From \Cref{eq:def_of_Gram}, 
each element of $G_r$ is given by
\begin{equation}\label{eq:Gij}
    (G_r)_{ij} = \bra{\psi}H^{i+j} \ket{\psi}.
\end{equation}
As a result, we can apply Hadamard test on the input $\ket{\psi}$ and $H^k$ for $k=0,\dots,2r-2$ to obtain the matrix element $G_{ij}$. 

Finally, we need to find $\boldsymbol{c}$.
Notice that $\ket{\psi}=V_r\boldsymbol{e}_0$ where $\boldsymbol{e}_0=\begin{bmatrix}1 & 0 &\cdots &0\end{bmatrix}^T\in\mathbb{C}^{r}$ is the first $r$-dimensional standard basis vector.
Combining \Cref{eq:decomposed_in_Lanczos}, we have
\begin{equation}\label{eq:c_expression}
    \boldsymbol{c}=V_r^\dagger f(H) V_r \boldsymbol{e}_0.
\end{equation}
Using $V_r^\dagger V_r=\mathbbm{1}_r$, we get 
\begin{equation}\label{eq:c_express_by_J}
    V_r^\dagger f(H) V_r =f(V_r^\dagger H V_r).
\end{equation}
From \Cref{eq:ArVrUr}, we get 
\begin{equation}\label{eq:c_express_by_J}
    V_r^\dagger H V_r=(U_r^{-1})^\dagger A_r^\dagger H A_r U_r.
\end{equation}
We define the $r \times r$ matrices 
\begin{equation}\label{eq:def_of_F}
    \hat{G}_r:=A_r^\dagger H A_r
\end{equation} 
and
\begin{equation}\label{eq:def_of_J}
    J_r:=V_r^\dagger H V_r = (U_r^{-1})^\dagger \hat{G}_r U_r^{-1}.
\end{equation}
Then, the vector $\boldsymbol{c}$ can be obtained from
\begin{equation}
    \boldsymbol{c}=f(J_r)\boldsymbol{e}_0.
\end{equation} 
The matrix element of $\hat{G}_r$ (and hence $J_r$) can be computed by 
\begin{equation}\label{eq:Fij}
    (\hat{G}_r)_{ij} = \bra{\psi}H^{i+j+1} \ket{\psi},
\end{equation}
which can be obtained by Hadamard test.

We present the algorithm as follows.

\begin{algorithm}[H]
    \caption{\textsc{StateAwareQSVT}}\label{alg:state-aware-QSVT_YC}
    \LinesNumbered
    \KwIn{
    $U_H$: $(1,l,\epsilon_B)$-block encoding of $H$;
    state-preparation access to $|\psi\rangle$;
    $f$: polynomial;
    $r$: the dimension of Krylov powers;
    $\epsilon$: precision;
    $\delta$: failure probability.
    }
    \KwOut{A quantum state $\ket{\eta}$ such that $\|\ket{\eta} - \frac{f(H) \ket{\psi}}{\|f(H) \ket{\psi}\|_2}\|_2\le \epsilon$ with probability $1-\delta$.}
    \BlankLine
    \textbf{Step 1: Calculate Gram matrix}\;
    Set $\tilde{g}_0:= 1$\;
    \For{$k=1,\ldots, 2r -1$}{
        $\tilde{g}_k \leftarrow \operatorname{HT}_{\epsilon_H, \delta_H}(U_{H^k}, \ket{\psi})$ \tcc*{Run Hadamard test on $(U,\ket{\psi})$}
    }
    \For{$i=0,\ldots, r -1$}{
        \For{$j=0,\ldots, r -1$}{
        $G_{ij} = \tilde{g}_{i+j}$\;
        $\hat{G}_{ij} = \tilde{g}_{i+j+1}$\;
        }
    }
    $\tilde{G}_r:=[G_{ij}]_{r\times r}$\;
    $\tilde{\hat{G}}_r:=[\hat{G}_{ij}]_{r\times r}$\;
    \textbf{Step 2: Calculate $\boldsymbol{\lambda}$ calculation}\;
    Calculate the Cholesky factor $U_r$, and the matrix $J_r$ by \Cref{alg:Ur_Jr_calculation} in \Cref{sec:subroutine}\;
    $\boldsymbol{c} := f(J_r)\boldsymbol{e}_0$\;
    Calculate $\boldsymbol{\lambda}$ by solving the linear equation $\boldsymbol{c} := U_r \boldsymbol{\lambda}$\;
    \textbf{Step 3: Final state preparation}\;
    \For{$t=1,\ldots,\mathbbm{p}$}{
    Implement the block encoding of $\sum_{i=0}^{r}\lambda_{i}H^i$ on the input state $\ket{\psi}$ by LCU/QSVT method with access to $U_H$\;
    Measure the first register\;
    \If{the measure measurement outcome is $\ket{0}$}
    {\textbf{return} the second register.
    }
    }
\end{algorithm}

\begin{theorem}[State-aware QSVT]\label{thm:main_alg1_formal_YC}
    For any $n, r\in\mathbb{N}$ satisfying $r\le n$, for any Hermitian $H\in\mathbb{C}^{n\times n}$ and quantum state $\ket{\psi}\in\mathbb{C}^{n}$ satisfying the $r$-dimensional Krylov subspace $\mathcal{K}_r(H, \ket{\psi})$ is $H$-invariant, for any $\epsilon, \delta>0$,  \Cref{alg:state-aware-QSVT_YC} outputs a quantum state $\ket{\eta}\in\mathbb{C}^{n}$ such that  $\left\|
    \ket{\mu}-
    \frac{f(H)\ket{\psi}}{\|f(H)\ket{\psi}\|_2}
    \right\|_2
    \leq \epsilon$ with probability at least $1-\delta$.
    
    The depth of the quantum circuit is $O(r D_H)$.
    The number of copies of input state $\ket{\psi}$ is $O \left(r^3 \frac{L_f^2 \log(r)^2}{\mu_{\min}^2} \frac{\ln ({r/\delta})}{\epsilon^2} + \mathbbm{p}\right)$, where $\mathbbm{p}=O(\ln(\frac{1}{\delta})\frac{1}{\|f(H)\ket{\psi}\|_2})$.
    The algorithm executes $(2r^2-r +1)$ application of control-$U_H$, and $r-1$ applications of $U_H$ and $U_H^\dagger$ respectively, where $U_H$ is a $(1,l,0)$-block encoding of $H$.
    The algorithm needs additional $O(r^2)$ classical computation.
\end{theorem}
In the following subsections, we will analyze the circuit depth, the gate complexity, and the error propagation at each step and combine them to prove \Cref{thm:main_alg1_formal_YC}. 
The first is the estimation of the matrices $G_r$ and $\hat{G}_r$. 
The second step is to compute the coefficient vector $\boldsymbol{c}$ and $\boldsymbol{\lambda}$. 
The final step is to prepare the target state with QSVT method.

For simplicity, we assume that the block-encoding of $H$ is exact, i.e., $U_H$ is a $(1,l,0)$-BE of $H$. 
We also assume that all classical computations are performed without numerical error. 
However, the Hadamard-test estimators introduce statistical estimation errors, which propagate through the subsequent computational steps. Therefore, in each subsection, we analyze both the computational complexity of the corresponding step and the propagation of errors arising from perturbed inputs.

\subsubsection{Step 1: Gram matrix estimation}\label{sec:Gram-matrix-estimation-complexity}

We are going to analyze the error propagation and the complexity of Step~1 in \Cref{alg:state-aware-QSVT_YC}. 
\begin{claim}\label{lemma:Hadamard_block_error_YC}
    Step~1 in \Cref{alg:state-aware-QSVT_YC} produces the estimation $\tilde{ G}_r$ and $\tilde{\hat{G}}_r$ satisfying
    \begin{equation}
        \|\tilde{G}_r - G_r \|_s \leq r \cdot \epsilon_{H},
    \end{equation}
    and
    \begin{equation}
        \|\tilde{\hat{G}}_r - \hat{G}_r \|_s \leq r \cdot \epsilon_{H}
    \end{equation}
    at the same time with probability at least $1-2r \delta_H$.
    The matrices $G_r$ and $\hat{G}_r$ are defined in \Cref{eq:def_of_Gram} and (\ref{eq:def_of_F}).

    The complexities of Step~1 are listed below.
    \begin{itemize}
        \item Sample complexity: Step~1 needs $O\left(r {\ln(1/\delta_H)}/{\epsilon_{H}^2}\right)$ number of copies of $\ket{\psi}$.
        \item Depth complexity: The quantum circuit depth in Step~1 is $O(rD_H)$, where $D_H$ is the circuit depth of controlled-$U_H$.
        \item Query complexity: Step~1 uses $O(r^2)$ applications of Control-$U_H$.
        \item Gate complexity:  The number of quantum gates in Step~1 is $O(r^2 C_H \ln(1/\delta_H)/\epsilon_H^2)$, where $C_H$ is the gate complexity of controlled-$U_H$.
    \end{itemize}
\end{claim}
\begin{proof}
    Since $H$ is Hermitian, it holds that $\bra{\psi} H^k \ket{\psi} \in \mathbb{R}$,
    and hence both $G_r$ and $\hat{G}_r$ are real matrices.
    The entries of $G_r$ and $\hat{G}_r$ are  $g_0,\dots,g_{2r-1}$ where $g_k = \bra{\psi} H^k \ket{\psi}$ for all $k=0,\dots,2r-1$.
    Let $\tilde{g}_k := \operatorname{HT}_{\epsilon_H, \delta_H}(U_{H^k}, \ket{0}^{kl}\ket{\psi})$ be an estimation of $g_k$ by Hadamard test.
    %
    By the definition of $(\epsilon, \delta)$-Haramard test estimator \Cref{def:hdamard_estimator}, and the union bound, we have the event $\|\tilde{g}_k - g_k \| \leq \epsilon_H
    $ for all $k=0,\dots, 2r-1$
    has probability at least $1-2r\delta_H$.
    Hence, with probability at least $1-2r\delta_H$, 
    \begin{equation*}
        |(\tilde{G}_r)_{ij} -(G_r)_{ij}| \leq \epsilon_H
    \end{equation*}
    and
    \begin{equation*}
        |(\tilde{\hat G}_r)_{ij} -(\hat{G}_r)_{ij}| \leq \epsilon_H
    \end{equation*}
    for all $i, j = 0, \dots r-1$.
    
    We then bound the error between $\tilde{G}_r$ and $G_r$ as follows.
    \begin{equation*}
        \|\tilde{G}_r -G_r\|_s \leq \|\tilde{G}_r -G_r\|_F =\sqrt{\sum_{i=1}^r \sum_{j=1}^r |(\tilde{G}_r -G_r)_{ij}|^2} \leq \sqrt{\sum_{i=1}^r \sum_{j=1}^r \epsilon_{H}^2} = \sqrt{r^2 \cdot \epsilon_{H}^2} = r \cdot \epsilon_{H}.
    \end{equation*}
    Similarly, 
    \begin{equation*}
        \|\tilde{\hat{G}}_r -G_r\|_s \leq \|\tilde{\hat{G}}_r -G_r\|_F =\sqrt{\sum_{i=1}^r \sum_{j=1}^r |(\tilde{\hat{G}}_r -\hat{G}_r)_{ij}|^2} \leq \sqrt{\sum_{i=1}^r \sum_{j=1}^r \epsilon_{H}^2} = \sqrt{r^2 \cdot \epsilon_{H}^2} = r \cdot \epsilon_{H}.
    \end{equation*}

    The number of Hadamard test in one execution of $\operatorname{HT}_{\epsilon_H, \delta_H}$ is $O \left( \frac{\ln(1/\delta_H)}{\epsilon_H^2}\right)$ by \Cref{lemma:HT_sample_complexity}. 
    Therefore, the total number of copies of $\ket{\psi}$ to estimate all $2r$ estimations is 
    \begin{equation*}
        O \left( r \frac{\ln(1/\delta_H)}{\epsilon_H^2}\right).
    \end{equation*}

    The circuit depth of controlled-$U_H$ is $D_H$.
    For the estimation of $g_k$,
    we can parallelly run every Hadamard test, and each Hadamard-test circuit uses $k$ applications of controlled-$U_H$.
    Therefore, the circuit depth of each $HT_{\epsilon_{H}, \delta_{H}}(U_{H^k},\ket{\psi})$ is $O(kD_H)$.
    We can run every $HT_{\epsilon_{H}, \delta_{H}}(U_{H^k},\ket{\psi})$ is $O(kD_H)$ parallelly. 
    As a result, the quantum circuit depth of Step~1 is 
    \begin{equation*}
        O(rD_H).
    \end{equation*}

    The number of applications of controlled-$U_H$ for evaluating all $2r$ estimators is $O(r^2)$.
    The total number of applications of controlled-$U_H$ is
    \begin{equation*}
        O\left(r^2\frac{\ln(1/\delta_H)}{\epsilon_H^2}\right).
    \end{equation*}

    Finally, the total gate complexity of Step~1
    \begin{equation*}
        O\left(r^2 C_H \frac{\ln(1/\delta_H)}{\epsilon_H^2}\right)
    \end{equation*}
    where $C_H$ is the gate complexity of controlled-$U_H$.
\end{proof}
\subsubsection{Step 2: Coefficient reconstruction}\label{sec:Coefficient-vector-calculation}

To show the correctness of Step~2, we first claim $U_r^{-1}$ exists.
It holds that the Gram matrix $G_r=A_r^\dagger A_r$ is positive due to the linear independency of $A_r$'s column vectors. 
Therefore, the Cholesky factor $U_r$ is non-singular, and hence $U_r^{-1}$ exists.

Then we claim that the uniqueness of the output $\boldsymbol{\lambda}$.
\begin{claim}
    The output vector $\boldsymbol{\lambda}$ in Step~2 is unique.
\end{claim}
\begin{proof}
    Because $G_r$ is positive definite, the Cholesky decomposition of $G_r$ is unique up to a diagonal unitary matrix.
    That is, if $U_r$ and $U'_r$ are upper-triangular satisfying $G_r=U_r^\dagger U_r$ and $G_r = {U_r^{\prime}}^{\dagger}$, then, there exist $\theta_0,\dots,\theta_{r-1}\in\mathbb{R}$ and $D:=\operatorname{diag}(e^{i\theta_0}, \dots, e^{i \theta_{r-1}})\in\mathbb{C}^{r \times r}$ such that $U'_r=DU_r$.
    A more detailed analysis is given in \Cref{section:proof_qr_factor_lanczos_equivalent}.

    The phase $D$ does not affect the result in Step~2.
    If the Cholesky factor we get in Step~2 is $U'_r:=DU_r$, we can absorb the pase into $V$.
    To be more clear, defining $V'_r:=V_rD$, \Cref{eq:ArVrUr} becomes $A_r=V'_rU'_r$, and we substitute $V_r$ with $V'_r$ and substitute $U_r$ with $U'_r$ in the analysis in \Cref{sec:State-Aware QSVT-Algorithm}. The analysis is unchanged. Hence, the output vector $\boldsymbol{\lambda}$ is unique. 
\end{proof}

Then we analyze the error in Step~2.
\begin{claim}\label{lemma:error_prop_alg1}
    Let $r\geq 2$, $\epsilon_g \in [0, 1]$, and let $\tilde{G}_r$ and
$\widetilde{\widehat{G}}_r$ be estimates of $G_r$ and $\hat{G}_r$, respectively, satisfying
        $\|\tilde{G}_r - G_r \|_s \leq \epsilon_{g}$
        and
        $\|\tilde{\hat{G}}_r - \hat{G}_r \|_s \leq \epsilon_{g}$
    Let $\tilde{P}_r(H):=\sum_{k=0}^{r-1}\lambda_{k}H^k$ be the polynomial prepared by \Cref{alg:state-aware-QSVT_YC}.
    Suppose  $\|J_r\|_s\leq 1$.
    It holds that
    \begin{equation}
        \|(\tilde{P}_r(H) - f(H))\ket{\psi}\|_2 \leq O\left( L_f \log(r) \frac{\epsilon_g}{\mu_{\min}} \right)
    \end{equation}
    where $L_f$ is a Lipschitz constant for $f$ on $x \in [\lambda_{\min}, \lambda_{\max}]$, and $\lambda_{\min}, \lambda_{\max}$ are minmum and maximum eigenvalues of $J_r$.
\end{claim}
\begin{proof}
Define
\begin{equation}
    c_r:=\frac{1}{2}+\lceil\log_2 r\rceil, \qquad \eta:=\frac{\epsilon_g}{\mu_{\min}}.
\end{equation}
By \Cref{lem:relative_cholesky_perturbation}, $\tilde{G}_r$ is positive definite, and its Cholesky factor satisfies
\begin{equation}
    \tilde{U}_r=(\mathbbm{1}+\Gamma)U_r, \qquad \gamma:=\|\Gamma\|_s\leq 2c_r\eta<\frac{1}{2}.
\end{equation}
Let $S:=(\mathbbm{1}+\Gamma)^{-1}$, the same lemma gives
\begin{equation}
    \|S\|_s\leq\frac{1}{1-\gamma}\leq 2.
\end{equation}

The coefficient reconstruction step determines
\begin{equation}
    \boldsymbol{c}=f(J_r)\ket{e_0}, \qquad \tilde{\boldsymbol{c}}=f(\tilde{J}_r)\ket{e_0},
\end{equation}
and
\begin{equation}
    U_r\boldsymbol{\lambda}=\boldsymbol{c}, \qquad \tilde{U}_r\tilde{\boldsymbol{\lambda}} =\tilde{\boldsymbol{c}}.
\end{equation}
Let $\delta\boldsymbol{c}:= \tilde{\boldsymbol{c}}-\boldsymbol{c}$ and $\delta\boldsymbol{\lambda}:= \tilde{\boldsymbol{\lambda}}-\boldsymbol{\lambda}$.
By \Cref{claim:alg1_J_r_error} and \Cref{claim:error_lanczos_coefficient},
\begin{equation}
    \|\delta\boldsymbol{c}\|_2
    \leq L_f\|\tilde{J}_r-J_r\|_s
    \leq 4L_f(4c_r+1)\eta.
\end{equation}

To relate this coefficient error to the reconstructed polynomial, observe that
\begin{equation}
    \begin{aligned}
        U_r\delta\boldsymbol{\lambda} &= S\tilde{U}_r (\tilde{\boldsymbol{\lambda}}-\boldsymbol{\lambda})\\
        &= S\left(\tilde{\boldsymbol{c}}-(\mathbbm{1}+\Gamma)\boldsymbol{c}\right)\\
        &= S\left(\delta\boldsymbol{c}-\Gamma\boldsymbol{c}\right).
    \end{aligned}
\end{equation}
Consequently,
\begin{equation}
    \|U_r\delta\boldsymbol{\lambda}\|_2 \leq 2\left( \|\delta\boldsymbol{c}\|_2 +\|\Gamma\boldsymbol{c}\|_2\right).
\end{equation}

It remains to bound $\|\Gamma\boldsymbol{c}\|_2$. 
Since $(G_r)_{00}=(\tilde{G}_r)_{00}=1$, $(U_r)_{00}$ which implies $\Gamma\ket{e_0}=0$.
$\|\boldsymbol{c}\|_2=1$ because it is a coefficient vector of orthonormal basis, we have
\begin{equation}
    \|\Gamma\boldsymbol{c}\|_2
    \leq \|\Gamma\|_s\|\boldsymbol{c}\|_2
    =\gamma
    \leq 2c_r\eta.
\end{equation}

Finally, $H$-invariance of $\mathcal{K}_r(H,\ket{\psi})$ implies $P_r(H)\ket{\psi}=f(H)\ket{\psi}$.
Using $G_r=A_r^\dagger A_r=U_r^\dagger U_r$, we obtain
\begin{equation}
    \begin{aligned}
        \|(\tilde{P}_r(H)-f(H))\ket{\psi}\|_2
        &= \|A_r\delta\boldsymbol{\lambda}\|_2\\
        &= \|U_r\delta\boldsymbol{\lambda}\|_2\\
        &\leq 2\bigl(4L_f(4c_r+1)+2L_fc_r\bigr)\eta\\
        &= 4L_f(9c_r+2)\frac{\epsilon_g}{\mu_{\min}}\\
        &= O\left(L_f\log_2(r)\frac{\epsilon_g}{\mu_{\min}}\right),
    \end{aligned}
\end{equation}
where the last step uses $c_r=O(\log_2 r)$ for $r\geq 2$.
\end{proof}

Then we analyze the running time in Step~2.
Instead of naive matrix calculation, using the relationship between each matrix element, we can directly calculate the $U_r$ and $J_r$ simultaneously. 
It is implemented in Algorithm \ref{alg:Ur_Jr_calculation} and \ref{alg:k_1_col_of_Ur}. The details are written in Appendix \ref{appendix:efficient_cal_ur_jr}.

\begin{claim}[Classical running time of Step~2]
\label{lemma:classical_complexity_alg1}
Let $G_r,\hat{G}_r\in\mathbb{R}^{r\times r}$
be the matrices used in \Cref{alg:state-aware-QSVT_YC}, with $G_r$ positive definite.
Assume fixed-precision arithmetic and $O(1)$ cost for each scalar evaluation of $f$.
Then step~2 of \Cref{alg:state-aware-QSVT_YC} can be implemented using $O(r^2)$ arithmetic operations.
\end{claim}
\begin{proof}
By \Cref{appendix:efficient_calculation_subroutine}, \Cref{alg:Ur_Jr_calculation} constructs $U_r$ and $J_r$ in $O(r^2)$ operations.

Since $J_r$ is real symmetric and tridiagonal,
its eigendecomposition $J_r=W_r\Lambda_rW_r^\dagger$ can be computed in $O(r^2)$ operations using the MRRR algorithm~\cite{antonelli2005mrrr}.
We then evaluate 
\begin{equation}
    \boldsymbol{c} =W_r\bigl(f(\Lambda_r) (W_r^\dagger\ket{e_0})\bigr).
\end{equation}
Extracting $W_r^\dagger\ket{e_0}$ and evaluating the diagonal entries of $f(\Lambda_r)$ require $O(r)$ operations under the stated assumptions.
The final matrix--vector multiplication requires $O(r^2)$ operations.

Finally, $U_r$ is nonsingular and upper triangular, so solving the linear equation $\boldsymbol{c}=U_r\boldsymbol{\lambda}$ by backward substitution
in $O(r^2)$ operations~\cite{golub2013matrix}.
Combining these costs gives the stated bound.
\end{proof}

\subsubsection{Sample complexity of Hadamard test}

\begin{lemma}
    \label{lemma:Sample_comp_alg1_Hadamard}
    Let $\epsilon,\delta\in(0,1)$ denote the target state-vector error tolerance and the allowed overall failure probability of \Cref{alg:state-aware-QSVT_YC}, respectively.
    For given inputs $H$, $\ket{\psi}$, and $f$, let
    $P_r$ and $\tilde{P}_r$ denote the polynomials obtained
    by \Cref{alg:state-aware-QSVT_YC} using exact and
    estimated moments, respectively.
    To ensure that
    \begin{equation}
        \|\tilde{P}_r(H)\ket{\psi}
        -P_r(H)\ket{\psi}\|_2 \leq \epsilon
    \end{equation}
    with probability at least $1-\delta/2$, it suffices to use
    \begin{equation}
        O\left(
            \frac{r^2 L_f^2 \log^2(r)}{\mu_{\min}^2}
            \frac{\ln(r/\delta)}{\epsilon^2}
        \right)
    \end{equation}
    copies of $\ket{\psi}$ per Hadamard estimator, where $L_f$ is a Lipschitz constant of $f(x)$ for $x \in[\lambda_{\min}, \lambda_{\max}]$, and $\lambda_{\min}, \lambda_{\max}$ are minimum and maximum eigenvalues of $J_r$, respectively.  
\end{lemma}
\begin{proof}
    From \Cref{lemma:error_prop_alg1}, if Step 1 of \Cref{alg:state-aware-QSVT_YC} estimates $G_r, \hat{G}_r$ with $\epsilon_g \in [0, 1]$ error such that
    \begin{equation}
        \label{eq:lemma_Sample_eg_bound}
        \epsilon_g \leq \frac{\mu_{\min}}{L_f \log_2(r)} \epsilon,
    \end{equation}
    then the following bound holds.
    \begin{equation}
        \| \tilde{P}_r(H)\ket{\psi} - {P}_r(H)\ket{\psi}\|_2 \leq \epsilon.
    \end{equation}
    By \Cref{lemma:Hadamard_block_error_YC}, the estimation error bounds imply $\epsilon_g\le r\epsilon_H$. Thus, to satisfy \Cref{eq:lemma_Sample_eg_bound}, it suffices to choose
    \begin{equation}
         \epsilon_H \leq \frac{\mu_{\min}}{r L_f \log_2(r)} \epsilon.
    \end{equation}
     Union bound guarantees that if 
    \begin{equation}
        \delta_H \leq  \frac{\delta}{2r},
    \end{equation}
    we can estimate all $2r$ moment estimates in $\epsilon_H$ precision with at least $1-\delta/2$ success probability.
    Finally, the required number of copies of $\ket{\psi}$ per each Hadamard test estimator is 
    \begin{equation}
        O\left( r^2 \frac{L_f^2\log_2^2(r)}{\mu_{\min}^2} \frac{\ln(r/\delta)}{\epsilon^2}\right).
    \end{equation}
\end{proof}

\subsubsection{Proof of \cref{thm:main_alg1_formal_YC}}
\begin{proof}[Proof of \cref{thm:main_alg1_formal_YC}]
    (Correctness):
    The exact reconstruction satisfies
    $P_r(H)\ket{\psi}=f(H)\ket{\psi}$.
    Choosing the estimation precision of $\epsilon_H$ and the failure probability of estimation $\delta_H$ as in \Cref{lemma:Sample_comp_alg1_Hadamard} with target error $\epsilon$ and failure probability of the whole algorithm $\delta$.
    \Cref{lemma:error_prop_alg1} guarantees that approximated state $\tilde{P}_r(H)\ket{\psi}$ by \cref{alg:state-aware-QSVT_YC} satisfies
    \begin{equation}
        \|(\tilde{P}_r(H)-f(H))\ket{\psi}\|_2
        \leq\epsilon
    \end{equation}
    with probability at least $1-\delta$. 
    QSVT therefore prepares the required state within error $\epsilon$ upon successful post-selection.

    The quantum parts of the algorithm consists of two part, Hadamard test and state preparation with QSVT so to calculate circuit depth and sample complexity we have to calculate them separately.
    
    (Circuit depth):
    \Cref{lemma:Hadamard_block_error_YC} states that the deepest circuit from Hadamard estimator is $O(r D_H)$ and QSVT of $r-1$-degree polynomial also has $O(r D_H )$ circuit depth to implement it from Eq.~\eqref{eq:qsvt_circuit_depth}, therefore the total circuit depth is
    \begin{equation}
        O(r D_H).
    \end{equation}

    (Application of gates):
    There are total $2r^2-r$ application of Control-$U_H$ on Hadamard estimator and in QSVT there are 1 Control-$U_H$, $r-1$ applications of $U_H$ or $U_H^\dagger$ and $O(lr)$ one or two qubit gates. Therefore, the total $2r^2-r+1$ Control-$U_H$ applications, $r-1$  applications of $U_H$ or $U_H^\dagger$ and $O(lr)$ one or two qubit gates applications are needed.  
    
    (Number of Copies of initial state):
    We have $2r$ number of $\operatorname{HT}_{\epsilon_H, \delta_H}$. From \Cref{lemma:Sample_comp_alg1_Hadamard}, the number of copies of $\ket{\psi}$ to estimate $\tilde{G}_r$ and $\tilde{\hat{G}}_r$ is  
    \begin{equation}
        O\left( r^3 \frac{L_f^2\log(r)^2}{\mu_{\min}^2} \frac{\log(r/\delta)}{\epsilon^2}\right).
    \end{equation}
    From \Cref{eq:qsvt_success_prob} to prepare the $P_r(H)\ket{\psi}$ with QSVT method with a success probability $1-\delta/2$, the number of copies of $\ket{\psi}$ is 
    \begin{equation}
        O \left( \ln\Big(\frac{1}{\delta}\Big)\frac{1}{ \|P_r(H)\ket{\psi}\|_2^2} \right).
    \end{equation}
    By union bound, \Cref{alg:state-aware-QSVT_YC} has success probability $1-\delta$.
    
    In QSVT preparation, the algorithm uses $\tilde{\boldsymbol{\lambda}}$ vector calculated on classic computer, therefore its initial state copies for QSVT is independent to the gram matrices estimation. 
     Finally, total number of copies of $\ket{\psi}$ is 
     \begin{equation}
        O\left( r^3 \frac{L_f^2\log(r)^2}{\mu_{\min}^2} \frac{\log(r/\delta)}{\epsilon^2} +  \mathbbm{p}\right)
     \end{equation}
    with $\mathbbm{p} :=  O \left( \ln\Big(\frac{1}{\delta}\Big)\frac{1}{ \|P_r(H)\ket{\psi}\|_2^2} \right)$.
    
    (Computational Complexity): Lastly, in Step 2. of \Cref{alg:state-aware-QSVT_YC} requires $O(r^2)$ computational complexity from \Cref{lemma:classical_complexity_alg1}.
    The phase estimation of $r-1$ degree polynomial for QSVT can be done in $O(r^2)$ \cite{ni2024fastphasefactorfinding}. Therefore, the total computational complexity is $O(r^2)$.
\end{proof}

\section{Classical-Shadow S-QSVT}

We now return to the original Shadow QSVT (S-QSVT) problem. The state-aware QSVT procedure ultimately implements the polynomial transformation $P_r(H)$ coherently in order to prepare the transformed state. If the goal is only to estimate observables of the form
\begin{equation}
    y_i
    :=
    \langle\psi|P_r(H)^\dagger O_i P_r(H)|\psi\rangle,
    \qquad i\in[M],
\end{equation}
this final-state-preparation step can be bypassed.

Let $\mathbf c$ denote the coefficient vector of the transformed state in the orthonormalized Krylov basis and let $\boldsymbol{\lambda}$ denote the corresponding coefficient vector in the original Krylov basis, so that $\mathbf c=U_r\boldsymbol{\lambda}$. Then
\begin{equation}
    y_i
    =
    \mathbf c^\dagger V_r^\dagger O_iV_r\mathbf c
    =
    \boldsymbol{\lambda}^\dagger U_r^\dagger V_r^\dagger O_iV_rU_r\boldsymbol{\lambda}.
\end{equation}
Define the reduced observable
\begin{equation}
    O_i^{(r)}:=V_r^\dagger O_iV_r
\end{equation}
and the observable Krylov matrix
\begin{equation}
    B_{r,i}
    :=
    A_r^\dagger O_iA_r
    =
    U_r^\dagger O_i^{(r)}U_r.
\end{equation}
Therefore,
\begin{equation}
    y_i
    =
    \boldsymbol{\lambda}^\dagger B_{r,i}\boldsymbol{\lambda}.
    \label{eq:observable-krylov-form}
\end{equation}
The entries of $B_{r,i}$ are the Krylov cross moments
\begin{equation}
    (B_{r,i})_{a,b}
    =
    \langle\psi|H^aO_iH^b|\psi\rangle.
\end{equation}
For convenience, define
\begin{equation}
    \mu_{i,a,b}
    :=
    \langle\psi|H^aO_iH^b|\psi\rangle.
    \label{eq:krylov-cross-moment}
\end{equation}
Thus observable estimation reduces to estimating the coefficient vector $\boldsymbol{\lambda}$ and the observable Krylov matrices $\{B_{r,i}\}_{i=1}^{M}$. The latter can be estimated directly from classical-shadow data without preparing $P_r(H)|\psi\rangle$.

\subsection{Classical-shadow access model}
Following the classical-shadow framework of
Huang, Kueng, and Preskill~\cite{huang2020predicting}, let $\mathcal U$ denote a classical-shadow measurement ensemble acting on the control-system register $C+S$, with measurement channel
\begin{equation}
    \mathcal M(\rho)
    :=
    \mathbb E_{U\sim\mathcal U}
    \sum_b
    \langle b|U\rho U^\dagger|b\rangle
    U^\dagger|b\rangle\langle b|U.
    \label{eq:shadow-channel}
\end{equation}
For a measurement unitary $U\sim\mathcal U$ and computational-basis outcome $b$, the corresponding classical snapshot is
\begin{equation}
    \widehat\rho
    :=
    \mathcal M^{-1}
    \left(U^\dagger|b\rangle\langle b|U\right),
    \label{eq:classical-snapshot}
\end{equation}
so that $\mathbb E[\widehat\rho]=\rho$. 
Throughout, we assume that the chosen shadow ensemble is efficiently samplable and that each snapshot can be stored using an efficient classical representation; in particular, we do not require $\widehat\rho$ to be explicitly formed as a dense matrix.

For an observable $Q$, define its shadow norm with respect to the
measurement ensemble $\mathcal U$ by
\begin{equation}
    \|Q\|_{\mathrm{shadow}}^2
    :=
    \sup_{\sigma}
    \mathbb E_{U\sim\mathcal U}
    \sum_b
    \bra{b}U\sigma U^\dagger\ket{b}
    \left|
        \bra{b}
        U\mathcal M^{-1}(Q)U^\dagger
        \ket{b}
    \right|^2,
    \label{eq:shadow-norm}
\end{equation}
where the supremum is over all density operators $\sigma$. To recover the real and imaginary parts of $\mu_{i,a,b}$, define
\begin{equation}
    Q_i^{(X)}:=X_C\otimes O_i,
    \qquad
    Q_i^{(Y)}:=Y_C\otimes O_i.
\end{equation}
The statistical cost is controlled by
\begin{equation}
\boxed{
    \nu_O^2
    :=
    \max_{i\in[M]}
    \left\{
        \|Q_i^{(X)}\|_{\mathrm{shadow}}^2,
        \|Q_i^{(Y)}\|_{\mathrm{shadow}}^2
    \right\}.
}
\label{eq:nu-shadow}
\end{equation}
The parameter $\nu_O^2$ controls the statistical cost of simultaneously
estimating the observables. To account separately for the classical
post-processing cost, write $\mathcal O:=\{O_i\}_{i=1}^M$ and let
$T_{\mathrm{eval}}(n,\mathcal O)$ denote the
maximum classical time, over $i\in[M]$, required to compute
\begin{equation}
    \operatorname{Tr}
    \left(
        Q_i^{(X)}\widehat\rho
    \right)
    \qquad\text{or}\qquad
    \operatorname{Tr}
    \left(
        Q_i^{(Y)}\widehat\rho
    \right)
    \label{eq:shadow-evaluation-cost}
\end{equation}
from the stored representation of $\widehat\rho$. This cost depends on both the shadow ensemble and the representation of
the queried observables, and need not be polynomial in the system size
$n$ for a general observable family.  For local Pauli classical shadows
and bounded-weight Pauli observables, this computation can be performed
efficiently, as shown in Sec.~\ref{sec:local-pauli-shadows}; further
details are given in Appendix~\ref{app:local-pauli-details}.

\subsection{Algorithm procedure}

The observable-dependent classical processing can be separated from the quantum data-acquisition stage.  For a fixed Krylov pair $(a,b)$, the quantum circuit below produces a reusable classical-shadow dataset that depends only on $H$, $|\psi\rangle$, $(a,b)$, and the chosen shadow ensemble, but not on the observables that will later be queried.

Here and below, we adopt the block-encoding normalization convention of Sec.~\ref{sec:section3}.

\begin{algorithm}[H]
\caption{\textsc{KrylovShadowSampling}}
\label{alg:krylov-shadow-sampling}
\LinesNumbered
\KwIn{
$(1,l,\epsilon_B)$-block encoding $U_H$ of $H$;
copies of the input state $\ket{\psi}$;
Krylov powers $a,b$;
shadow ensemble $(\mathcal U,\mathcal M)$;
number of samples $N$
}

\KwOut{Krylov-shadow dataset
$\mathcal D_{a,b}=\{(s_t,\widehat\rho_t)\}_{t=1}^{N}$}

\textbf{Step 1: Construct the controlled Krylov branches.}

Set $m\leftarrow\max\{a,b\}$ and $L_{a,b}\leftarrow lm$\;

Using sequential block-encoding multiplication~\cite{gilyen2018QSingValTransf}, construct block encodings $U_{H^a}$ and $U_{H^b}$ on a common $L_{a,b}$-qubit ancilla register $A$, with the identity acting on unused ancilla qubits\;

Construct $W_{a,b}\leftarrow|0\rangle\langle0|_C\otimes U_{H^a}+|1\rangle\langle1|_C\otimes U_{H^b}.$

\textbf{Step 2: Collect reusable classical-shadow data.}

Set $\mathcal D_{a,b}\leftarrow\{\}$\;

\For{$t=1,\ldots,N$}{
    Prepare $|+\rangle_C|0^{L_{a,b}}\rangle_A|\psi\rangle_S$\;

    Apply $W_{a,b}$ and denote the resulting state by
    $|\Psi_{a,b}\rangle_{CAS}$\;

    Sample $U_t\sim\mathcal U$ and apply $U_t$ to $C+S$\;

    Measure $A$ in the computational basis and obtain $z_t$\;

    Measure $C+S$ in the computational basis and obtain $b_t$\;

    Set
    $s_t\leftarrow\mathbf 1[z_t=0^{L_{a,b}}]$\;

    Store a classical representation of  $\widehat\rho_t\leftarrow\mathcal M^{-1}\left(U_t^\dagger|b_t\rangle\langle b_t|U_t\right).$

    Append $(s_t,\widehat\rho_t)$ to $\mathcal D_{a,b}$\;
}

\Return{$\mathcal D_{a,b}$}\;
\end{algorithm}

The next routine is entirely classical.  It may be applied immediately after the data are collected, or later to the same stored dataset after a new family of observables has been specified.

\begin{algorithm}[H]
\caption{\textsc{KrylovShadowPostProcessing}}
\label{alg:krylov-shadow-postprocessing}
\LinesNumbered
\KwIn{
Krylov-shadow dataset $\mathcal D_{a,b}$;
observables $\{O_i\}_{i=1}^{M}$;
shadow-norm upper bound $\nu_O^2$;
precision $\epsilon_S$;
failure probability $\eta$
}

\KwOut{$\{\widetilde\mu_{i,a,b}\}_{i=1}^{M}$}

\BlankLine
Set $K
    \leftarrow
    \left\lceil
        8\log\left({2M}/{\eta}\right)
    \right\rceil,
    S
    \leftarrow
    \left\lceil
        {8\nu_O^2}/{\epsilon_S^2}
    \right\rceil,
    N_{\mathrm{req}}
    \leftarrow KS.$

\If{$|\mathcal D_{a,b}|<N_{\mathrm{req}}$}{
    \Return{\textnormal{insufficient shadow data}}\;
}

Select any $N_{\mathrm{req}}$ samples $\{(s_t,\widehat\rho_t)\}_{t=1}^{N_{\mathrm{req}}}$ from $\mathcal D_{a,b}$\;

\For{$i=1,\ldots,M$}{
    \For{$t=1,\ldots,N_{\mathrm{req}}$}{
        $R_{i,t}\leftarrow
        s_t\operatorname{Tr}[(X_C\otimes O_i)\widehat\rho_t]$\;

        $I_{i,t}\leftarrow
        s_t\operatorname{Tr}[(Y_C\otimes O_i)\widehat\rho_t]$\;
    }

    $\widetilde R_{i,a,b}\leftarrow
    \operatorname{MedianOfMeans}
    (\{R_{i,t}\}_{t=1}^{N_{\mathrm{req}}};K,S)$\;

    $\widetilde I_{i,a,b}\leftarrow
    \operatorname{MedianOfMeans}
    (\{I_{i,t}\}_{t=1}^{N_{\mathrm{req}}};K,S)$\;

    $\widetilde\mu_{i,a,b}
    \leftarrow
    \widetilde R_{i,a,b}
    +
    \mathrm i\widetilde I_{i,a,b}$\;
}

\Return{$\{\widetilde\mu_{i,a,b}\}_{i=1}^{M}$}\;
\end{algorithm}

We use the standard median-of-means estimator~\cite{jerrum1986random} in Algorithm~\ref{alg:krylov-shadow-postprocessing}; the corresponding subroutine is given in Appendix~\ref{app:observable-aware-details}.

\begin{algorithm}[H]
\caption{\textsc{ClassicalShadowSQSVT}}
\label{alg:classical-shadow-sqsvt}
\LinesNumbered

\KwIn{
$(1,l,\epsilon_B)$-block encoding $U_H$ of $H$;
copies of the input state $\ket{\psi}$;
matrix function $f(\cdot): \mathbb{C}^{2^n \times 2^n} \to \mathbb{C}^{2^n \times 2^n}$;
Krylov dimension $r$;
observables $\{O_i\}_{i=1}^{M}$;
shadow ensemble $(\mathcal U,\mathcal M)$;
shadow-norm upper bound $\nu_O^2$;
coefficient precision $\epsilon_\lambda$;
cross-moment precision $\epsilon_S$;
failure probability $\delta$
}

\KwOut{$\{\widetilde y_i\}_{i=1}^{M}$}

\BlankLine
\textbf{Step 1: Estimate the Krylov coefficients.}

$\widetilde{\boldsymbol{\lambda}}
\leftarrow
\textsc{StateAwareQSVT}
(U_H,|\psi\rangle,f,r,\epsilon_\lambda,\delta/2)
\,[\text{Steps 1--2}]$\;

\BlankLine
\textbf{Step 2: Collect the Krylov-shadow datasets.}

Set
\[
    N_{\mathrm{pair}}
    \leftarrow
    \frac{r(r+1)}{2},
    \qquad
    \delta_{\mathrm{pair}}
    \leftarrow
    \frac{\delta}{2N_{\mathrm{pair}}},
\]
\[
    K_B
    \leftarrow
    \left\lceil
        8\log\left(\frac{2M}{\delta_{\mathrm{pair}}}\right)
    \right\rceil,
    \qquad
    S_B
    \leftarrow
    \left\lceil
        \frac{8\nu_O^2}{\epsilon_S^2}
    \right\rceil,
    \qquad
    N_B^{\mathrm{pair}}
    \leftarrow K_BS_B.
\]

\For{$0\leq a\leq b\leq r-1$}{
    $\mathcal D_{a,b}
    \leftarrow
    \textsc{KrylovShadowSampling}
    (U_H,|\psi\rangle,a,b,(\mathcal U,\mathcal M),N_B^{\mathrm{pair}})$\;
}

\BlankLine
\textbf{Step 3: Reconstruct the observable Krylov matrices.}

\For{$i=1,\ldots,M$}{
    $\widetilde B_{r,i}\leftarrow0_{r\times r}$\;
}

\For{$0\leq a\leq b\leq r-1$}{
    $\{\widetilde\mu_{i,a,b}\}_{i=1}^{M}
    \leftarrow
    \textsc{KrylovShadowPostProcessing}
    (\mathcal D_{a,b},\{O_i\}_{i=1}^{M},\nu_O^2,
    \epsilon_S,\delta_{\mathrm{pair}})$\;

    \For{$i=1,\ldots,M$}{
        \eIf{$a=b$}{
            $(\widetilde B_{r,i})_{a,a}
            \leftarrow
            \operatorname{Re}\widetilde\mu_{i,a,a}$\;
        }{
            $(\widetilde B_{r,i})_{a,b}
            \leftarrow
            \widetilde\mu_{i,a,b}$\;

            $(\widetilde B_{r,i})_{b,a}
            \leftarrow
            \widetilde\mu_{i,a,b}^{*}$\;
        }
    }
}

\BlankLine
\textbf{Step 4: Compute the observable estimates.}

\For{$i=1,\ldots,M$}{
    $\widetilde y_i
    \leftarrow
    \widetilde{\boldsymbol{\lambda}}^\dagger
    \widetilde B_{r,i}
    \widetilde{\boldsymbol{\lambda}}$\;
}

\Return{$\{\widetilde y_i\}_{i=1}^{M}$}\;
\end{algorithm}

\begin{remark}[Post-hoc observable queries]
The datasets $\{\mathcal D_{a,b}\}$ generated by Algorithm~\ref{alg:krylov-shadow-sampling} in Step~2 are independent of the identities of the observables.  They may therefore be stored and reused for a later observable family.  A new family can be evaluated without additional quantum measurements whenever the stored number of samples is sufficient for its shadow-norm bound and target precision; otherwise additional samples must be collected.
\end{remark}

\subsection{Main observable-estimation guarantee}

The guarantees below adopt the block-encoding normalization convention of Sec.~\ref{sec:section3} and are stated for an exact $(1,l,0)$-block encoding of $H$.\footnote{The additional bias arising from an approximate block encoding is not included in the present analysis.}
Let $\widetilde{\boldsymbol{\lambda}}$ denote the output of the coefficient estimation procedure and suppose that $\|\widetilde{\boldsymbol{\lambda}}-\boldsymbol{\lambda}\|_2\leq\epsilon_\lambda$ with failure probability at most $\delta/2$, using $N_\lambda(\epsilon_\lambda,\delta/2)$ quantum samples. We use $\epsilon_S$ to denote the target entrywise precision for the observable Krylov matrices, $|(\widetilde B_{r,i})_{a,b}-(B_{r,i})_{a,b}|\leq \epsilon_S.$ Define $\theta:=\|U_r^{-1}\|_2$ and $\beta_O:=\max_{i\in[M]}\|B_{r,i}\|_2$.

For the resource analysis, let $Q_\lambda$, $D_\lambda$, $G_\lambda$, and $T_\lambda^{\mathrm{cl}}$ denote the block-encoding query complexity, maximum circuit depth, total quantum gate complexity, and classical computational cost, respectively, of the coefficient-estimation step.

\begin{theorem}[End-to-end observable estimation]
\label{thm:end-to-end-observable-estimation}
Assume that $U_H$ is an exact block encoding of $H$, and let $\epsilon_O>0$ be the target observable precision. Choose $\epsilon_\lambda>0$ such that
$\beta_O(2\theta\epsilon_\lambda+\epsilon_\lambda^2)<\epsilon_O$, and set
\begin{equation}
    \epsilon_S
    =
    \frac{
        \epsilon_O-\beta_O(2\theta\epsilon_\lambda+\epsilon_\lambda^2)
    }{
        r(\theta+\epsilon_\lambda)^2
    },
    \label{eq:end-to-end-epsilon-S}
\end{equation}
Then Algorithm~\ref{alg:classical-shadow-sqsvt} outputs estimates satisfying
\begin{equation}
    |\widetilde y_i-y_i|
    \leq
    \epsilon_O
    \qquad
    \text{for all }i\in[M]
\end{equation}
with probability at least $1-\delta$. The required number of shadow samples is
\begin{equation}
    N_B
    =
    O\!\left(
        \frac{
            r^4(\theta+\epsilon_\lambda)^4\nu_O^2
        }{
            [\epsilon_O-\beta_O(2\theta\epsilon_\lambda+\epsilon_\lambda^2)]^2
        }
        \log\!\left(\frac{Mr^2}{\delta}\right)
    \right).
    \label{eq:end-to-end-NB}
\end{equation}
Therefore, the total sample complexity is
\begin{equation}
    N_{\mathrm{obs}}
    =
    N_\lambda(\epsilon_\lambda,\delta/2)+N_B.
    \label{eq:end-to-end-sample}
\end{equation}
The corresponding block-encoding query complexity, maximum circuit depth, total gate complexity, and total classical computational cost are
\begin{align}
    Q_{\mathrm{obs}}
    &=
    Q_\lambda+Q_B,
    \\
    D_{\mathrm{obs}}
    &=
    \max\!\left\{
        D_\lambda,\,
        D_B
    \right\},
    \\
    G_{\mathrm{obs}}
    &=
    G_\lambda+G_B,
    \\
    T_{\mathrm{obs}}^{\mathrm{cl}}
    &=
    T_\lambda^{\mathrm{cl}}
    +
    T_B^{\mathrm{cl}}
    +
    O(Mr^2),
\end{align}
where $Q_B$, $D_B$, $G_B$, and $T_B^{\mathrm{cl}}$ are given in Lemma~\ref{lem:observable-krylov-resources}.
\end{theorem}

\begin{proof}[Proof sketch]
Let $\Delta B_{r,i}:=\widetilde B_{r,i}-B_{r,i}$ and $\Delta\boldsymbol{\lambda}:=\widetilde{\boldsymbol{\lambda}}-\boldsymbol{\lambda}$. From Lemma~\ref{lem:observable-krylov-estimation}, $\|\Delta B_{r,i}\|_2\leq r\epsilon_S$. Using $\|\boldsymbol{\lambda}\|_2\leq\theta$ and $\|\Delta\boldsymbol{\lambda}\|_2\leq\epsilon_\lambda$, we obtain
\begin{equation}
    |\widetilde y_i-y_i|
    \leq
    (\theta+\epsilon_\lambda)^2 r\epsilon_S
    +
    \beta_O(2\theta\epsilon_\lambda+\epsilon_\lambda^2).
    \label{eq:end-to-end-y-bound}
\end{equation}
The stated choice of $\epsilon_S$ then guarantees $|\widetilde y_i-y_i|\leq\epsilon_O$. The sample complexity follows by combining the coefficient-estimation cost with Lemma~\ref{lem:observable-krylov-estimation}. Likewise, the query and gate complexities are additive across the two estimation stages, while the maximum circuit depth is the larger of the two stage-wise depths. For the classical cost,
Algorithm~\ref{alg:krylov-shadow-postprocessing} computes
$R_{i,t}$ and $I_{i,t}$ for every observable $i$ and every stored
sample $t$.  This requires
$O(MN_B[T_{\mathrm{eval}}(n,\mathcal O)+1])$ classical time, in addition
to the coefficient-estimation cost and the $O(Mr^2)$ cost of assembling
the observable Krylov matrices and evaluating the final quadratic forms. The complete proof is given in Appendix~\ref{app:observable-aware-details}.
\end{proof}

\begin{remark}
    Under $\|H\|_2\leq1$ and normalized $|\psi\rangle$, one may further use $\beta_O\leq r\max_{i\in[M]}\|O_i\|_2$.
\end{remark}

\begin{remark}[No final-state preparation]
Once $\widetilde{\boldsymbol{\lambda}}$ and $\{\widetilde B_{r,i}\}_{i=1}^{M}$ have been obtained, the remaining computation $\widetilde y_i=\widetilde{\boldsymbol{\lambda}}^{\,\dagger} \widetilde B_{r,i}\widetilde{\boldsymbol{\lambda}}$ is entirely classical. Thus no additional quantum samples are required, and the transformed state $P_r(H)|\psi\rangle$ never needs to be prepared.
\end{remark}



\subsection{Observable Krylov matrix estimation}

The remaining ingredient is the simultaneous estimation of the observable Krylov matrices.

\begin{lemma}[Observable Krylov matrix estimation]
\label{lem:observable-krylov-estimation}
Assume that $U_H$ is an exact $(1,l,0)$-block encoding of $H$. For any fixed Krylov pair $(a,b)$, let $\mathcal D_{a,b}$ be a dataset generated by Algorithm~\ref{alg:krylov-shadow-sampling}.  If the dataset contains at least
\begin{equation}
    N_{a,b}
    =
    O\!\left(
        \frac{\nu_O^2}{\epsilon_S^2}
        \log\left(\frac{M}{\eta}\right)
    \right)
\end{equation}
samples, then Algorithm~\ref{alg:krylov-shadow-postprocessing} outputs estimates $\{\widetilde\mu_{i,a,b}\}_{i=1}^{M}$ such that
\begin{equation}
    |\widetilde\mu_{i,a,b}-\mu_{i,a,b}|
    \leq
    \epsilon_S
    \qquad
    \text{for all }i\in[M]
    \label{eq:fixed-pair-error}
\end{equation}
with probability at least $1-\eta$. Consequently, applying the procedure to all $0\leq a\leq b\leq r-1$ and using Hermitian symmetry yields estimates $\{\widetilde B_{r,i}\}_{i=1}^{M}$ satisfying
\begin{equation}
    |(\widetilde B_{r,i})_{a,b}-(B_{r,i})_{a,b}|
    \leq
    \epsilon_S
    \label{eq:all-B-entrywise}
\end{equation}
simultaneously for all $i\in[M]$ and $a,b\in\{0,\ldots,r-1\}$ with probability at least $1-\delta_B$. The total number of classical-shadow samples is
\begin{equation}
    N_B
    =
    O\!\left(
        \frac{r^2\nu_O^2}{\epsilon_S^2}
        \log\!\left(\frac{Mr^2}{\delta_B}\right)
    \right).
    \label{eq:all-B-sample}
\end{equation}
\end{lemma}
The proof is deferred to Appendix~\ref{app:observable-aware-details}.

\begin{lemma}[Resources for observable Krylov matrix estimation]
\label{lem:observable-krylov-resources}
Under the assumptions of Lemma~\ref{lem:observable-krylov-estimation}, suppose that one application of $U_H$ or controlled-$U_H$ has circuit depth at most $D_H$ and gate count at most $C_H$. Let $D_{\mathrm{sh}}$ and $G_{\mathrm{sh}}$ denote the circuit depth and gate count of one classical-shadow measurement circuit on $C+S$. For a fixed Krylov pair $(a,b)$, one sample generated by Algorithm~\ref{alg:krylov-shadow-sampling} uses $\max\{a,b\}$ block-encoding queries and has circuit depth $D_{a,b}=O\!\left(\max\{a,b\}D_H +D_{\mathrm{sh}}\right).$ Estimating all observable Krylov matrices to entrywise precision $\epsilon_S$ and failure probability at most $\delta_B$ requires
\begin{equation}
    Q_B
    =
    O\!\left(
        \frac{r^3\nu_O^2}{\epsilon_S^2}
        \log\!\left(
            \frac{Mr^2}{\delta_B}
        \right)
    \right)
    \label{eq:all-B-query}
\end{equation}
queries to $U_H$ or controlled-$U_H$. The maximum circuit depth among all shadow-measurement rounds is
\begin{equation}
    D_B
    =
    O\!\left(
        rD_H+D_{\mathrm{sh}}
    \right).
    \label{eq:all-B-depth}
\end{equation}
Algorithm~\ref{alg:krylov-shadow-postprocessing} is entirely classical and therefore contributes no additional block-encoding queries or quantum circuit depth.  The total quantum gate count over all rounds is
\begin{equation}
    G_B
    =
    O\!\left[
        \frac{\nu_O^2}{\epsilon_S^2}
        \log\!\left(
            \frac{Mr^2}{\delta_B}
        \right)
        \left(
            r^3C_H
            +
            r^2G_{\mathrm{sh}}
        \right)
    \right].
    \label{eq:all-B-gates}
\end{equation}
For the classical cost, recall that $N_B$ is the total number of shadow samples in Eq.~\eqref{eq:all-B-sample}, and $T_{\mathrm{eval}}(n,\mathcal O)$, defined in Eq.~\eqref{eq:shadow-evaluation-cost}, is the maximum classical time required to compute $\operatorname{Tr}(Q_i^{(X)}\widehat\rho)$ or $\operatorname{Tr}(Q_i^{(Y)}\widehat\rho)$ from one stored snapshot, maximized over $i\in[M]$. The classical post-processing cost for all observable Krylov matrices is therefore
\begin{equation}
    T_B^{\mathrm{cl}}
    =
    O\!\left(
        M N_B
        \left[
            T_{\mathrm{eval}}(n,\mathcal O)+1
        \right]
        +
        Mr^2
    \right),
    \label{eq:all-B-classical}
\end{equation}
where the $Mr^2$ term accounts for assembling the $M$ estimated $r\times r$ observable Krylov matrices.
\end{lemma}
The proof is deferred to Appendix~\ref{app:observable-aware-details}.

\begin{remark}[Dependence on the block-encoding normalization]
\label{rem:algo2-alpha-normalization}
For an exact $(\alpha,l,0)$-block encoding, define
\[
    \overline H:=H/\alpha,
    \qquad
    \overline f(x):=f(\alpha x),
    \qquad
    D_\alpha
    :=
    \operatorname{diag}
    (1,\alpha^{-1},\ldots,\alpha^{-(r-1)}),
\]
as in Sec.~\ref{sec:section3}. The same unitary $U_H$ is then an
exact $(1,l,0)$-block encoding of $\overline H$, and
\[
    \overline A_r=A_rD_\alpha,
    \qquad
    \overline B_{r,i}
    =
    D_\alpha B_{r,i}D_\alpha.
\]
Thus, for a fixed entrywise precision $\epsilon_S$, the sample,
query, gate, and classical-computational bounds in
Lemmas~\ref{lem:observable-krylov-estimation}
and~\ref{lem:observable-krylov-resources}
retain the same form for the rescaled problem. The number of
sequential block-encoding layers is also unchanged.

For the end-to-end guarantee, since
$\overline U_r=U_rD_\alpha$, define
\[
    \overline\theta
    :=
    \|\overline U_r^{-1}\|_2
    =
    \|D_\alpha^{-1}U_r^{-1}\|_2,
    \qquad
    \overline\beta_O
    :=
    \max_{i\in[M]}
    \|D_\alpha B_{r,i}D_\alpha\|_2.
\]
For $\alpha\geq1$,
\[
    \overline\theta
    \leq
    \alpha^{r-1}\theta,
    \qquad
    \overline\beta_O
    \leq
    \beta_O.
\]
Define
\[
    \Delta_\alpha
    :=
    \epsilon_O
    -
    \beta_O
    \left(
        2\alpha^{r-1}\theta\epsilon_\lambda
        +
        \epsilon_\lambda^2
    \right),
\]
and suppose that $\Delta_\alpha>0$. Then the observable-Krylov-matrix
resources satisfy
\begin{align}
    N_B^{(\alpha)}
    &=
    O\!\left(
        \frac{
            r^4
            (\alpha^{r-1}\theta+\epsilon_\lambda)^4
            \nu_O^2
        }{
            \Delta_\alpha^2
        }
        \log\!\left(
            \frac{Mr^2}{\delta}
        \right)
    \right),
    \\
    Q_B^{(\alpha)}
    &=
    O\!\left(
        \frac{
            r^5
            (\alpha^{r-1}\theta+\epsilon_\lambda)^4
            \nu_O^2
        }{
            \Delta_\alpha^2
        }
        \log\!\left(
            \frac{Mr^2}{\delta}
        \right)
    \right),
    \\
    D_B^{(\alpha)}
    &=
    O\!\left(
        rD_H+D_{\mathrm{sh}}
    \right),
    \\
    G_B^{(\alpha)}
    &=
    O\!\left[
        \frac{
            (\alpha^{r-1}\theta+\epsilon_\lambda)^4
            \nu_O^2
        }{
            \Delta_\alpha^2
        }
        \log\!\left(
            \frac{Mr^2}{\delta}
        \right)
        \left(
            r^5C_H+r^4G_{\mathrm{sh}}
        \right)
    \right],
\end{align}
and
\[
    T_{B,\alpha}^{\mathrm{cl}}
    =
    O\!\left(
        M N_B^{(\alpha)}
        \left[
            T_{\mathrm{eval}}(n,\mathcal O)+1
        \right]
        +
        Mr^2
    \right).
\]
Hence, when $\Delta_\alpha$ is bounded away from zero and the
remaining parameters are fixed, $N_B^{(\alpha)}$,
$Q_B^{(\alpha)}$, $G_B^{(\alpha)}$, and the sample-dependent part of
$T_{B,\alpha}^{\mathrm{cl}}$ may incur an
$\alpha^{4(r-1)}$ worst-case amplification, whereas the maximum
circuit-depth scaling is unchanged. The resources required for the
coefficient-estimation step are those of the rescaled instance
$(\overline H,\overline f)$ and are combined with the above bounds
as in Theorem~\ref{thm:end-to-end-observable-estimation}.
\end{remark}
\subsection{Local Pauli specialization}
\label{sec:local-pauli-shadows}

For the local Pauli classical-shadow ensemble, the inverse measurement channel and the snapshot both factorize over qubits.  If each $O_i$ is a Pauli string of weight at most $k$, then $Q_i^{(X)}$ and $Q_i^{(Y)}$ have weight at most $k+1$.  This gives the bound $\nu_O^2\leq 3^{k+1}$ and also allows $\operatorname{Tr}(Q_i^{(X)}\widehat\rho)$ and $\operatorname{Tr}(Q_i^{(Y)}\widehat\rho)$ to be evaluated using only the qubits in their supports. The detailed construction and classical post-processing are given in Appendix~\ref{app:local-pauli-details}; here we record the resulting resource bounds.

\begin{corollary}[Local Pauli specialization]
\label{cor:local-pauli-specialization}
Suppose that every $O_i$ is a Pauli string of weight at most $k$. Under local Pauli classical shadows, $\nu_O^2\leq3^{k+1}$. Hence the observable Krylov matrices can be estimated using
\begin{equation}
    N_B
    =
    O\!\left(
        \frac{r^2 3^{k+1}}{\epsilon_S^2}
        \log\!\left(\frac{Mr^2}{\delta_B}\right)
    \right)
    \label{eq:all-B-pauli-sample}
\end{equation}
shadow samples. Under the choice of $\epsilon_S$ in Eq.~\eqref{eq:end-to-end-epsilon-S}, the end-to-end shadow complexity is
\begin{equation}
    N_B
    =
    O\!\left(
        \frac{
            r^4(\theta+\epsilon_\lambda)^4 3^{k+1}
        }{
            [\epsilon_O-\beta_O(2\theta\epsilon_\lambda+\epsilon_\lambda^2)]^2
        }
        \log\!\left(\frac{Mr^2}{\delta}\right)
    \right).
\end{equation}
Moreover, the total block-encoding query complexity is
\begin{equation}
    Q_B
    =
    O\!\left(
        \frac{r^3 3^{k+1}}{\epsilon_S^2}
        \log\!\left(
            \frac{Mr^2}{\delta_B}
        \right)
    \right).
\end{equation}
Since local Pauli measurements have $D_{\mathrm{sh}}=O(1)$ and $G_{\mathrm{sh}}=O(n)$, the maximum circuit depth is
\begin{equation}
    D_B
    =
    O(rD_H),
\end{equation}
and the total gate complexity is
\begin{equation}
    G_B
    =
    O\!\left[
        \frac{3^{k+1}}{\epsilon_S^2}
        \log\!\left(
            \frac{Mr^2}{\delta_B}
        \right)
        \left(
            r^3C_H+r^2n
        \right)
    \right].
\end{equation}
If each $O_i$ is specified by its support and Pauli labels, then computing
$\operatorname{Tr}(Q_i^{(X)}\widehat\rho)$ or $\operatorname{Tr}(Q_i^{(Y)}\widehat\rho)$ requires accessing only the control qubit and the at most $k$ system qubits in the support of $O_i$. Therefore,
\begin{equation}
    T_{\mathrm{eval}}(n,\mathcal O)
    =
    O(k),
    \label{eq:local-pauli-evaluation-cost}
\end{equation}
and
\begin{equation}
    T_B^{\mathrm{cl}}
    =
    O\!\left(
        \frac{
            M k r^2 3^{k+1}
        }{
            \epsilon_S^2
        }
        \log\!\left(
            \frac{Mr^2}{\delta_B}
        \right)
        +
        Mr^2
    \right).
    \label{eq:local-pauli-classical-complexity}
\end{equation}
Moreover, a local-Pauli snapshot can be stored using $O(n)$ classical data, so the total storage required for the Krylov-shadow datasets is $O(nN_B)$.
\end{corollary}

Thus, for bounded-weight Pauli observables, the shadow-sample, block-encoding query, and quantum gate complexities depend only logarithmically on the number $M$ of queried observables.  The classical post-processing cost is necessarily at least linear in $M$, since $M$ estimates are produced.  Under local Pauli shadows, however, each quantity $\operatorname{Tr}(Q_i^{(X)}\widehat\rho)$ or $\operatorname{Tr}(Q_i^{(Y)}\widehat\rho)$ can be computed in $O(k)$ time and hence introduces no exponential dependence on the system size $n$.
\section{Observable-Aware S-QSVT}

The Classical-Shadow S-QSVT method developed in the previous section performs the reduced dynamics within the full state-Krylov space and subsequently evaluates the prescribed observables through the Krylov matrices $B_{r,i}=A_r^\dagger O_i A_r.$ However, for observable estimation, not every direction in the state-Krylov space need be distinguishable by the observable family. This motivates an observable-aware reduction in which Krylov directions that are indistinguishable under the prescribed observables are identified, allowing the effective dynamics to be represented on a smaller space.

The main idea is to identify state directions that remain invisible to all observables of interest under every future application of the Hamiltonian. After quotienting out these directions, the resulting observable-relevant dynamics may have dimension $s<r$. Rather than manipulating the quotient space directly, we construct a concrete representation of it through a history map. The image of this map will be called the \emph{history space}.

\paragraph{Roadmap.}
The observable-aware construction proceeds as follows.
\begin{enumerate}
    \item We identify the observable-invisible directions
    $K^{\mathcal O}$ and quotient them out, obtaining the
    observable-relevant space
    $\mathcal Q_{\mathcal O}=\mathcal K_r/K^{\mathcal O}$ of
    dimension $s$.

    \item We show that the first $s$ quotient Krylov vectors form a
    basis of $\mathcal Q_{\mathcal O}$ and that every polynomial
    transformation admits an observable-equivalent representative of
    degree at most $s-1$.

    \item Because the quotient space is abstract, we realize it through
    the finite history map $F_s^{\mathcal O}$.  This gives a concrete
    $s$-dimensional history space in which the induced Hamiltonian
    action is represented by $J_s^{\mathcal O}$.  The target observable
    values can then be expressed using
    $P(J_s^{\mathcal O})\mathbf e_0$ and the reduced observable matrices
    $B_{s,i}$.

    \item Finally, we estimate the matrices needed to reconstruct
    $J_s^{\mathcal O}$ and $B_{s,i}$ using the Krylov-shadow routines
    developed in the previous section, and propagate these estimation
    errors to obtain the end-to-end observable-estimation guarantee.
\end{enumerate}

The technical proofs for the history-space construction and the end-to-end error analysis are collected in Appendix~\ref{app:observable-aware-details-algo3}. Throughout this section, we assume that the observable-relevant dimension $s$ is known. 

\subsection{The observable-aware Krylov subspace}
\label{subsec:observable-invisible}
Let
\begin{equation}
    \mathcal K_r(H,\ket{\psi})
    :=
    \operatorname{span}
    \left\{
        \ket{\psi},
        H\ket{\psi},
        \ldots,
        H^{r-1}\ket{\psi}
    \right\}
\end{equation}
be the saturated state-Krylov space generated by the initial state $\ket{\psi}$. We assume
\begin{equation}
    \dim \mathcal K_r(H,\ket{\psi})=r,
    \qquad
    H\mathcal K_r(H,\ket{\psi})
    \subseteq
    \mathcal K_r(H,\ket{\psi}).
    \label{eq:state-krylov-invariance-OA}
\end{equation}
Let
\begin{equation}
    \mathcal O
    :=
    \{O_1,\ldots,O_M\}
\end{equation}
be the collection of observables whose expectation values we wish to predict. We define the subspace of directions that are invisible to all observables, both immediately and after every future application of $H$, by
\begin{equation}
    K^{\mathcal O}
    :=
    \left\{
        \ket{\phi}\in\mathcal K_r(H,\ket{\psi})
        :
        O_i H^\ell\ket{\phi}=0
        \text{ for all }
        i\in[M]
        \text{ and }
        \ell\geq0
    \right\}.
    \label{eq:multi-observable-invisible}
\end{equation}

\begin{lemma}[Structural properties]
\label{rem:observable-aware-structural-properties}
The following properties will be used repeatedly throughout this section.
\begin{enumerate}
    \item \textbf{State-Krylov invariance.}
    Recall from Eq.~\eqref{eq:state-krylov-invariance-OA} that $H\mathcal K_r(H,\ket{\psi}) \subseteq \mathcal K_r(H,\ket{\psi}).$

    \item \textbf{Observable-invisible invariance.}
    The subspace $K^{\mathcal O}$ is also $H$-invariant:
    \begin{equation}
        HK^{\mathcal O}
        \subseteq
        K^{\mathcal O}.
        \label{eq:observable-invisible-invariance}
    \end{equation}
    Indeed, if $\ket{\phi}\in K^{\mathcal O}$, then Eq.~\eqref{eq:state-krylov-invariance-OA} ensures that $H\ket{\phi}\in\mathcal K_r(H,\ket{\psi})$, while $O_iH^\ell H\ket{\phi}=O_iH^{\ell+1}\ket{\phi}=0$ for every $i\in[M]$ and $\ell\geq0$. This invariance ensures that the Hamiltonian action is well-defined after quotienting out $K^{\mathcal O}$.

    \item \textbf{Observable annihilation.}
    Since the definition in Eq.~\eqref{eq:multi-observable-invisible} includes $\ell=0$,
    \begin{equation}
        K^{\mathcal O}
        \subseteq
        \bigcap_{i=1}^{M}\ker O_i.
        \label{eq:observable-invisible-annihilation}
    \end{equation}
    Hence components in $K^{\mathcal O}$ do not contribute to the target observable values and may be discarded without changing those quantities.
\end{enumerate}
\end{lemma}

We may therefore define the observable-relevant quotient space
\begin{equation}
    \mathcal Q_{\mathcal O}
    :=
    \mathcal K_r/K^{\mathcal O}.
    \label{eq:observable-quotient}
\end{equation}
Two vectors $\ket{\phi_1},\ket{\phi_2}\in\mathcal K_r$ represent the same element of the quotient space whenever their difference belongs to the observable-invisible subspace:
\begin{equation}
    \ket{\phi_1}\sim\ket{\phi_2}
    \quad\Longleftrightarrow\quad
    \ket{\phi_1}-\ket{\phi_2}\in K^{\mathcal O}.
    \label{eq:observable-equivalence-relation}
\end{equation}
For $\ket{\phi}\in\mathcal K_r$, we denote its equivalence class by
\begin{equation}
    [\ket{\phi}]
    :=
    \left\{
        \ket{\phi}+\ket{\eta}
        :
        \ket{\eta}\in K^{\mathcal O}
    \right\}.
    \label{eq:observable-equivalence-class}
\end{equation}
Thus, two state-Krylov vectors are represented by the same element of $\mathcal Q_{\mathcal O}$ precisely when they differ only by a direction that remains invisible to all observables in $\mathcal O$ under every future application of $H$. We assume that
\begin{equation}
    \dim\mathcal Q_{\mathcal O}=s.
    \label{eq:observable-aware-dimension}
\end{equation}

\subsection{Observable-aware Krylov reduction}
\label{subsec:observable-aware-krylov-reduction}
Since $K^{\mathcal O}$ is $H$-invariant, $H$ induces a well-defined linear operator on the quotient space, 
\begin{equation}
    \overline H[\ket{\phi}]
    :=
    [H\ket{\phi}],
    \label{eq:quotient-H}
\end{equation}
The quotient space is generated by the quotient Krylov sequence
\begin{equation}
    [\ket{\psi}],
    [H\ket{\psi}],
    [H^2\ket{\psi}],
    \ldots.
\end{equation}
\begin{lemma}[Observable-aware Krylov basis]
\label{lem:observable-aware-basis}
Suppose $\dim\mathcal Q_{\mathcal O}=s.$ Then
\begin{equation}
    [\ket{\psi}],
    [H\ket{\psi}],
    \ldots,
    [H^{s-1}\ket{\psi}]
    \label{eq:first-s-quotient-vectors}
\end{equation}
form a basis of $\mathcal Q_{\mathcal O}$. Consequently, for every $k\geq s$, there exist coefficients $a_{k,0},\ldots,a_{k,s-1}$ such that
\begin{equation}
    H^k\ket{\psi}
    -
    \sum_{j=0}^{s-1}
        a_{k,j}H^j\ket{\psi}
    \in
    K^{\mathcal O}.
    \label{eq:observable-krylov-recurrence}
\end{equation}

\end{lemma}

\begin{proof}
Suppose, toward a contradiction, that the vectors in Eq.~\eqref{eq:first-s-quotient-vectors} are linearly dependent. Let $d<s$ be the smallest index such that $[H^d\ket{\psi}]\in\operatorname{span}\left\{[\ket{\psi}],\ldots,[H^{d-1}\ket{\psi}]\right\}.$ Repeated application of the induced operator $\overline H$ shows that every higher quotient Krylov vector belongs to the same span. Hence $\dim\mathcal Q_{\mathcal O}\leq d<s,$ contradicting the assumption $\dim\mathcal Q_{\mathcal O}=s$. Therefore the first $s$ quotient Krylov vectors are linearly independent and hence form a basis of $\mathcal Q_{\mathcal O}$. The second claim follows immediately by expressing $[H^k\ket{\psi}]$ in this basis for every $k\geq s$. 
\end{proof}

\begin{corollary}[Observable-equivalent polynomial reduction]
\label{cor:observable-polynomial-reduction}
For every polynomial $P$, there exists a polynomial
$P_s^{\mathcal O}$ of degree at most $s-1$ such that
\begin{equation}
    P(H)\ket{\psi}
    -
    P_s^{\mathcal O}(H)\ket{\psi}
    \in
    K^{\mathcal O}.
    \label{eq:observable-equivalent-polynomial-reduction}
\end{equation}
Consequently, for every Hermitian observable $O_i\in\mathcal O$,
\begin{equation}
    \bra{\psi}P(H)^\dagger O_iP(H)\ket{\psi}
    =
    \bra{\psi}
        P_s^{\mathcal O}(H)^\dagger
        O_i
        P_s^{\mathcal O}(H)
    \ket{\psi}.
    \label{eq:observable-equivalent-polynomial-output}
\end{equation}
\end{corollary}

\begin{proof}
Expand $P(H)\ket{\psi}$ as a linear combination of Krylov powers.
By Lemma~\ref{lem:observable-aware-basis}, every term
$H^k\ket{\psi}$ with $k\geq s$ is equivalent modulo
$K^{\mathcal O}$ to a linear combination of
$\ket{\psi},\ldots,H^{s-1}\ket{\psi}$.  Collecting these reduced terms
defines a polynomial $P_s^{\mathcal O}$ of degree at most $s-1$ and
gives Eq.~\eqref{eq:observable-equivalent-polynomial-reduction}.
Since $K^{\mathcal O}\subseteq\ker O_i$ and $O_i$ is Hermitian, the
invisible component and all cross terms vanish in the quadratic form,
which proves Eq.~\eqref{eq:observable-equivalent-polynomial-output}.
\end{proof}

\subsection{History-space representation}
\label{subsec:history-space}
\paragraph{History map and quotient realization.}
The quotient space provides the abstract observable-relevant description. For the algorithm, however, it is convenient to represent each quotient class by a concrete vector containing its observable histories. Define the length-$s$ history map $F_s^{\mathcal O}:\mathcal K_r(H,\ket{\psi})\longrightarrow\mathbb C^{M}\otimes\mathbb C^{s}\otimes\mathcal H$ by
\begin{equation}
    F_s^{\mathcal O}(\ket{\phi})
    :=
    \sum_{i=1}^{M}
    \sum_{\ell=0}^{s-1}
        \ket{i}\ket{\ell}
        \otimes
        O_iH^\ell\ket{\phi}.
    \label{eq:history-map}
\end{equation}
The vector $F_s^{\mathcal O}(\ket{\phi})$ stores the first $s$ observable histories of $\ket{\phi}$ simultaneously for all observables in $\mathcal O$. The following lemma shows that these first $s$ histories contain all observable-relevant information. We define the history space as the image of the history map:
\begin{equation}
    \mathcal S_{\mathrm{hist}}^{\mathcal O}
    :=
    F_s^{\mathcal O}(\mathcal K_r).
    \label{eq:history-space}
\end{equation}
For each Krylov vector, define the corresponding history vector
\begin{equation}
    \ket{\widehat\phi_j}
    :=
    F_s^{\mathcal O}(H^j\ket{\psi}),
    \qquad
    j\geq0.
    \label{eq:history-vector}
\end{equation}
Explicitly,
\begin{equation}
    \ket{\widehat\phi_j}
    =
    \sum_{i=1}^{M}
    \sum_{\ell=0}^{s-1}
        \ket{i}\ket{\ell}
        \otimes
        O_iH^{\ell+j}\ket{\psi}.
    \label{eq:history-vector-expanded}
\end{equation}

\begin{lemma}[Finite-history characterization of invisible directions]
\label{lem:finite-history-kernel}
Suppose $\dim\mathcal Q_{\mathcal O}=s.$ Then $\ker F_s^{\mathcal O}=K^{\mathcal O}.$ Equivalently, for every $\ket{\phi}\in\mathcal K_r$,
\begin{equation}
    F_s^{\mathcal O}(\ket{\phi})=0
    \quad\Longleftrightarrow\quad
    \ket{\phi}\in K^{\mathcal O}.
    \label{eq:history-kernel-equivalence}
\end{equation}
\end{lemma}

The proof is deferred to Appendix~\ref{app:observable-aware-auxiliary-results}.

\begin{corollary}[History-space realization of the quotient]
\label{cor:history-space-isomorphism}
Suppose $\dim\mathcal Q_{\mathcal O}=s.$ Then $\mathcal Q_{\mathcal O} \cong \mathcal S_{\mathrm{hist}}^{\mathcal O},$ and therefore $\dim\mathcal S_{\mathrm{hist}}^{\mathcal O}=s.$ Moreover, $\{\ket{\widehat\phi_0},\ket{\widehat\phi_1},\ldots,\ket{\widehat\phi_{s-1}}\}$ form a basis of $\mathcal S_{\mathrm{hist}}^{\mathcal O}$.
\end{corollary}
The proof is deferred to Appendix~\ref{app:observable-aware-auxiliary-results}.

\begin{remark}
    The quotient space and the history space describe the same observable-relevant degrees of freedom, but they play different roles. The quotient space is an abstract construction that identifies which directions can be discarded, whereas the history space gives a concrete vector-space representation that can be reconstructed from observable moments.
\end{remark}

\paragraph{Effective dynamics in the history space.}
Having established that the observable-relevant quotient space admits an $s$-dimensional realization through the history space, we now construct a finite-dimensional representation of the induced Hamiltonian dynamics. Specifically, we use the first $s$ history vectors as a basis of $\mathcal S_{\mathrm{hist}}^{\mathcal O}$ and represent the shift
\begin{equation}
    \ket{\widehat\phi_j}
    \longmapsto
    \ket{\widehat\phi_{j+1}}
\end{equation}
by an $s\times s$ matrix.
This effective history-space dynamics can be reconstructed from inner products among the history vectors and will later allow us to evaluate polynomial transformations of $H$ without retaining the full $r$-dimensional state-Krylov representation.

Collect the first $s$ history vectors into
\begin{equation}
    A_s^{\mathcal O}
    :=
    \begin{bmatrix}
        \ket{\widehat\phi_0}
        &
        \ket{\widehat\phi_1}
        &
        \cdots
        &
        \ket{\widehat\phi_{s-1}}
    \end{bmatrix}.
    \label{eq:history-basis-matrix}
\end{equation}
By Corollary~\ref{cor:history-space-isomorphism}, $A_s^{\mathcal O}$ has full column rank. Define also the shifted history matrix
\begin{equation}
    A_{s,+}^{\mathcal O}
    :=
    \begin{bmatrix}
        \ket{\widehat\phi_1}
        &
        \ket{\widehat\phi_2}
        &
        \cdots
        &
        \ket{\widehat\phi_s}
    \end{bmatrix}.
    \label{eq:shifted-history-matrix}
\end{equation}
Since every column of $A_{s,+}^{\mathcal O}$ lies in the $s$-dimensional history space $\mathcal S_{\mathrm{hist}}^{\mathcal O}$, and the columns of $A_s^{\mathcal O}$ form a basis of this space, there exists a unique matrix
\begin{equation}
    J_s^{\mathcal O}
    \in
    \mathbb C^{s\times s}
\end{equation}
such that
\begin{equation}
    A_s^{\mathcal O} J_s^{\mathcal O}
    =
    A_{s,+}^{\mathcal O}.
    \label{eq:history-effective-dynamics}
\end{equation}
The matrix $J_s^{\mathcal O}$ represents the induced Hamiltonian dynamics in the generally nonorthogonal history basis $\{\ket{\widehat\phi_j}\}_{j=0}^{s-1}$. To reconstruct $J_s^{\mathcal O}$ from inner products of history vectors, define the history Gram matrix
\begin{equation}
    G_s^{\mathcal O}
    :=
    \left(A_s^{\mathcal O}\right)^\dagger
    A_s^{\mathcal O},
    \label{eq:history-gram}
\end{equation}
and the shifted history Gram matrix
\begin{equation}
    \widehat G_s^{\mathcal O}
    :=
    \left(A_s^{\mathcal O}\right)^\dagger
    A_{s,+}^{\mathcal O}.
    \label{eq:shifted-history-gram}
\end{equation}
Multiplying Eq.~\eqref{eq:history-effective-dynamics} from the left by $\left(A_s^{\mathcal O}\right)^\dagger$ gives
\begin{equation}
    G_s^{\mathcal O}J_s^{\mathcal O}
    =
    \widehat G_s^{\mathcal O}.
    \label{eq:history-generalized-equation}
\end{equation}
Since the columns of $A_s^{\mathcal O}$ are linearly independent, $G_s^{\mathcal O}$ is positive definite and therefore invertible. Hence
\begin{equation}
    J_s^{\mathcal O}
    =
    \left(G_s^{\mathcal O}\right)^{-1}
    \widehat G_s^{\mathcal O}.
    \label{eq:history-effective-matrix}
\end{equation}
The entries of $G_s^{\mathcal O}$ can be written explicitly as
\begin{align}
    \left(G_s^{\mathcal O}\right)_{a,b}
    =
    \braket{\widehat\phi_a|\widehat\phi_b}
    =
    \sum_{i=1}^{M}
    \sum_{\ell=0}^{s-1}
    \bra{\psi}
        H^{a+\ell}
        O_i^2
        H^{b+\ell}
    \ket{\psi},
    \label{eq:history-gram-moments}
\end{align}
for $a,b=0,\ldots,s-1$. Similarly,
\begin{align}
    \left(\widehat G_s^{\mathcal O}\right)_{a,b}
    = \braket{\widehat\phi_a|\widehat\phi_{b+1}}
    =
    \sum_{i=1}^{M}
    \sum_{\ell=0}^{s-1}
    \bra{\psi}
        H^{a+\ell}
        O_i^2
        H^{b+\ell+1}
    \ket{\psi}.
    \label{eq:shifted-history-gram-moments}
\end{align}
\paragraph{Reduced observable representation.}
We next use $J_s^{\mathcal O}$ to represent polynomial functions of the Hamiltonian within the observable-relevant dynamics. Define
\begin{equation}
    A_s
    :=
    \begin{bmatrix}
        \ket{\psi}
        &
        H\ket{\psi}
        &
        \cdots
        &
        H^{s-1}\ket{\psi}
    \end{bmatrix}.
    \label{eq:truncated-state-krylov}
\end{equation}
For each observable $O_i$, define the reduced observable Krylov matrix
\begin{equation}
    B_{s,i}
    :=
    A_s^\dagger O_iA_s,
    \label{eq:observable-aware-B}
\end{equation}
whose entries are
\begin{equation}
    (B_{s,i})_{a,b}
    =
    \bra{\psi}
        H^aO_iH^b
    \ket{\psi},
    \qquad
    a,b=0,\ldots,s-1.
    \label{eq:observable-aware-cross-moment}
\end{equation}
Let
\begin{equation}
    P(x)
    =
    \sum_{k=0}^{d}p_kx^k
\end{equation}
be a degree-$d$ polynomial. Define
\begin{equation}
    \boldsymbol{\lambda}^{\mathcal O}
    :=
    P(J_s^{\mathcal O})\mathbf e_0,
    \label{eq:observable-aware-coefficients}
\end{equation}
where
\begin{equation}
    \mathbf e_0
    :=
    (1,0,\ldots,0)^{\mathsf T}.
\end{equation}

\begin{remark}
    The matrix $J_s^{\mathcal O}$ is defined with respect to the generally nonorthogonal history basis and therefore need not be Hermitian in the standard Euclidean inner product on its coefficient space. Its role here is to represent the induced Hamiltonian action on the observable-relevant degrees of freedom.
\end{remark}
\subsection{Algorithm procedure}

For convenience, define
\begin{equation}
    \Omega_{\mathcal O}
    :=
    \sum_{i=1}^{M} O_i^2.
    \label{eq:history-aggregate-observable}
\end{equation}
Then Eqs.~\eqref{eq:history-gram-moments}
and~\eqref{eq:shifted-history-gram-moments} can be written as
\begin{align}
    \left(G_s^{\mathcal O}\right)_{a,b}
    &=
    \sum_{\ell=0}^{s-1}
    \bra{\psi}
        H^{a+\ell}
        \Omega_{\mathcal O}
        H^{b+\ell}
    \ket{\psi},
    \label{eq:history-G-Omega}
    \\
    \left(\widehat G_s^{\mathcal O}\right)_{a,b}
    &=
    \sum_{\ell=0}^{s-1}
    \bra{\psi}
        H^{a+\ell}
        \Omega_{\mathcal O}
        H^{b+\ell+1}
    \ket{\psi}.
    \label{eq:history-Ghat-Omega}
\end{align}
Define $Q_{\Omega}^{(X)}:=X_C\otimes\Omega_{\mathcal O}, Q_{\Omega}^{(Y)}:=Y_C\otimes\Omega_{\mathcal O},$
and let
\begin{equation}
    \nu_{\Omega}^2
    :=
    \max\left\{
        \|Q_{\Omega}^{(X)}\|_{\mathrm{shadow}}^2,
        \|Q_{\Omega}^{(Y)}\|_{\mathrm{shadow}}^2
    \right\}.
\label{eq:nu-Omega-shadow}
\end{equation}
Here $\|\cdot\|_{\mathrm{shadow}}$ is the shadow norm defined in Eq.~\eqref{eq:shadow-norm}. For the history-matrix estimation, define $T_{\mathrm{eval}}(n,\{\Omega_{\mathcal O}\})$ as the classical time required, for one stored snapshot $\widehat\rho$, to compute
\begin{equation}
    \operatorname{Tr}
    \left(
        Q_{\Omega}^{(X)}\widehat\rho
    \right)
    \qquad\text{or}\qquad
    \operatorname{Tr}
    \left(
        Q_{\Omega}^{(Y)}\widehat\rho
    \right).
    \label{eq:history-evaluation-cost}
\end{equation}
This definition follows the convention of
Eq.~\eqref{eq:shadow-evaluation-cost} and uses the chosen classical representations of $\widehat\rho$ and $\Omega_{\mathcal O}$. As in the previous section, this quantity accounts for the cost of the trace evaluation itself and may depend on both the shadow ensemble and the observable representation.

Here and below, we adopt the block-encoding normalization convention of Sec.~\ref{sec:section3}.

\begin{algorithm}[H]
\small
\caption{\textsc{ObservableAwareSQSVT}}
\label{alg:observable-aware-qsvt}
\LinesNumbered

\KwIn{
exact $(1,l,0)$-block encoding $U_H$ of $H$;
copies of the input state $\ket{\psi}$;
polynomial $P$;
known observable-relevant dimension $s$;
observables $\{O_i\}_{i=1}^{M}$;
shadow ensemble $(\mathcal U,\mathcal M)$;
shadow-norm upper bounds $\nu_O^2$ and $\nu_{\Omega}^2$;
history-matrix precisions $\epsilon_G,\epsilon_{\widehat G}$;
observable-matrix precision $\epsilon_S$;
failure probability $\delta$
}

\KwOut{$\{\widetilde y_i\}_{i=1}^{M}$}

\BlankLine
Set $\delta_{\mathrm{hist}}\leftarrow\delta/2$ and
$\delta_B\leftarrow\delta/2$\;

\BlankLine
\textbf{Step 1: Estimate the history Gram matrices.}

Set $\eta_{\mathrm{hist}}
    \leftarrow
    {\min\{\epsilon_G,\epsilon_{\widehat G}\}}/{s^2},
    C_{\mathrm{hist}}
    \leftarrow
    s(2s+1),
    \delta_{\mathrm{hist}}^{\mathrm{pair}}
    \leftarrow
    {\delta_{\mathrm{hist}}}/{C_{\mathrm{hist}}}.$

Set $K_{\mathrm{hist}}
    \leftarrow
    \left\lceil
        8\log\left({2}/{\delta_{\mathrm{hist}}^{\mathrm{pair}}}\right)
    \right\rceil,
    S_{\mathrm{hist}}
    \leftarrow
    \left\lceil
        {8\nu_{\Omega}^2}/{\eta_{\mathrm{hist}}^2}
    \right\rceil,
    N_{\mathrm{hist}}^{\mathrm{pair}}
    \leftarrow
    K_{\mathrm{hist}}S_{\mathrm{hist}}.$

\For{$0\leq p\leq q\leq 2s-1$}{
    $\mathcal D_{p,q}^{\mathrm{hist}}
    \leftarrow
    \textsc{KrylovShadowSampling}
    (U_H,\ket{\psi},p,q,(\mathcal U,\mathcal M),
    N_{\mathrm{hist}}^{\mathrm{pair}})$\;

    $\{\widetilde c_{p,q}\}
    \leftarrow
    \textsc{KrylovShadowPostProcessing}
    (\mathcal D_{p,q}^{\mathrm{hist}},\{\Omega_{\mathcal O}\},
    \nu_{\Omega}^2,\eta_{\mathrm{hist}},
    \delta_{\mathrm{hist}}^{\mathrm{pair}})$\;

    \If{$p<q$}{
        $\widetilde c_{q,p}
        \leftarrow
        \widetilde c_{p,q}^{*}$\;
    }
}

\For{$a,b=0,\ldots,s-1$}{
    $(\widetilde G_s^{\mathcal O})_{a,b}
    \leftarrow
    \displaystyle\sum_{\ell=0}^{s-1}
    \widetilde c_{a+\ell,b+\ell}$;
    $(\widetilde{\widehat G}_s^{\mathcal O})_{a,b}
    \leftarrow
    \displaystyle\sum_{\ell=0}^{s-1}
    \widetilde c_{a+\ell,b+\ell+1}$\;
}

\textbf{Step 2: Reconstruct the effective history-space dynamics.}

\If{$\widetilde G_s^{\mathcal O}$ is singular}{
    \Return{\textnormal{failure}}\;
}

$\widetilde J_s^{\mathcal O}
\leftarrow
(\widetilde G_s^{\mathcal O})^{-1}
\widetilde{\widehat G}_s^{\mathcal O}$\;

\BlankLine
\textbf{Step 3: Compute the observable-aware Krylov coefficients.}

$\widetilde{\boldsymbol{\lambda}}^{\mathcal O}
\leftarrow
P(\widetilde J_s^{\mathcal O})\mathbf e_0$,
where
$\mathbf e_0=(1,0,\ldots,0)^{\mathsf T}\in\mathbb C^s$\;

\BlankLine
\textbf{Step 4: Estimate the reduced observable Krylov matrices.}

Set $C_B
    \leftarrow
    {s(s+1)}/{2}, 
    \quad
    \delta_B^{\mathrm{pair}}
    \leftarrow
    {\delta_B}/{C_B}.$
    
Set $K_B
    \leftarrow
    \left\lceil
        8\log\left({2M}/{\delta_B^{\mathrm{pair}}}\right)
    \right\rceil,
    S_B
    \leftarrow
    \left\lceil
        {8\nu_O^2}/{\epsilon_S^2}
    \right\rceil,
    N_B^{\mathrm{pair}}
    \leftarrow
    K_BS_B.$

\For{$0\leq a\leq b\leq s-1$}{
    $\mathcal D_{a,b}^{B}
    \leftarrow
    \textsc{KrylovShadowSampling}
    (U_H,\ket{\psi},a,b,(\mathcal U,\mathcal M),N_B^{\mathrm{pair}})$\;

    $\{\widetilde\mu_{i,a,b}\}_{i=1}^{M}
    \leftarrow
    \textsc{KrylovShadowPostProcessing}
    (\mathcal D_{a,b}^{B},\{O_i\}_{i=1}^{M},
    \nu_O^2,\epsilon_S,\delta_B^{\mathrm{pair}})$\;

    \For{$i=1,\ldots,M$}{
        \eIf{$a=b$}{
            $(\widetilde B_{s,i})_{a,a}
            \leftarrow
            \operatorname{Re}\widetilde\mu_{i,a,a}$\;
        }{
            $(\widetilde B_{s,i})_{a,b}
            \leftarrow
            \widetilde\mu_{i,a,b}$;
            $(\widetilde B_{s,i})_{b,a}
            \leftarrow
            \widetilde\mu_{i,a,b}^{*}$\;
        }
    }
}

\textbf{Step 5: Compute the observable estimates.}

\For{$i=1,\ldots,M$}{
    $\widetilde y_i
    \leftarrow
    (\widetilde{\boldsymbol{\lambda}}^{\mathcal O})^\dagger
    \widetilde B_{s,i}
    \widetilde{\boldsymbol{\lambda}}^{\mathcal O}$\;
}
\BlankLine
\Return{$\{\widetilde y_i\}_{i=1}^{M}$}\;
\end{algorithm}

\begin{remark}[Classical shadows versus direct estimation]
\label{rem:observable-aware-direct-estimation}
Algorithm~\ref{alg:observable-aware-qsvt} uses classical shadows to estimate the cross moments required to construct the history Gram matrices and the reduced observable matrices,
\[
    \bra{\psi}H^p\Omega_{\mathcal O}H^q\ket{\psi}
    \qquad\text{and}\qquad
    \bra{\psi}H^aO_iH^b\ket{\psi}.
\]
This use of classical shadows is not intrinsic to the observable-aware reduction.  If suitable coherent access to the observables is available, for example through block encodings, these quantities may instead be estimated directly by Hadamard-test-type procedures.

There is a tradeoff between the two implementations.  Direct estimation avoids the shadow-norm and classical post-processing overhead and can be more sample-efficient when only a small number of observables are needed or when the efficient observable-specific measurements are available. However, when the observables are estimated separately, its sample complexity typically scales linearly with the size of the observable family and requires stronger coherent access to the observables. In contrast, classical shadows allow the same quantum data to be reused to estimate many observables simultaneously, yielding only logarithmic dependence on the size of the observable family in the sample complexity, at the cost of the corresponding shadow-norm and classical-processing overheads.
\end{remark}

\subsection{Exact observable reconstruction}
\label{subsec:observable-aware-correctness}

\begin{lemma}[Exact observable reconstruction]
\label{lem:observable-aware-reconstruction}
For every polynomial $P$ and every $i\in[M]$,
\begin{equation}
    \bra{\psi}
        P(H)^\dagger O_iP(H)
    \ket{\psi}
    =
    \left(\boldsymbol{\lambda}^{\mathcal O}\right)^\dagger
    B_{s,i}
    \boldsymbol{\lambda}^{\mathcal O},
    \label{eq:observable-aware-output}
\end{equation}
where
\begin{equation}
    \boldsymbol{\lambda}^{\mathcal O}
    =
    P(J_s^{\mathcal O})\mathbf e_0.
\end{equation}
\end{lemma}
\begin{proof}
By Corollary~\ref{cor:history-space-isomorphism} and
Eq.~\eqref{eq:history-effective-dynamics}, $J_s^{\mathcal O}$ is the
matrix representation of the induced Hamiltonian action in the history
basis. Therefore, for every $k\geq0$,
\begin{equation}
    F_s^{\mathcal O}(H^k\ket{\psi})
    =
    A_s^{\mathcal O}
    \left(J_s^{\mathcal O}\right)^k
    \mathbf e_0.
\end{equation}
By linearity, for
$P(x)=\sum_{k=0}^{d}p_kx^k$,
\begin{equation}
    F_s^{\mathcal O}(P(H)\ket{\psi})
    =
    A_s^{\mathcal O}
    P(J_s^{\mathcal O})\mathbf e_0
    =
    A_s^{\mathcal O}\boldsymbol{\lambda}^{\mathcal O}.
    \label{eq:observable-aware-history-polynomial}
\end{equation}
On the other hand, by the definitions of $A_s$ and
$A_s^{\mathcal O}$,
\begin{equation}
    F_s^{\mathcal O}
    \left(A_s\boldsymbol{\lambda}^{\mathcal O}\right)
    =
    A_s^{\mathcal O}\boldsymbol{\lambda}^{\mathcal O}.
\end{equation}
Combining the last two equations and using
Lemma~\ref{lem:finite-history-kernel} gives
\begin{equation}
    P(H)\ket{\psi}
    -
    A_s\boldsymbol{\lambda}^{\mathcal O}
    \in
    \ker F_s^{\mathcal O}
    =
    K^{\mathcal O}.
\end{equation}
Hence there exists $\ket{\eta_P}\in K^{\mathcal O}$ such that
\begin{equation}
    P(H)\ket{\psi}
    =
    A_s\boldsymbol{\lambda}^{\mathcal O}
    +
    \ket{\eta_P}.
\end{equation}
By Eq.~\eqref{eq:observable-invisible-annihilation},
$O_i\ket{\eta_P}=0$ for every $i\in[M]$; since $O_i$ is Hermitian,
$\bra{\eta_P}O_i=0$ as well. Therefore
\begin{align}
    \bra{\psi}P(H)^\dagger O_iP(H)\ket{\psi}
    &=
    \left(\boldsymbol{\lambda}^{\mathcal O}\right)^\dagger
    A_s^\dagger O_iA_s
    \boldsymbol{\lambda}^{\mathcal O}
    \\
    &=
    \left(\boldsymbol{\lambda}^{\mathcal O}\right)^\dagger
    B_{s,i}
    \boldsymbol{\lambda}^{\mathcal O},
\end{align}
which proves the claim.
\end{proof}

\subsection{End-to-end error analysis}
\label{subsec:observable-aware-error}

\paragraph{Technical ingredients.}
We first record two perturbation results that will be used in the end-to-end analysis below.

\begin{lemma}[Spectral structure of the effective dynamics]
\label{lem:observable-aware-spectral-structure}
Under the standing assumptions of this section, suppose that $H=H^\dagger$ and $\|H\|_2\leq1$. Although $J_s^{\mathcal O}$ need not itself be Hermitian because the history basis is generally nonorthogonal, it is similar to a Hermitian contraction. Equivalently, there exist an invertible matrix $X_{\mathcal O}$ and a real diagonal matrix $\Lambda_{\mathcal O}$ such that
\begin{equation}
    J_s^{\mathcal O}
    =
    X_{\mathcal O}
    \Lambda_{\mathcal O}
    X_{\mathcal O}^{-1},
    \qquad
    \operatorname{spec}(\Lambda_{\mathcal O})
    \subseteq[-1,1].
    \label{eq:observable-aware-spectral-decomposition}
\end{equation}
Fix such a diagonalization and define its spectral conditioning by
\begin{equation}
    \kappa_{\mathcal O}
    :=
    \|X_{\mathcal O}\|_2
    \|X_{\mathcal O}^{-1}\|_2.
    \label{eq:observable-aware-spectral-conditioning}
\end{equation}
\end{lemma}
The proof is deferred to Appendix~\ref{app:observable-aware-auxiliary-results}.
\begin{lemma}[Polynomial stability under perturbations]
\label{lem:observable-aware-polynomial-stability}
Let
\begin{equation}
    P(x)=\sum_{k=0}^{d}p_kx^k,
    \qquad
    B_P:=\max_{x\in[-1,1]}|P(x)|.
\end{equation}
There exist universal constants $c,C>0$ such that, for any perturbation $E$ satisfying
\begin{equation}
    \|E\|_2
    \leq
    \frac{c}{\kappa_{\mathcal O}d^2},
    \label{eq:observable-aware-degree-stability-condition}
\end{equation}
we have
\begin{equation}
    \left\|
        P(J_s^{\mathcal O}+E)
        -
        P(J_s^{\mathcal O})
    \right\|_2
    \leq
    C\kappa_{\mathcal O}^2
    B_P d^4
    \|E\|_2.
    \label{eq:observable-aware-degree-stability}
\end{equation}
\end{lemma}
The proof is deferred to Appendix~\ref{app:observable-aware-auxiliary-results}.
\paragraph{Error parameters and target quantity.}
Throughout this subsection, the polynomial $P$ is fixed, and we analyze the estimation error relative to
\begin{equation}
    y_i
    :=
    \bra{\psi}
        P(H)^\dagger O_iP(H)
    \ket{\psi}.
    \label{eq:observable-aware-exact-y}
\end{equation}
Define
\begin{equation}
    \gamma_{\mathcal O}
    :=
    \lambda_{\min}
    \left(
        G_s^{\mathcal O}
    \right)
    >
    0,
    \label{eq:observable-aware-gamma}
\end{equation}
and assume
\begin{equation}
    \epsilon_G
    <
    \gamma_{\mathcal O}.
    \label{eq:observable-aware-invertibility-condition}
\end{equation}
Define
\begin{equation}
    \epsilon_J
    :=
    \frac{
        \epsilon_{\widehat G}
        +
        \epsilon_G
        \left\|
            J_s^{\mathcal O}
        \right\|_2
    }{
        \gamma_{\mathcal O}-\epsilon_G
    }.
    \label{eq:observable-aware-epsilon-J}
\end{equation}
For
\begin{equation}
    P(x)
    =
    \sum_{k=0}^{d}p_kx^k,
\end{equation}
define
\begin{equation}
    B_P
    :=
    \max_{x\in[-1,1]}|P(x)|.
    \label{eq:observable-aware-BP}
\end{equation}
Assume in addition that
\begin{equation}
    \epsilon_J
    \leq
    \frac{c}{
        \kappa_{\mathcal O}d^2
    },
    \label{eq:observable-aware-polynomial-stability-condition}
\end{equation}
where $c>0$ is the universal constant in
Lemma~\ref{lem:observable-aware-polynomial-stability}. Set
\begin{equation}
    \epsilon_\lambda
    :=
    C\kappa_{\mathcal O}^2
    B_P d^4
    \epsilon_J,
    \label{eq:observable-aware-epsilon-lambda}
\end{equation}
where $C>0$ is the universal constant in Lemma~\ref{lem:observable-aware-polynomial-stability}.
Finally, define
\begin{equation}
    \theta_{\mathcal O}
    :=
    \left\|
        \boldsymbol{\lambda}^{\mathcal O}
    \right\|_2,
    \qquad
    \beta_{\mathcal O}
    :=
    \max_{i\in[M]}
    \left\|
        B_{s,i}
    \right\|_2.
    \label{eq:observable-aware-theta-beta}
\end{equation}
By Lemma~\ref{lem:observable-aware-spectral-structure},
\begin{equation}
    \theta_{\mathcal O}
    =
    \|P(J_s^{\mathcal O})\mathbf e_0\|_2
    \leq
    \kappa_{\mathcal O} B_P.
    \label{eq:observable-aware-theta-bound}
\end{equation}

\begin{theorem}[Observable-Aware S-QSVT]
\label{thm:observable-aware-end-to-end}
Under the standing assumptions of this section, let $P$ be a fixed degree-$d$ polynomial, and let $\epsilon_O>0$ and $\delta\in(0,1)$. Choose $\epsilon_G,\epsilon_{\widehat G},\epsilon_S>0$, and define $\epsilon_J$ and $\epsilon_\lambda$ by Eqs.~\eqref{eq:observable-aware-epsilon-J} and~\eqref{eq:observable-aware-epsilon-lambda}, respectively. Assume that $U_H$ is an exact $(1,l,0)$-block encoding of $H$, that $\epsilon_G<\gamma_{\mathcal O},\epsilon_J\leq\frac{c}{\kappa_{\mathcal O}d^2},$ and that
\begin{equation}
    s\epsilon_S
    \left(
        \theta_{\mathcal O}+\epsilon_\lambda
    \right)^2
    +
    \beta_{\mathcal O}
    \left(
        2\theta_{\mathcal O}\epsilon_\lambda
        +
        \epsilon_\lambda^2
    \right)
    \leq
    \epsilon_O.
    \label{eq:observable-aware-target-condition}
\end{equation}
Then, with probability at least $1-\delta$, Algorithm~\ref{alg:observable-aware-qsvt} does not return failure and simultaneously satisfies
\begin{equation}
    |\widetilde y_i-y_i|
    \leq
    \epsilon_O
    \qquad
    \text{for all }i\in[M].
\end{equation}

Let $N_{\mathrm{hist}}$ and $N_B$ denote the numbers of copies of $\ket{\psi}$ used in Steps~1 and~4, respectively. Then
\begin{align}
    N_{\mathrm{hist}}
    &=
    O\!\left(
        \frac{
            s^6\nu_{\Omega}^2
        }{
            \min\{\epsilon_G,\epsilon_{\widehat G}\}^2
        }
        \log\!\left(
            \frac{s^2}{\delta}
        \right)
    \right),
    \\
    N_B
    &=
    O\!\left(
        \frac{
            s^2\nu_O^2
        }{
            \epsilon_S^2
        }
        \log\!\left(
            \frac{Ms^2}{\delta}
        \right)
    \right).
\end{align}
Thus, the total number of copies of the input state consumed by Algorithm~\ref{alg:observable-aware-qsvt} is
\begin{equation}
    N_{\mathrm{OA}}
    :=
    N_{\mathrm{hist}}+N_B
    =
    O\!\left(
        \frac{
            s^6\nu_{\Omega}^2
        }{
            \min\{\epsilon_G,\epsilon_{\widehat G}\}^2
        }
        \log\!\left(
            \frac{s^2}{\delta}
        \right)
        +
        \frac{
            s^2\nu_O^2
        }{
            \epsilon_S^2
        }
        \log\!\left(
            \frac{Ms^2}{\delta}
        \right)
    \right).
    \label{eq:observable-aware-total-samples}
\end{equation}

The total classical computational cost is
\begin{align}
    T_{\mathrm{OA}}^{\mathrm{cl}}
    =
    O\!\Bigl(
        N_{\mathrm{hist}}
        \bigl[
            T_{\mathrm{eval}}
            (n,\{\Omega_{\mathcal O}\})+1
        \bigr]
        +
        MN_B
        \bigl[
            T_{\mathrm{eval}}(n,\mathcal O)+1
        \bigr]
        +
        s^3
        +
        ds^2
        +
        Ms^2
    \Bigr),
    \label{eq:observable-aware-classical-complexity}
\end{align}
where the two evaluation costs are defined in Eqs.~\eqref{eq:history-evaluation-cost} and~\eqref{eq:shadow-evaluation-cost}.

Finally, if $D_H$ and $D_{\mathrm{sh}}$ denote the circuit depths of one application of $U_H$ or controlled-$U_H$ and one shadow-measurement layer, respectively, then
\begin{equation}
    D_{\mathrm{OA}}
    =
    O\!\left(
        sD_H+D_{\mathrm{sh}}
    \right).
    \label{eq:observable-aware-depth}
\end{equation}
\end{theorem}

\begin{proof}[Proof sketch]
The proof follows the error-propagation chain
\[
    (\widetilde G_s^{\mathcal O},
     \widetilde{\widehat G}_s^{\mathcal O})
    \longrightarrow
    \widetilde J_s^{\mathcal O}
    \longrightarrow
    P(\widetilde J_s^{\mathcal O})
    \longrightarrow
    \widetilde{\boldsymbol{\lambda}}^{\mathcal O}
    \longrightarrow
    \widetilde y_i.
\]
First, perturbation of the estimated history Gram matrices gives
\[
    \|\widetilde J_s^{\mathcal O}-J_s^{\mathcal O}\|_2
    \leq \epsilon_J.
\]
The key technical step is then to control how this perturbation propagates through the polynomial $P$.  Lemmas~\ref{lem:observable-aware-spectral-structure}
and~\ref{lem:observable-aware-polynomial-stability} show that $J_s^{\mathcal O}$ has controlled spectral structure and, consequently, that
\[
    \|P(\widetilde J_s^{\mathcal O})
      -P(J_s^{\mathcal O})\|_2
\]
is controlled by $\|\widetilde J_s^{\mathcal O}-J_s^{\mathcal O}\|_2$. This yields the coefficient error
\[
    \|
        \widetilde{\boldsymbol{\lambda}}^{\mathcal O}
        -
        \boldsymbol{\lambda}^{\mathcal O}
    \|_2
    \leq
    \epsilon_\lambda.
\]
Finally, combining this coefficient error with the estimation error of $B_{s,i}$ and expanding the difference of the corresponding quadratic forms gives the claimed bound on $|\widetilde y_i-y_i|$. For the classical cost, Step~1 processes the $N_{\mathrm{hist}}$ stored samples against the singleton observable family $\{\Omega_{\mathcal O}\}$, while Step~4 processes the $N_B$ stored samples for all $M$ queried observables. The remaining classical work consists of assembling the estimated matrices, reconstructing $\widetilde J_s^{\mathcal O}$, evaluating $P(\widetilde J_s^{\mathcal O})\mathbf e_0$, and forming the final quadratic forms. Evaluating the polynomial-vector product by Horner's rule gives the $O(ds^2)$ contribution in Eq.~\eqref{eq:observable-aware-classical-complexity}. The complete proof is given in Appendix~\ref{app:observable-aware-end-to-end-proof}.
\end{proof}

\begin{corollary}[Degree dependence of the resources]
\label{cor:observable-aware-degree-dependence}
Suppose that 
$\gamma_{\mathcal O}^{-1}$,
$\|J_s^{\mathcal O}\|_2$,
$\kappa_{\mathcal O}$,
$B_P$,
$\theta_{\mathcal O}$, and
$\beta_{\mathcal O}$ 
are bounded by constants. Then the internal precisions can be chosen as
\begin{equation}
    \epsilon_G,\epsilon_{\widehat G}
    =
    \Theta\!\left(
        \frac{\epsilon_O}{d^4}
    \right),
    \qquad
    \epsilon_S
    =
    \Theta\!\left(
        \frac{\epsilon_O}{s}
    \right),
\end{equation}
which gives
\begin{equation}
    N_{\mathrm{OA}}
    =
    O\!\left(
        \frac{s^6\nu_\Omega^2d^8}{\epsilon_O^2}
        \log\!\frac{s^2}{\delta}
        +
        \frac{s^4\nu_O^2}{\epsilon_O^2}
        \log\!\frac{Ms^2}{\delta}
    \right).
\end{equation}
The classical computational cost inherits this $d^8$ dependence through $N_{\mathrm{hist}}$ and additionally contains the $O(ds^2)$ cost of evaluating $P(\widetilde J_s^{\mathcal O})\mathbf e_0$. In contrast, the maximum quantum circuit depth $D_{\mathrm{OA}}=O(sD_H+D_{\mathrm{sh}})$ is independent of $d$.
\end{corollary}

\begin{remark}[Dependence on the block-encoding normalization]
\label{rem:algo3-alpha-normalization}
For an exact $(\alpha,l,0)$-block encoding, temporarily write $H_\alpha:=H/\alpha$ and $P_\alpha(x):=P(\alpha x)$, so that the same unitary $U_H$ is an exact $(1,l,0)$-block encoding of $H_\alpha$, as in Sec.~\ref{sec:section3}. The state-Krylov space, the observable-invisible subspace, and hence the observable-relevant dimension $s$ are unchanged under this rescaling. Let
\[
    D_\alpha
    :=
    \operatorname{diag}
    (1,\alpha^{-1},\ldots,\alpha^{-(s-1)}),
    \qquad
    L_\alpha
    :=
    I_M\otimes D_\alpha\otimes I_{\mathcal H}.
\]
The history-basis matrices of the rescaled problem satisfy
\[
    A_{s,\alpha}^{\mathcal O}
    =
    L_\alpha A_s^{\mathcal O}D_\alpha,
    \qquad
    A_{s,+,\alpha}^{\mathcal O}
    =
    \alpha^{-1}
    L_\alpha A_{s,+}^{\mathcal O}D_\alpha.
\]
Therefore,
\[
    G_{s,\alpha}^{\mathcal O}
    =
    D_\alpha
    (A_s^{\mathcal O})^\dagger
    L_\alpha^2
    A_s^{\mathcal O}
    D_\alpha,
    \qquad
    \widehat G_{s,\alpha}^{\mathcal O}
    =
    \alpha^{-1}
    D_\alpha
    (A_s^{\mathcal O})^\dagger
    L_\alpha^2
    A_{s,+}^{\mathcal O}
    D_\alpha,
\]
and the effective dynamics transforms as
\[
    J_{s,\alpha}^{\mathcal O}
    =
    \alpha^{-1}
    D_\alpha^{-1}
    J_s^{\mathcal O}
    D_\alpha.
\]
The reduced observable Krylov matrices satisfy $B_{s,i}^{(\alpha)}=D_\alpha B_{s,i}D_\alpha$. Moreover,
\[
    \boldsymbol{\lambda}_{\alpha}^{\mathcal O}
    =
    P_\alpha(J_{s,\alpha}^{\mathcal O})\mathbf e_0
    =
    D_\alpha^{-1}\boldsymbol{\lambda}^{\mathcal O}.
\]
Hence, for $\alpha\geq1$, $\theta_{\mathcal O,\alpha} \leq\alpha^{s-1}\theta_{\mathcal O}$ and $\beta_{\mathcal O,\alpha}\leq\beta_{\mathcal O}$. Writing $\gamma_{\mathcal O,\alpha} := \lambda_{\min}(G_{s,\alpha}^{\mathcal O})$, the factorization above implies
\[
    \gamma_{\mathcal O,\alpha}
    \geq
    \alpha^{-4(s-1)}
    \gamma_{\mathcal O}.
\]
Furthermore, if $J_s^{\mathcal O} =X_{\mathcal O}\Lambda_{\mathcal O}X_{\mathcal O}^{-1}$, then
\[
    J_{s,\alpha}^{\mathcal O}
    =
    (D_\alpha^{-1}X_{\mathcal O})
    \frac{\Lambda_{\mathcal O}}{\alpha}
    (D_\alpha^{-1}X_{\mathcal O})^{-1}.
\]
Thus one may choose $X_{\mathcal O,\alpha}=D_\alpha^{-1}X_{\mathcal O}$, which gives $\kappa_{\mathcal O,\alpha} \leq\alpha^{s-1}\kappa_{\mathcal O}$. The polynomial-dependent bound is correspondingly replaced by
\[
    B_{P_\alpha}
    :=
    \max_{x\in[-1,1]}
    |P(\alpha x)|.
\]
Accordingly, all history-space perturbation parameters are evaluated for the rescaled instance. In particular,
\[
    \epsilon_{J,\alpha}
    =
    \frac{
        \epsilon_{\widehat G,\alpha}
        +
        \epsilon_{G,\alpha}
        \|J_{s,\alpha}^{\mathcal O}\|_2
    }{
        \gamma_{\mathcal O,\alpha}
        -
        \epsilon_{G,\alpha}
    },
    \qquad
    \epsilon_{\lambda,\alpha}
    =
    C\kappa_{\mathcal O,\alpha}^{\,2}
    B_{P_\alpha}d^4
    \epsilon_{J,\alpha},
\]
with $\epsilon_{G,\alpha}<\gamma_{\mathcal O,\alpha}$ and $\epsilon_{J,\alpha} \leq c/(\kappa_{\mathcal O,\alpha}d^2)$. For the observable-matrix contribution, using the bounds above define the sufficient error margin
\[
    \Delta_\alpha
    :=
    \epsilon_O
    -
    \beta_{\mathcal O}
    \left(
        2\alpha^{s-1}\theta_{\mathcal O}
        \epsilon_{\lambda,\alpha}
        +
        \epsilon_{\lambda,\alpha}^2
    \right),
\]
and suppose that $\Delta_\alpha>0$. It is then sufficient to choose
\[
    \epsilon_{S,\alpha}
    =
    \frac{
        \Delta_\alpha
    }{
        s
        (\alpha^{s-1}\theta_{\mathcal O}
        +\epsilon_{\lambda,\alpha})^2
    }.
\]
The resulting observable-matrix sample complexity is
\[
    N_B^{(\alpha)}
    =
    O\!\left(
        \frac{
            s^4
            (\alpha^{s-1}\theta_{\mathcal O}
            +\epsilon_{\lambda,\alpha})^4
            \nu_O^2
        }{
            \Delta_\alpha^2
        }
        \log\!\left(
            \frac{Ms^2}{\delta}
        \right)
    \right).
\]
The history-matrix contribution remains
\[
    N_{\mathrm{hist}}^{(\alpha)}
    =
    O\!\left(
        \frac{
            s^6\nu_{\Omega}^2
        }{
            \min\{
                \epsilon_{G,\alpha},
                \epsilon_{\widehat G,\alpha}
            \}^2
        }
        \log\!\left(
            \frac{s^2}{\delta}
        \right)
    \right),
\]
so that $N_{\mathrm{OA}}^{(\alpha)} =N_{\mathrm{hist}}^{(\alpha)}+N_B^{(\alpha)}$. The corresponding classical computational cost is
\[
    T_{\mathrm{OA},\alpha}^{\mathrm{cl}}
    =
    O\!\left(
        N_{\mathrm{hist}}^{(\alpha)}
        [T_{\mathrm{eval}}(n,\{\Omega_{\mathcal O}\})+1]
        +
        M N_B^{(\alpha)}
        [T_{\mathrm{eval}}(n,\mathcal O)+1]
        +
        s^3+ds^2+Ms^2
    \right).
\]
Thus, when $\Delta_\alpha$ is bounded away from zero and $\epsilon_{\lambda,\alpha}$ remains controlled, the observable-matrix contribution and its sample-dependent classical cost may incur an $\alpha^{4(s-1)}$ worst-case amplification. The history-matrix contribution can have additional normalization dependence through $\gamma_{\mathcal O,\alpha}$, $\kappa_{\mathcal O,\alpha}$, $B_{P_\alpha}$, and the required history-matrix precisions. The number of sequential block-encoding layers is unchanged, so the maximum circuit-depth scaling remains $O(sD_H+D_{\mathrm{sh}})$.
\end{remark}

\bibliographystyle{alphaurl}  
\bibliography{references, Qalg, nai}

\newpage
\appendix
\appendixsection{Technical Details of State-Aware QSVT}
\subsection{Proof of Equivalence between QR-Lanczos basis}
\label{section:proof_qr_factor_lanczos_equivalent}
\begin{claim}
    \label{claim:QR_Ar_has_Vr}
    For $r \in \mathbb{N}_{+}$, let $A_r$ and $V_r$ be the Krylov matrix and the Lanczos basis matrix of $\mathcal{K}_r(H, \ket{\psi})$ as defined in \Cref{eq:Ar_definition_YC} and (\ref{eq:Vr_definition_YC}), and let
    \begin{equation}
        A_r = Q_r U_r
    \end{equation}
    be a QR decomposition of $A_r$, where $Q_r$ is orthogonal and $U_r$ is upper triangular.
    Then there is a diagonal unitary matrix $D$ 
    \begin{equation}
        D = \operatorname{diag}(e^{i\theta_0}, \dots, e^{i \theta_{r-1}}), \theta_j \in \mathbb{R},  j = 0, \dots r-1,
    \end{equation}
    such that 
    \begin{equation}
        Q_r = V_r D
    \end{equation}
\end{claim}
\begin{proof}
    For QR decomposition of $A_r$
    \begin{equation}
        Ar = Q_r U_r,
    \end{equation}
    write
    \begin{equation}
        \begin{aligned}
        Q_r = \biggl[ \begin{matrix}  |q_0\rangle \biggm| |q_1\rangle \biggm| \cdots \biggm| |q_{r-1}\rangle \end{matrix} \biggr]\\
        V_r = \biggl[ \begin{matrix}  |v_0\rangle \biggm| |v_1\rangle \biggm| \cdots \biggm| |v_{r-1}\rangle \end{matrix} \biggr]
        \end{aligned}
    \end{equation}
    where $|q_0\rangle = |v_0\rangle = |\psi\rangle$.
    Since $A_r$ has full column rank $r$, the upper triangular factor $U_r\in\mathbb{C}^{r\times r}$ is nonsingular.
    Thus, for every $k=1,\ldots,r$, the first $k$ columns of $A_r$ and $Q_r$ span the same subspace and since $\{\ket{v_i}\}_{i=0}^{k-1}$ spans $\mathcal{K}_k(H, \ket{\psi})$,
    \begin{equation}
        \operatorname{span}\{\ket{q_i}\}_{i=0}^{k-1}
        = \mathcal{K}_k(H,\ket{\psi}).
    \end{equation}
    Moreover, the full column rank of $A_r$ implies
    $\dim\mathcal{K}_k(H,\ket{\psi})=k$ for
    $k=1,\ldots,r$.

    We prove by induction on $k$ that
    \begin{equation}
        Q_k = V_k D_k,
    \end{equation}
    where $Q_k$ and $V_k$ consist of the first $k$ columns
    of $Q_r$ and $V_r$, respectively, and $D_k$ is a
    diagonal unitary matrix.

    For the base case $k=1$, the first column of
    $A_r=Q_rU_r$ gives
    \begin{equation}
        \ket{\psi}=u_{00}\ket{q_0}.
    \end{equation}
    Since $\ket{\psi}$ and $\ket{q_0}$ are normalized,
    $|u_{00}|=1$. As $\ket{v_0}=\ket{\psi}$, there exists
    $\theta_0\in\mathbb{R}$ such that
    \begin{equation}
        \ket{q_0}=e^{i\theta_0}\ket{v_0}.
    \end{equation}
    Hence, $Q_1=V_1D_1$ with $D_1=[e^{i\theta_0}]$.

    For the induction step, let $1\leq k<r$ and assume
    that $Q_k=V_kD_k$ for some diagonal unitary matrix
    $D_k$.
    Both $\ket{q_k}$ and $\ket{v_k}$ belong to
    $\mathcal{K}_{k+1}(H,\ket{\psi})$.
    Since the columns of $Q_r$ and $V_r$ are orthonormal,
    both vectors are orthogonal to
    $\mathcal{K}_k(H,\ket{\psi})$. Therefore,
    \begin{equation}
        \ket{q_k},\ket{v_k}
        \in W_k :=
        \mathcal{K}_{k+1}(H,\ket{\psi})
        \cap \mathcal{K}_k(H,\ket{\psi})^\perp.
    \end{equation}
    Because
    $\mathcal{K}_k(H,\ket{\psi})
    \subseteq\mathcal{K}_{k+1}(H,\ket{\psi})$,
    \begin{equation}
        \dim W_k
        = \dim\mathcal{K}_{k+1}(H,\ket{\psi})
          - \dim\mathcal{K}_k(H,\ket{\psi})
        = 1.
    \end{equation}
    Thus, $\ket{q_k}$ and $\ket{v_k}$ are unit vectors
    in the same one-dimensional subspace, so there exists
    $\theta_k\in\mathbb{R}$ such that
    \begin{equation}
        \ket{q_k}=e^{i\theta_k}\ket{v_k}.
    \end{equation}
    Defining
    \begin{equation}
        D_{k+1}
        = \begin{bmatrix}
            D_k & 0 \\
            0 & e^{i\theta_k}
          \end{bmatrix},
    \end{equation}
    we obtain $Q_{k+1}=V_{k+1}D_{k+1}$.
    This completes the induction and proves
    $Q_r=V_rD_r$.
\end{proof}

The above proof holds true for all the process using $|\psi\rangle$ as the initial basis vector and finding a new perpendicular orthonormal basis to the previous Krylov space. 

\subsection[Recursive Construction of $U_r$ and $J_r$]{Recursive Construction of $U_r$ and $J_r$}\label{sec:subroutine}

The following algorithm calculates the $U_r$ and $J_r$ in \ref{alg:state-aware-QSVT_YC} simultaneously.

\begin{algorithm}[H]
    \caption{$U_r, J_r$ \textsc{calculation}}
    \label{alg:Ur_Jr_calculation}
    \LinesNumbered
    \KwIn{$\{g_k\}_{k=0}^{2r-1}$, where $g_k = \bra{\psi} H^k\ket{\psi}$}
    \KwOut{$U_r, J_r$}
    $\mathbf{u}_r \leftarrow \{\vec{u}_0\}$ where $\vec{u}_0 := \begin{bmatrix}
        1 &0 & \dots & 0
    \end{bmatrix}^T$\;
    $\boldsymbol{\alpha} \leftarrow \{\alpha_0\}$, $\alpha_0 \leftarrow 1$ \;
    $\boldsymbol{\beta} \leftarrow \varnothing, \beta_{-1} = \beta_0 \leftarrow 0$ \;
    \For{$k \in [0, r-2]$}{
        $\vec{u}_{k+1}, \alpha_k, \beta_{k+1} \leftarrow \leftarrow$ $\textsc{Algorithm} \ref{alg:k_1_col_of_Ur} (\mathbf{u}_r, g_{k+1}, g_{2k+1}, g_{2k+2}, \boldsymbol{\alpha}, \boldsymbol{\beta})$ \; 
        Append $\vec{u}_{k+1}$ to $\mathbf{u}_r$\;
        Append $\alpha_{k}$ to $\boldsymbol{\alpha}$\;
        Append $\beta_{k+1}$ to $\boldsymbol{\beta}$\;
    }
    $U_r \leftarrow \mathbf{u}_r$\;
    $J_r \leftarrow \operatorname{diag}_{0}(\boldsymbol{\alpha}) + \operatorname{diag_{1}}(\boldsymbol{\beta}) + \operatorname{diag}_{-1}(\boldsymbol{\beta})$\;
    \textbf{return } $U_r, J_r$ \;
\end{algorithm}
\Cref{alg:k_1_col_of_Ur} is a subroutine to calculate next column vector $\boldsymbol{u}_{k+1}$, $\alpha_{k}$ and $\beta_{k+1}$ values in $k$-th step.
\subsubsection{\Cref{alg:k_1_col_of_Ur}}
\Cref{alg:k_1_col_of_Ur} computes the $k+1$ column of $U_r$ and $J_r$ matrix with previous computed columns and additional moments.
The goal of this algorithm is to replace the direct Cholesky decomposition to calculate $U_r$ and $J_r$ matrices with a more efficient recursive procedure.
Given $G_r$ and $\hat{G}_r$, one can obtain $U_r$ from Cholesky decomposition of $G_r$ and subsequently $J_r = U_r^{-\dagger} \hat{G}_r U_r^{-1}$.
Standard Cholesky factorization algorithms of $r\times r$ require $O(r^{>2})$ computational complexity\cite{dupont2026improvingmatrixmultiplicationexponent}.
This algorithm achieves $O(r^2)$ computational complexity and includes $J_r$ calculation inside, thereby avoiding the inverse computation of $U_r$.
\label{appendix:efficient_calculation_subroutine}
\begin{algorithm}[H]
    \caption{$k+1$ column calculation of $U_r$}
    \label{alg:k_1_col_of_Ur}
    \LinesNumbered
    \KwIn{$\mathbf{u}_r, g_{k+1}, g_{2k+1}, g_{2k+2}, \boldsymbol{\alpha}, \boldsymbol{\beta}$}
    \KwOut{$\vec{u}_{k+1}, \alpha_k, \beta_{k+1}$}
    ${u}_{0, k+1} \leftarrow g_{k+1}$ \;
    \For{$i \in 1, \dots, k-1$}{
        $u_{i, k+1} = \beta_{i+1} u_{i+1, k} + \alpha_{i} u_{i, k} + \beta_i u_{i-1, k}$
    }
    ${u}_{k, k+1} \leftarrow \frac{1}{u_{k,k}}(g_{2k+1} - \sum_{i=0}^{k-1} u_{i,k} u_{i, k+1})$ \;
    $u_{k+1, k+1} \leftarrow \sqrt{g_{2k +2} - \sum_{i=0}^{k} u_{i, k+1}^2}$ \;
    $\alpha_k \leftarrow \frac{u_{k, k+1}}{u_{k, k}} - \frac{u_{k-1, k}}{u_{k-1, k-1}}$\;
    $\beta_{k+1} \leftarrow \frac{u_{k+1, k+1}}{u_{k, k}}$\;
    $\vec{u}_{k+1} \leftarrow \begin{bmatrix}
        u_{0, k+1}, \dots, u_{k+1, k+1}
    \end{bmatrix}^T$\;
    \textbf{return } $\vec{u}_{k+1}, \alpha_k, \beta_{k+1}$\;
\end{algorithm}

\subsubsection{Correctness of \Cref{alg:k_1_col_of_Ur}}
\label{appendix:efficient_cal_ur_jr}

From $A_r = V_r U_r$, we can derive the next equation.
\begin{equation}
\label{eq:Ar_column_U_r}
    H^k\ket{\psi} = \sum_{i=0}^k u_{i,k} \ket{v_i}
\end{equation}
where $k = 0, \dots, r-1$ and $u_{i,j} = (U_r)_{i+1, j+1}$ for $i, j = 0, \dots, r-1$. 
To align the matrix index with $|v_i\rangle$ basis index, this section uses zero-based indexing.

Since $\{\ket{v_i}\}_{i=0}^{r-1}$ is an orthonormal basis, for indices $p, q$ that satisfy $k = p+q$,
\begin{equation}\label{eq:g_and_uiks}
    g_{p+q} =  \bra{\psi}H^{p+q}\ket{\psi} = \sum_{i=0}^{\min{(p, q)}} u_{i,p} u_{i, q}.
\end{equation}
Specifically, 
\begin{equation}
    \label{eq:u_0_k+1_cal}
    u_{0, k+1} = g_{k+1}.
\end{equation}
Therefore, a single $g_k$ has many different representations with $u_{ij}$ elements. 
For example, if $r\geq 5$
\begin{equation}
    \begin{aligned}
        g_3 &= u_{0,3}= u_{0,1}u_{0,2} + u_{1,1}u_{1,2}\\
        g_4 &= u_{0,4}= u_{0,1}u_{0,3} + u_{1,1}u_{1,3}= u_{0,2}^2 + u_{1,2}^2 + u_{2,2}^2
    \end{aligned}
\end{equation}

Specifically, set $p=k, q=k+1$ Eq.~\eqref{eq:g_and_uiks} yields
\begin{equation}
    g_{2k+1} = \sum_{i=0}^{k}u_{i, k} u_{i, k+1}
\end{equation}
and $p=q=k+1$ yields
\begin{equation}
    g_{2k+2} = \sum_{i=0}^{k+1} u_{i, k+1}^2
\end{equation}

For $u_{k, k+1}$ and $u_{k+1, k+1}$, we have follows.
\begin{align}
    &u_{k, k+1} = \frac{1}{u_{k,k}} (g_{2k+1} - \sum_{i=0}^{k-1} u_{i,k} u_{i, k+1})\label{eq:u_k_k+1_cal},\\
    &u_{k+1 k+1} =  \sqrt{g_{2k+2} - \sum_{i=0}^{k} u_{i,k+1}^2}\label{eq:u_k+1_k+1_cal}
\end{align}

Finally, from Eq (\ref{eq:Ar_column_U_r}),
\begin{align}
    H^{k+1} \ket{\psi} &= \sum_{i=0}^{k+1} u_{i k+1} \ket{v_i}\label{eq:U_r_u_k1}\\
    &= \sum_{i=0}^{k} u_{i,k} H\ket{v_i} = \sum_{i=0}^{k} u_{i,k} (\beta_{i+1} \ket{v_{i+1}} + \alpha_i \ket{v_i} + \beta_{i} \ket{v_{i-1}})
\end{align}
The expansion yields follows for $i = 1, \dots k-1$
\begin{align}
    &u_{i k+1} = \beta_{i+1} u_{i+1 k} + \alpha_i u_{i k} + \beta_i u_{i-1 k}\label{eq:u_i_k+1_cal},\\
    &u_{k, k+1} = \beta_{k} u_{k-1, k} + \alpha_k u_{k, k}\label{eq:alpha_beta_update},\\
    &u_{k+1, k+1} = \beta_{k+1} u_{k, k}\label{eq:bk_original}.
\end{align}
Combine Eq.~\eqref{eq:u_0_k+1_cal}, \eqref{eq:u_i_k+1_cal}, \eqref{eq:u_k_k+1_cal}, and \eqref{eq:u_k+1_k+1_cal} we get the $k+1$ column update routine.
\begin{align}
    u_{0, k+1} &= g_{k+1}\\
    u_{i, k+1} &= \beta_{i+1} u_{i+1 k} + \alpha_i u_{i k} + \beta_i u_{i-1 k}, \, i = 1, \dots, k-1 \\
    u_{k, k+1} &= \frac{1}{u_{k,k}} (g_{2k+1} - \sum_{i=0}^{k-1} u_{i,k} u_{i, k+1})\\
    u_{k+1, k+1} &= \sqrt{g_{2k+2} - \sum_{i=0}^{k} u_{i,k+1}^2}\\
\end{align}
Lastly, Eq.~\eqref{eq:alpha_beta_update} and Eq.~\eqref{eq:bk_original} yield
\begin{equation} \label{eq:jr_coeff_calculation}
    \begin{aligned}
        &\alpha_{k} = \frac{u_{k k+1}}{u_{k k}} - \frac{u_{k-1 k}}{u_{k-1 k-1}}\\
        &\beta_{k+1} = \frac{u_{k+1 k+1}}{u_{kk}}\\
    \end{aligned}.
\end{equation}
These are non-zero element in $k+1$-th column of $J_r$.
Therefore, repeating \Cref{alg:k_1_col_of_Ur} $r-1$ times with $\{g_k\}_{k=0}^{2r-1}$ yields $U_r$ and $J_r$.
\subsubsection{Computational complexity}
To compute the $k+1$ column of $U_r$, the algorithm uses four update equations Eq.~\eqref{eq:u_0_k+1_cal}, \eqref{eq:u_i_k+1_cal}, \eqref{eq:u_k_k+1_cal} and \eqref{eq:u_k+1_k+1_cal}.
For the $k+1$ column of $J_r$, it uses Eq.~\eqref{eq:jr_coeff_calculation}.

The update $u_{0, k+1} \leftarrow g_{k+1}$ requires $O(1)$ operations. 
Each evaluation of Eq.~\eqref{eq:u_i_k+1_cal} requires $O(1)$ operations and it is evaluated $k-1$ times resulting in a total costs of $O(k)$.
Eq.~\eqref{eq:u_k_k+1_cal} contains a summation over $k$ terms and therefore $O(k)$. 
The same holds for Eq.~\eqref{eq:u_k+1_k+1_cal}.
Each computation of $\alpha_k, \beta_{k+1}$ requires $O(1)$.
Therefore, \Cref{alg:k_1_col_of_Ur} requires total $O(k)$ computational complexity to compute $k+1$ column of $U_r$ and $J_r$.

\subsection{Auxiliary  Results for \Cref{thm:main_alg1_formal_YC}}
\subsubsection{Cholesky Factorization Bound}

Cholesky factorization expresses a Hermitian positive definite matrix as the product of an upper triangular matrix and its conjugate transpose. Requiring the triangular factor to have positive real diagonal entries makes this factorization unique.

\begin{theorem}[Cholesky factorization \cite{dramc_doi:10.1137/S0895479893244717}]
    Let $G\in\mathbb{C}^{n\times n}$ be Hermitian
    positive definite. Then there exists a unique
    upper triangular matrix $U$ with positive real
    diagonal entries such that
    \begin{equation}
        G=U^\dagger U.
    \end{equation}
    The matrix $U$ is called the Cholesky factor of $G$.
    Its entries satisfy
    \begin{equation}
        G_{ij}
        =\sum_{k=1}^{\min(i,j)}
        {U_{ki}^\ast}\,U_{kj},
        \qquad 1\leq i,j\leq n.
    \end{equation}
\end{theorem}

\begin{claim}[Absolute perturbation of Cholesky factors]
    \label{claim:Cholesky_Error_Propagation}
    Let $G\in\mathbb{C}^{n\times n}$ be Hermitian
    positive definite, and let $\tilde{G}=G+\delta G$
    be Hermitian. Define the minimum and maximum eigenvalues of $G$
    \begin{equation}
        \mu_{\min}:=\lambda_{\min}(G),
        \qquad
        \mu_{\max}:=\lambda_{\max}(G),
    \end{equation}
    and suppose that
    \begin{equation}
        \|\delta G\|_F\leq\epsilon_F<\mu_{\min}.
    \end{equation}
    Then $\tilde{G}$ is positive definite.
    Let $U$ and $\tilde{U}$ be the Cholesky factors
    of $G$ and $\tilde{G}$, respectively, with
    positive real diagonal entries. Then
    \begin{equation}
        \|\tilde{U}-U\|_s
        \leq\|\tilde{U}-U\|_F
        \leq
        \frac{\sqrt{\mu_{\max}+\epsilon_F}}
        {\sqrt{2}(\mu_{\min}-\epsilon_F)}
        \epsilon_F.
    \end{equation}
\end{claim}
\begin{proof}
    This proof in here we used the method proved in \cite{CHAG_1996}.
    Consider the path
    \begin{equation}
         G(t):=G+t\,\delta G,
        \qquad t\in[0,1].
    \end{equation}
    Since $\|\delta G\|_s\leq\|\delta G\|_F\leq\epsilon_F$,
    \begin{align}
        \lambda_{\min}(G(t))
        &\geq\mu_{\min}-t\epsilon_F
        \geq\mu_{\min}-\epsilon_F>0,\\
        \lambda_{\max}(G(t))
        &\leq\mu_{\max}+t\epsilon_F
        \leq\mu_{\max}+\epsilon_F.
    \end{align}
    Thus, $G(t)$ is positive definite throughout the path.

    Let $U(t)$ be the Cholesky factor of $G(t)$ with positive real diagonal entries. 
    Its recursive construction implies that $U(t)$ depends smoothly on $t$, with $U(0)=U$ and $U(1)=\tilde{U}$.
    Differentiating
    $G(t)=U(t)^\dagger U(t)$ gives
    \begin{equation}
        \delta G
        =\dot{U}(t)^\dagger U(t)
        +U(t)^\dagger\dot{U}(t).
    \end{equation}
    Define
    \begin{equation}
        E(t):=U(t)^{-\dagger}\delta G\,U(t)^{-1},
        \qquad
        X(t):=\dot{U}(t)U(t)^{-1}.
    \end{equation}
    Multiplying the differentiated identity on the left
    by $U(t)^{-\dagger}$ and on the right by $U(t)^{-1}$
    yields
    \begin{equation}
        E(t)=X(t)^\dagger+X(t).
    \end{equation}

    The matrix $X(t)$ is upper triangular, and its
    diagonal entries are real because
    \begin{equation}
        X(t)_{ii}
        =\frac{\dot{U}(t)_{ii}}{U(t)_{ii}}.
    \end{equation}
    Therefore,
    \begin{equation}
        X(t)_{ij}
        =
        \begin{cases}
            E(t)_{ij}, & i<j,\\
            E(t)_{ii}/2, & i=j,\\
            0, & i>j.
        \end{cases}
    \end{equation}
    Since $E(t)$ is Hermitian,
    \begin{align}
        \|X(t)\|_F^2
        &=\sum_{i<j}|E(t)_{ij}|^2
          +\frac14\sum_i|E(t)_{ii}|^2\\
        &=\frac12\|E(t)\|_F^2
          -\frac14\sum_i|E(t)_{ii}|^2\\
        &\leq\frac12\|E(t)\|_F^2.
    \end{align}

    Using $\dot{U}(t)=X(t)U(t)$, we obtain
    \begin{align}
        \|\dot{U}(t)\|_F
        &\leq\|X(t)\|_F\|U(t)\|_s\\
        &\leq\frac{1}{\sqrt{2}}
          \|U(t)^{-1}\|_s^2
          \|\delta G\|_F\|U(t)\|_s\\
        &=\frac{\sqrt{\lambda_{\max}(G(t))}}
          {\sqrt{2}\lambda_{\min}(G(t))}
          \|\delta G\|_F\\
        &\leq
          \frac{\sqrt{\mu_{\max}+\epsilon_F}}
          {\sqrt{2}(\mu_{\min}-\epsilon_F)}
          \epsilon_F.
    \end{align}
    Finally,
    \begin{equation}
        \tilde{U}-U=\int_0^1\dot{U}(t)\,dt,
    \end{equation}
    so
    \begin{equation}
        \|\tilde{U}-U\|_s
        \leq\|\tilde{U}-U\|_F
        \leq\int_0^1\|\dot{U}(t)\|_F\,dt,
    \end{equation}
    which proves the stated bound.
\end{proof}

For the reconstruction of reduced matrices, it is useful to control the relative perturbation $\Gamma=(\tilde{U}-U)U^{-1}$ directly.
The following lemma provides spectral-norm bounds on $\Gamma$ and on the associated inverse $(I_n+\Gamma)^{-1}$.

\begin{lemma}[Relative perturbation of Cholesky factors]
    \label{lem:relative_cholesky_perturbation}
    Let $G\in\mathbb{C}^{n\times n}$ be Hermitian
    positive definite, and let $\tilde{G}=G+\delta G$
    be Hermitian with
    \begin{equation}
        \|\delta G\|_s\leq\epsilon_g.
    \end{equation}
    Define
    \begin{equation}
        \mu_{\min}:=\lambda_{\min}(G),
        \qquad
        c_n:=\frac12+\lceil\log_2n\rceil,
        \qquad
        \eta:=\frac{\epsilon_g}{\mu_{\min}},
    \end{equation}
    and suppose that
    \begin{equation}
        0\leq\eta<
        \min\left\{
            \frac{1}{4c_n^2},
            \frac{1}{4c_n}
        \right\}.
    \end{equation}
    Then $\tilde{G}$ is positive definite.

    Let $U$ and $\tilde{U}$ be the Cholesky factors of $G$ and $\tilde{G}$, respectively, with positive real diagonal entries. 
    Define
    \begin{equation}
        A:=U^{-\dagger}\delta G U^{-1},
        \qquad
        \Gamma:=(\tilde{U}-U)U^{-1},
        \qquad
        S:=(I_n+\Gamma)^{-1}.
    \end{equation}
    Then $\Gamma$ is upper triangular with real diagonal entries greater than $-1$, and
    \begin{equation}
        \tilde{U}=(I_n+\Gamma)U,
        \qquad
        I_n+A=(I_n+\Gamma)^\dagger(I_n+\Gamma).
    \end{equation}
    Moreover, if
    \begin{equation}
        \label{eq:gamma_upper}
        \|\Gamma\|_s\leq2c_n\eta<\frac12,
    \end{equation}
    then
    \begin{equation}
        \label{eq:cholesky_S_bounds}
        \begin{aligned}
            \|S\|_s
            &\leq\frac{1}{1-\|\Gamma\|_s}\leq2,\\
            \|S-I_n\|_s
            &\leq
            \frac{\|\Gamma\|_s}{1-\|\Gamma\|_s}
            \leq4c_n\eta.
        \end{aligned}
    \end{equation}
\end{lemma}
\begin{proof}
    Since $\|U^{-1}\|_s^2=1/\mu_{\min}$,
    \begin{equation}
        \|A\|_s
        \leq\|U^{-1}\|_s^2\|\delta G\|_s
        \leq\eta<1.
    \end{equation}
    Thus, $I_n+A$ is positive definite, and so is
    $\tilde{G}=U^\dagger(I_n+A)U$.

    The definition of $\Gamma$ gives $\tilde{U}=(I_n+\Gamma)U$ and hence
    \begin{equation}
        I_n+A=(I_n+\Gamma)^\dagger(I_n+\Gamma).
    \end{equation}
    Since both Cholesky factors are upper triangular with positive real diagonal entries, $\Gamma$
    is upper triangular and
    \begin{equation}
        \Gamma_{ii}
        =\frac{\tilde{U}_{ii}}{U_{ii}}-1>-1.
    \end{equation}

    Applying the spectral-norm bound in Eq.~(11) and Corollary~2.3 of
    \cite{dramc_doi:10.1137/S0895479893244717},
    we obtain
    \begin{equation}
        \|\Gamma\|_s
        \leq
        \frac{2c_n\|A\|_s}
        {1+\sqrt{1-4c_n^2\|A\|_s}}
        \leq2c_n\eta<\frac12.
    \end{equation}
    The square root is well defined because $4c_n^2\|A\|_s\leq4c_n^2\eta<1$.

    Finally, the Neumann series gives
    \begin{equation}
        \|S\|_s\leq\frac{1}{1-\|\Gamma\|_s}.
    \end{equation}
    Combining this with $S-I_n=-S\Gamma$ yields
    \begin{equation}
        \|S-I_n\|_s
        \leq
        \frac{\|\Gamma\|_s}{1-\|\Gamma\|_s}
        \leq2\|\Gamma\|_s
        \leq4c_n\eta.
    \end{equation}
\end{proof}

\subsubsection[Perturbation Bounds for Jr]{Perturbation Bound for $J_r$}
\begin{lemma}[Perturbation bound for the reduced matrix]
    \label{claim:alg1_J_r_error}
    Let $G_r,\tilde{G}_r$ and their Cholesky factors
    $U_r,\tilde{U}_r$ satisfy the assumptions of
    \Cref{lem:relative_cholesky_perturbation}
    with $n=r$.
    Let $\widehat{G}_r$ and
    $\tilde{\widehat{G}}_r$ be Hermitian matrices
    satisfying
    \begin{equation}
        \delta\widehat{G}_r
        :=\tilde{\widehat{G}}_r-\widehat{G}_r,
        \qquad
        \|\delta\widehat{G}_r\|_s\leq\epsilon_g.
    \end{equation}
    Define
    \begin{equation}
        J_r:=U_r^{-\dagger}\widehat{G}_rU_r^{-1},
        \qquad
        \tilde{J}_r:=
        \tilde{U}_r^{-\dagger}
        \tilde{\widehat{G}}_r\tilde{U}_r^{-1},
        \qquad
        \delta J_r:=\tilde{J}_r-J_r.
    \end{equation}
    If $\|J_r\|_s\leq1$, then
    \begin{equation}
        \|\delta J_r\|_s
        \leq4(4c_r+1)\frac{\epsilon_g}{\mu_{\min}},
    \end{equation}
    where $\mu_{\min}:=\lambda_{\min}(G_r)$ and
    $c_r:=\frac12+\lceil\log_2r\rceil$.
\end{lemma}
\begin{proof}
    Let $\Gamma$ and $S=(I_r+\Gamma)^{-1}$ be as in
    \Cref{lem:relative_cholesky_perturbation}.
    Writing $\gamma:=\|\Gamma\|_s$, that lemma gives
    \begin{equation}
        \gamma
        \leq2c_r\frac{\epsilon_g}{\mu_{\min}}<\frac12,
        \qquad
        \|S\|_s\leq\frac{1}{1-\gamma},
        \qquad
        \|S-I_r\|_s\leq\frac{\gamma}{1-\gamma}.
    \end{equation}

    Set
    \begin{equation}
        E_r:=U_r^{-\dagger}
        \delta\widehat{G}_rU_r^{-1}.
    \end{equation}
    Since $\|U_r^{-1}\|_s^2=1/\mu_{\min}$,
    \begin{equation}
        \|E_r\|_s
        \leq\|U_r^{-1}\|_s^2
        \|\delta\widehat{G}_r\|_s
        \leq\frac{\epsilon_g}{\mu_{\min}}.
    \end{equation}
    Using $\tilde{U}_r^{-1}=U_r^{-1}S$, we obtain
    \begin{equation}
        \tilde{J}_r=S^\dagger(J_r+E_r)S.
    \end{equation}
    Therefore,
    \begin{equation}
        \delta J_r
        =(S^\dagger-I_r)J_rS
        +J_r(S-I_r)+S^\dagger E_rS.
    \end{equation}
    Using $\|J_r\|_s\leq1$, we conclude that
    \begin{equation}
        \begin{aligned}
            \|\delta J_r\|_s
            &\leq
            \|S-I_r\|_s(\|S\|_s+1)
            +\|S\|_s^2\|E_r\|_s\\
            &\leq
            \frac{\gamma(2-\gamma)
            +\epsilon_g/\mu_{\min}}{(1-\gamma)^2}\\
            &\leq
            4\left(2\gamma+
            \frac{\epsilon_g}{\mu_{\min}}\right)\\
            &\leq
            4(4c_r+1)\frac{\epsilon_g}{\mu_{\min}}.
        \end{aligned}
    \end{equation}
\end{proof}

\subsubsection{Error in the Lanczos Coefficients}
\begin{claim}[Error in the Lanczos Coefficients]
\label{claim:error_lanczos_coefficient}
Let $J_r,\tilde{J}_r\in\mathbb{C}^{r\times r}$ be Hermitian matrices sharing a common orthonormal eigen basis.
Let $f:I\to\mathbb{C}$ be $L_f$-Lipschitz on a real interval $I$ containing the spectra of both matrices, so that
\begin{equation}
    |f(x)-f(y)|\leq L_f|x-y|,
    \qquad x,y\in I.
\end{equation}
Define
\begin{equation}
    \boldsymbol{c}:=f(J_r)\ket{e_0},
    \qquad
    \tilde{\boldsymbol{c}}:=f(\tilde{J}_r)\ket{e_0},
\end{equation}
where $\ket{e_0}$ is the first standard basis vector.
Then, with $\delta\boldsymbol{c}:=\tilde{\boldsymbol{c}}-\boldsymbol{c}$ and
$\delta J_r:=\tilde{J}_r-J_r$,
\begin{equation}
    \|\delta\boldsymbol{c}\|_2
    \leq L_f\|\delta J_r\|_s.
\end{equation}
\end{claim}
\begin{proof}
By simultaneous diagonalization, there exist a unitary matrix
$Q_r$ and real diagonal matrices $D,\delta D$ such that
\begin{equation}
    J_r=Q_r^\dagger D Q_r,
    \qquad
    \tilde{J}_r=Q_r^\dagger(D+\delta D)Q_r.
\end{equation}
Consequently,
$\|\delta J_r\|_s=\|\delta D\|_s$.
Using $\|Q_r\ket{e_0}\|_2=1$ and the Lipschitz condition,
we obtain
\begin{equation}
    \begin{aligned}
        \|\delta\boldsymbol{c}\|_2
        &=
        \left\|
        Q_r^\dagger\bigl(f(D+\delta D)-f(D)\bigr)
        Q_r\ket{e_0}
        \right\|_2\\
        &\leq
        \max_j
        |f(D_{jj}+\delta D_{jj})-f(D_{jj})|\\
        &\leq L_f\max_j|\delta D_{jj}|\\
        &=L_f\|\delta J_r\|_s.
    \end{aligned}
\end{equation}
\end{proof}

\appendixsection{Technical Details for Classical-Shadow S-QSVT}
\label{app:observable-aware-details}


This appendix collects the statistical subroutine, technical lemmas, and resource analyses used in Theorem~\ref{thm:end-to-end-observable-estimation}, Lemma~\ref{lem:observable-krylov-estimation}, and Lemma~\ref{lem:observable-krylov-resources}. Throughout this appendix, $U_H$ is assumed to be an exact $(1,l,0)$-block encoding of $H$.

\subsection{Median-of-means subroutine}

\begin{algorithm}[H]
\caption{\textsc{MedianOfMeans}}
\label{alg:median-of-means}
\LinesNumbered
\KwIn{
samples $\{Z_t\}_{t=1}^{KS}$;
number of groups $K$;
group size $S$
}
Partition the samples into $K$ disjoint groups $\mathcal G_1,\ldots,\mathcal G_K$, each containing $S$ samples\;
\For{$k=1,\ldots,K$}{
    $\overline Z_k
    \gets
    \dfrac{1}{S}
    \sum_{t\in\mathcal G_k} Z_t$\;
}
\Return{
$\operatorname{median}
(\overline Z_1,\ldots,\overline Z_K)$
}\;
\end{algorithm}
Given the $KS$ input values, the group means can be computed in $O(KS)$ arithmetic operations.  The median of the $K$ group means can be found in $O(K)$ time using a linear-time selection algorithm. Hence one call to \textsc{MedianOfMeans} has classical computational cost $O(KS)$.
\subsection{Technical lemmas for the cross-moment estimator}

\subsubsection{Controlled Krylov branches}

\begin{lemma}[Controlled Krylov-branch encoding]
\label{lem:controlled-krylov-branch}
For fixed $a,b$, let $|\Psi_{a,b}\rangle_{CAS}$ be the state prepared in Algorithm~\ref{alg:krylov-shadow-sampling}, let $L_{a,b}:=l\max\{a,b\}$, and define $\Pi_A:=|0^{L_{a,b}}\rangle\langle0^{L_{a,b}}|_A$. Then
\begin{equation}
    (\Pi_A\otimes I_{CS})|\Psi_{a,b}\rangle
    =
    |0^{L_{a,b}}\rangle_A|\chi_{a,b}\rangle_{CS},
    \label{eq:good-controlled-branch}
\end{equation}
where
\begin{equation}
    |\chi_{a,b}\rangle
    :=
    \frac{1}{\sqrt2}
    \left(
        |0\rangle_C H^a|\psi\rangle_S
        +
        |1\rangle_C H^b|\psi\rangle_S
    \right).
    \label{eq:chi-ab}
\end{equation}
The probability of obtaining the ancilla outcome $0^{L_{a,b}}$ is
\begin{equation}
    p_{a,b}
    =
    \frac12
    \left(
        \langle\psi|H^{2a}|\psi\rangle
        +
        \langle\psi|H^{2b}|\psi\rangle
    \right).
    \label{eq:p-ab}
\end{equation}
Furthermore, one implementation of $W_{a,b}$ uses $\max\{a,b\}$ queries to $U_H$ or controlled-$U_H$.
\end{lemma}

\begin{proof}
Let $m:=\max\{a,b\}$. The common ancilla register contains $L_{a,b}=lm$ qubits and may be viewed as $m$ $l$-qubit registers. By sequential block-encoding multiplication, the top-left blocks of $U_{H^a}$ and $U_{H^b}$ are $H^a$ and $H^b$, respectively. Therefore, for
\begin{equation}
    W_{a,b}
    =
    |0\rangle\langle0|_C\otimes U_{H^a}
    +
    |1\rangle\langle1|_C\otimes U_{H^b},
\end{equation}
projection onto the all-zero ancilla subspace gives
\begin{align}
    (\Pi_A\otimes I_{CS})
    W_{a,b}
    |+\rangle_C|0^{L_{a,b}}\rangle_A|\psi\rangle_S
    =
    \frac{|0^{L_{a,b}}\rangle_A}{\sqrt2}
    \left(
        |0\rangle_C H^a|\psi\rangle_S
        +
        |1\rangle_C H^b|\psi\rangle_S
    \right),
\end{align}
which proves Eq.~\eqref{eq:good-controlled-branch}. Its squared norm is
\begin{align}
    p_{a,b}
    =
    \langle\chi_{a,b}|\chi_{a,b}\rangle
    =
    \frac12
    \left(
        \|H^a|\psi\rangle\|_2^2
        +
        \|H^b|\psi\rangle\|_2^2
    \right).
\end{align}
Since $H$ is Hermitian, $\|H^a|\psi\rangle\|_2^2=\langle\psi|H^{2a}|\psi\rangle$, and similarly for $b$, proving Eq.~\eqref{eq:p-ab}. Finally, the two branches can be implemented using $m$ sequential block-encoding layers, with at most one application of $U_H$ or controlled-$U_H$ in each layer.
\end{proof}

\begin{remark}[No postselection penalty]
Algorithm~\ref{alg:krylov-shadow-sampling} does not postselect on $0^{L_{a,b}}$. Failed block-encoding rounds are retained with weight $s_t=0$, so no additional factor $1/p_{a,b}$ appears in the sample complexity.
\end{remark}

\subsubsection{Weighted classical snapshots}

Define the subnormalized good-branch state on $C+S$ by
\begin{equation}
    \sigma_{a,b}
    :=
    {}_A\langle0^{L_{a,b}}|
    \Psi_{a,b}\rangle
    \langle\Psi_{a,b}|
    0^{L_{a,b}}\rangle_A.
    \label{eq:sigma-ab}
\end{equation}
By Lemma~\ref{lem:controlled-krylov-branch}, $\sigma_{a,b}=|\chi_{a,b}\rangle\langle\chi_{a,b}|$ and $\operatorname{Tr}(\sigma_{a,b})=p_{a,b}$.

\begin{lemma}[Unbiased weighted classical snapshot]
\label{lem:weighted-shadow}
For one sample generated by Algorithm~\ref{alg:krylov-shadow-sampling},
\begin{equation}
    \mathbb E[s_t\widehat\rho_t]
    =
    \sigma_{a,b}.
    \label{eq:weighted-shadow-unbiased}
\end{equation}
\end{lemma}

\begin{proof}
Fix a shadow measurement unitary $U$. The joint probability of obtaining the good block-encoding ancilla outcome and computational-basis outcome $b$ on $C+S$ is
\begin{equation}
    \Pr(0^{L_{a,b}},b|U)
    =
    \operatorname{Tr}
    \left[
        U^\dagger|b\rangle\langle b|U\,\sigma_{a,b}
    \right].
\end{equation}
Therefore,
\begin{align}
    \mathbb E[s_t\widehat\rho_t]
    &=
    \mathbb E_{U\sim\mathcal U}
    \sum_b
    \Pr(0^{L_{a,b}},b|U)
    \mathcal M^{-1}
    \left(U^\dagger|b\rangle\langle b|U\right)
    \nonumber\\
    &=
    \mathcal M^{-1}(\mathcal M(\sigma_{a,b}))
    =
    \sigma_{a,b}.
\end{align}
The channel is linear, so the same identity applies to the subnormalized positive operator $\sigma_{a,b}$.
\end{proof}

\subsubsection{Variance bound}

\begin{lemma}[Variance bound]
\label{lem:shadow-cross-variance}
For every $i\in[M]$,
\begin{equation}
    \operatorname{Var}(R_{i,t}),\;
    \operatorname{Var}(I_{i,t})
    \leq
    \nu_O^2.
\end{equation}
\end{lemma}

\begin{proof}
Assume first that $p_{a,b}>0$ and define $\rho_{a,b}^{(\mathrm{good})}:=\sigma_{a,b}/p_{a,b}$. Conditioned on $s_t=1$, the shadow outcome on $C+S$ is distributed as a classical-shadow measurement of $\rho_{a,b}^{(\mathrm{good})}$. Since $s_t^2=s_t$,
\begin{align}
    \mathbb E[R_{i,t}^2]
    =
    p_{a,b}
    \mathbb E_{\rho_{a,b}^{(\mathrm{good})}}
    \left[
        \operatorname{Tr}
        \left((X_C\otimes O_i)\widehat\rho\right)^2
    \right]
    \leq
    p_{a,b}\|X_C\otimes O_i\|_{\mathrm{shadow}}^2
    \leq
    \nu_O^2,
\end{align}
where the first inequality follows from the definition of the shadow norm in Eq.~\eqref{eq:shadow-norm}, and the second uses $p{a,b}\leq1$ and the definition of $\nu_O^2$. Hence $\operatorname{Var}(R_{i,t})\leq\mathbb E[R_{i,t}^2]\leq\nu_O^2$. The same argument applies to $I_{i,t}$. If $p_{a,b}=0$, both weighted estimators vanish identically.
\end{proof}

\subsection{Proof of observable Krylov matrix estimation and its resources}

\begin{proof}[Proof of Lemma~\ref{lem:observable-krylov-estimation}]
By Lemma~\ref{lem:weighted-shadow},
\begin{align}
    \mathbb E[R_{i,t}]
    &=
    \operatorname{Tr}[(X_C\otimes O_i)\sigma_{a,b}],\\
    \mathbb E[I_{i,t}]
    &=
    \operatorname{Tr}[(Y_C\otimes O_i)\sigma_{a,b}].
\end{align}
Using Eq.~\eqref{eq:chi-ab},
\begin{align}
    \mathbb E[R_{i,t}]
    &=
    \frac{\mu_{i,a,b}+\mu_{i,a,b}^*}{2}
    =
    \operatorname{Re}\mu_{i,a,b},\\
    \mathbb E[I_{i,t}]
    &=
    \frac{-i\mu_{i,a,b}+i\mu_{i,a,b}^*}{2}
    =
    \operatorname{Im}\mu_{i,a,b}.
\end{align}
By Lemma~\ref{lem:shadow-cross-variance}, both estimators have variance at most $\nu_O^2$. For
\begin{equation}
    S
    =
    \left\lceil\frac{8\nu_O^2}{\epsilon_S^2}\right\rceil,
\end{equation}
each group mean estimates the corresponding real quantity to accuracy $\epsilon_S/\sqrt2$ with constant success probability. Taking the median over
\begin{equation}
    K
    =
    \left\lceil8\log\left(\frac{2M}{\eta}\right)\right\rceil
\end{equation}
groups and applying a union bound over the $2M$ real quantities gives, with probability at least $1-\eta$,
\begin{align}
    |\widetilde R_{i,a,b}-\operatorname{Re}\mu_{i,a,b}|
    &\leq\frac{\epsilon_S}{\sqrt2},\\
    |\widetilde I_{i,a,b}-\operatorname{Im}\mu_{i,a,b}|
    &\leq\frac{\epsilon_S}{\sqrt2}
\end{align}
simultaneously for every $i\in[M]$. Therefore,
\begin{equation}
    |\widetilde\mu_{i,a,b}-\mu_{i,a,b}|
    \leq
    \epsilon_S.
\end{equation}
Since $N_{a,b}=KS$,
\begin{equation}
    N_{a,b}
    =
    O\!\left(
        \frac{\nu_O^2}{\epsilon_S^2}
        \log\left(\frac{M}{\eta}\right)
    \right).
\end{equation}
Because $(B_{r,i})_{b,a}=(B_{r,i})_{a,b}^*$, it is sufficient to estimate $N_{\mathrm{pair}}=r(r+1)/2$ pairs with $0\leq a\leq b\leq r-1$. Assign failure probability $\eta=\delta_B/N_{\mathrm{pair}}$ to each pair. A union bound then yields Eq.~\eqref{eq:all-B-entrywise} for all $i,a,b$ with probability at least $1-\delta_B$. Summing the sample complexity over all pairs gives
\begin{equation}
    N_B
    =
    O\!\left(
        \frac{r^2\nu_O^2}{\epsilon_S^2}
        \log\!\left(\frac{Mr^2}{\delta_B}\right)
    \right),
\end{equation}
which proves Eq.~\eqref{eq:all-B-sample}.
\end{proof}

\begin{proof}[Proof of Lemma~\ref{lem:observable-krylov-resources}]
Fix a Krylov pair $(a,b)$ and let $m:=\max\{a,b\}$. The unitary
\[
    W_{a,b}
    =
    |0\rangle\langle0|_C\otimes U_{H^a}
    +
    |1\rangle\langle1|_C\otimes U_{H^b}
\]
can be implemented in $m$ sequential block-encoding layers. Indeed, in layer $j\in\{1,\ldots,m\}$, an application of $U_H$ is required on the $|0\rangle_C$ branch only when $j\leq a$, and on the $|1\rangle_C$ branch only when $j\leq b$. If both branches are active, the same $U_H$ is applied independently of the control qubit; if only one branch is active, a controlled-$U_H$ is applied. Thus each layer uses at most one query to $U_H$ or controlled-$U_H$, and therefore
\begin{equation}
    Q_{a,b}^{(\mathrm{shot})}
    =
    m
    =
    \max\{a,b\}.
    \label{eq:fixed-pair-query}
\end{equation}
Since the $m$ block-encoding layers are sequential, the preparation of one classical-shadow sample has depth
\begin{equation}
    D_{a,b}
    =
    O\!\left(
        mD_H+D_{\mathrm{sh}}
    \right),
\end{equation}
Its gate count is similarly
\begin{equation}
    G_{a,b}^{(\mathrm{shot})}
    =
    O\!\left(
        mC_H+G_{\mathrm{sh}}
    \right).
    \label{eq:fixed-pair-gates}
\end{equation}
By Lemma~\ref{lem:observable-krylov-estimation}, each pair requires
\begin{equation}
    N_{a,b}
    =
    O\!\left(
        \frac{\nu_O^2}{\epsilon_S^2}
        \log\!\left(
            \frac{Mr^2}{\delta_B}
        \right)
    \right)
    \label{eq:fixed-pair-resource-samples}
\end{equation}
samples after assigning failure probability $\delta_B/N_{\mathrm{pair}}$ to each of the $N_{\mathrm{pair}}=r(r+1)/2$ upper-triangular Krylov pairs. For $0\leq a\leq b\leq r-1$, we have $\max\{a,b\}=b$. Hence the total number of block-encoding queries is
\begin{align}
    Q_B
    =
    \sum_{b=0}^{r-1}
    \sum_{a=0}^{b}
    N_{a,b}\,b
    =
    O\!\left(
        \frac{\nu_O^2}{\epsilon_S^2}
        \log\!\left(
            \frac{Mr^2}{\delta_B}
        \right)
        \sum_{b=0}^{r-1}(b+1)b
    \right).
\end{align}
Using
\begin{equation}
    \sum_{b=0}^{r-1}(b+1)b
    =
    \frac{r(r-1)(r+1)}{3}
    =
    \Theta(r^3),
    \label{eq:krylov-query-sum}
\end{equation}
we obtain Eq.~\eqref{eq:all-B-query}. Finally, summing Eq.~\eqref{eq:fixed-pair-gates} over all shadow rounds gives
\begin{align}
    G_B
    &=
    O\!\left(
        Q_B C_H
        +
        N_B G_{\mathrm{sh}}
    \right).
\end{align}
Substituting
\[
    Q_B
    =
    O\!\left(
        \frac{r^3\nu_O^2}{\epsilon_S^2}
        \log\!\frac{Mr^2}{\delta_B}
    \right)
\]
and
\[
    N_B
    =
    O\!\left(
        \frac{r^2\nu_O^2}{\epsilon_S^2}
        \log\!\frac{Mr^2}{\delta_B}
    \right)
\]
gives Eq.~\eqref{eq:all-B-gates}. We next account for the classical post-processing cost. For a fixed Krylov pair $(a,b)$, Algorithm~\ref{alg:krylov-shadow-postprocessing} computes $R_{i,t}$ and $I_{i,t}$ for every $i\in[M]$ and every one of the $N_{a,b}$ stored samples.  By the definition of $T_{\mathrm{eval}}(n,\mathcal O)$ in Eq.~\eqref{eq:shadow-evaluation-cost}, computing either $    \operatorname{Tr}\!\left(
        Q_i^{(X)}\widehat\rho_t
    \right)
    \text{or}
    \operatorname{Tr}\!\left(
        Q_i^{(Y)}\widehat\rho_t
    \right)$
requires at most $T_{\mathrm{eval}}(n,\mathcal O)$ classical time. The additional multiplication by $s_t$, accumulation of group means, and median-of-means processing require $O(1)$ amortized work per observable and stored sample.  Thus, for a fixed pair $(a,b)$, the classical cost is
\begin{equation}
    O\!\left(
        M N_{a,b}
        \left[
            T_{\mathrm{eval}}(n,\mathcal O)+1
        \right]
    \right).
\end{equation}
Summing over all upper-triangular Krylov pairs and using $N_B=\sum_{0\leq a\leq b\leq r-1} N_{a,b}$ gives
\begin{equation}
    O\!\left(
        M N_B
        \left[
            T_{\mathrm{eval}}(n,\mathcal O)+1
        \right]
    \right)
\end{equation}
for the classical estimation of all Krylov cross moments. Finally, assembling these estimates into the $M$ Hermitian $r\times r$ matrices $\{\widetilde B_{r,i}\}_{i=1}^M$ requires $O(Mr^2)$ additional operations.  Therefore,
\begin{equation}
    T_B^{\mathrm{cl}}
    =
    O\!\left(
        M N_B
        \left[
            T_{\mathrm{eval}}(n,\mathcal O)+1
        \right]
        +
        Mr^2
    \right),
\end{equation}
which proves Eq.~\eqref{eq:all-B-classical}.
\end{proof}
\subsection{Proof of end-to-end observable estimation}

\begin{proof}[Proof of Theorem~\ref{thm:end-to-end-observable-estimation}]
Let
\begin{equation}
    \Delta B_{r,i}
    :=
    \widetilde B_{r,i}-B_{r,i},
    \qquad
    \Delta\boldsymbol{\lambda}
    :=
    \widetilde{\boldsymbol{\lambda}}-\boldsymbol{\lambda}.
\end{equation}
Lemma~\ref{lem:observable-krylov-estimation}, applied with $\delta_B=\delta/2$, gives the entrywise bound $|(\Delta B_{r,i})_{a,b}|\leq\epsilon_S$ simultaneously for all $i,a,b$ with probability at least $1-\delta/2$. Hence
\begin{equation}
    \|\Delta B_{r,i}\|_2
    \leq
    \|\Delta B_{r,i}\|_F
    \leq
    r\epsilon_S.
    \label{eq:delta-B-shadow}
\end{equation}
Using $\widetilde B_{r,i}=B_{r,i}+\Delta B_{r,i}$ and $\widetilde{\boldsymbol{\lambda}}=\boldsymbol{\lambda}+\Delta\boldsymbol{\lambda}$,
\begin{align}
    \widetilde y_i-y_i
    =\widetilde{\boldsymbol{\lambda}}^{\,\dagger}
    \Delta B_{r,i}
    \widetilde{\boldsymbol{\lambda}}
    +
    \widetilde{\boldsymbol{\lambda}}^{\,\dagger}
    B_{r,i}
    \widetilde{\boldsymbol{\lambda}}
    -
    \boldsymbol{\lambda}^\dagger B_{r,i}\boldsymbol{\lambda}.
\end{align}
The first term satisfies
\begin{equation}
    \left|
        \widetilde{\boldsymbol{\lambda}}^{\,\dagger}
        \Delta B_{r,i}
        \widetilde{\boldsymbol{\lambda}}
    \right|
    \leq
    \|\widetilde{\boldsymbol{\lambda}}\|_2^2
    \|\Delta B_{r,i}\|_2.
\end{equation}
For the second term,
\begin{align}
    \left|
        \widetilde{\boldsymbol{\lambda}}^{\,\dagger}
        B_{r,i}
        \widetilde{\boldsymbol{\lambda}}
        -
        \boldsymbol{\lambda}^\dagger B_{r,i}\boldsymbol{\lambda}
    \right|
    \leq
    \|B_{r,i}\|_2
    \left(
        2\|\boldsymbol{\lambda}\|_2\|\Delta\boldsymbol{\lambda}\|_2
        +
        \|\Delta\boldsymbol{\lambda}\|_2^2
    \right).
\end{align}
Since $\|\boldsymbol{\lambda}\|_2\leq\theta$, $\|\Delta\boldsymbol{\lambda}\|_2\leq\epsilon_\lambda$, and therefore $\|\widetilde{\boldsymbol{\lambda}}\|_2\leq\theta+\epsilon_\lambda$, we obtain
\begin{equation}
    |\widetilde y_i-y_i|
    \leq
    (\theta+\epsilon_\lambda)^2r\epsilon_S
    +
    \beta_O(2\theta\epsilon_\lambda+\epsilon_\lambda^2).
    \label{eq:end-to-end-intermediate-bound}
\end{equation}
With the choice in Eq.~\eqref{eq:end-to-end-epsilon-S}, the right-hand side is exactly $\epsilon_O$, proving the accuracy guarantee.

The coefficient-estimation event fails with probability at most $\delta/2$, and the observable-Krylov-matrix event fails with probability at most $\delta/2$. A union bound therefore gives overall success probability at least $1-\delta$.

Finally, substituting Eq.~\eqref{eq:end-to-end-epsilon-S} and $\delta_B=\delta/2$ into Eq.~\eqref{eq:all-B-sample} gives
\begin{equation}
    N_B
    =
    O\!\left(
        \frac{
            r^4(\theta+\epsilon_\lambda)^4\nu_O^2
        }{
            [\epsilon_O-\beta_O(2\theta\epsilon_\lambda+\epsilon_\lambda^2)]^2
        }
        \log\!\left(\frac{Mr^2}{\delta}\right)
    \right),
\end{equation}
and the total sample complexity is $N_\mathrm{obs}=N_\lambda(\epsilon_\lambda,\delta/2)+N_B$.

For the remaining quantum resources, Lemma~\ref{lem:observable-krylov-resources} gives the query complexity $Q_B$, maximum circuit depth $D_B$, and total gate complexity $G_B$ of the observable-Krylov-matrix estimation stage. Since the coefficient-estimation and observable-estimation stages are executed separately, their query and gate costs add, whereas the maximum circuit depth is the larger of the two stage-wise depths. Therefore,
\begin{equation}
    Q_{\mathrm{obs}}
    =
    Q_\lambda+Q_B,
    \qquad
    D_{\mathrm{obs}}
    =
    \max\{D_\lambda,D_B\},
    \qquad
    G_{\mathrm{obs}}
    =
    G_\lambda+G_B.
\end{equation}
For the classical computational cost,
Lemma~\ref{lem:observable-krylov-resources} gives $T_B^{\mathrm{cl}}$ for estimating and assembling the observable Krylov matrices. Once $\widetilde{\boldsymbol{\lambda}}$ and $\{\widetilde B_{r,i}\}_{i=1}^M$ have been obtained, each final estimate
\begin{equation}
    \widetilde y_i
    =
    \widetilde{\boldsymbol{\lambda}}^\dagger
    \widetilde B_{r,i}
    \widetilde{\boldsymbol{\lambda}}
\end{equation}
can be evaluated in $O(r^2)$ classical time.  Evaluating all $M$ observables therefore requires $O(Mr^2)$ additional operations. Including the classical cost $T_\lambda^{\mathrm{cl}}$ of the coefficient-estimation stage gives
\begin{equation}
    T_{\mathrm{obs}}^{\mathrm{cl}}
    =
    T_\lambda^{\mathrm{cl}}
    +
    T_B^{\mathrm{cl}}
    +
    O(Mr^2),
\end{equation}
as stated in Theorem~\ref{thm:end-to-end-observable-estimation}. This completes the proof.
\end{proof}

\begin{remark}[Bounding $\beta_O$]
If $\|H\|_2\leq1$ and $|\psi\rangle$ is normalized, then
\begin{equation}
    \beta_O
    \leq
    r\max_{i\in[M]}\|O_i\|_2.
\end{equation}
Indeed, $B_{r,i}=A_r^\dagger O_iA_r$ implies
$\|B_{r,i}\|_2\leq\|A_r\|_2^2\|O_i\|_2$, while
$\|A_r\|_2^2=\|G_r\|_2\leq\operatorname{Tr}(G_r)\leq r$.
\end{remark}

\subsection{Local Pauli classical shadows}
\label{app:local-pauli-details}

For the local Pauli ensemble, choose independently on each qubit $j\in C+S$ a measurement basis $P_j\in\{X,Y,Z\}$ uniformly at random. Let $U_j$ satisfy $P_j=U_j^\dagger ZU_j$, and let $b_j\in\{0,1\}$ be the computational-basis outcome after applying $U_j$. Writing $m_j:=(-1)^{b_j}$,
\begin{equation}
    U_j^\dagger|b_j\rangle\langle b_j|U_j
    =
    \frac{I+m_jP_j}{2}.
\end{equation}
For the single-qubit local Pauli measurement channel,
\begin{equation}
    \mathcal M_1^{-1}(X)
    =
    3X-\operatorname{Tr}(X)I,
\end{equation}
and therefore
\begin{equation}
    \widehat\rho_j
    =
    \frac12I+\frac32m_jP_j.
\end{equation}
The full snapshot factorizes as
\begin{equation}
    \widehat\rho
    =
    \bigotimes_{j\in C+S}
    \left(
        3U_j^\dagger|b_j\rangle\langle b_j|U_j-I
    \right).
    \label{eq:local-pauli-snapshot}
\end{equation}
The factorized form in Eq.~\eqref{eq:local-pauli-snapshot} also gives an efficient classical representation of the snapshot: it is sufficient to store the measurement basis $P_j$ and outcome $m_j$ for each qubit in $C+S$.  Hence one local-Pauli snapshot requires $O(n)$ classical storage, without constructing a $2^{n+1}\times2^{n+1}$ matrix.

Now suppose that $O_i$ is a Pauli string and let $\operatorname{supp}(O_i)$ denote its support. Since $\operatorname{Tr}(\widehat\rho_j)=1$, the identity factors outside $\operatorname{supp}(O_i)$ contribute trivially, and therefore
\begin{align}
    \operatorname{Tr}
    \left(
        Q_i^{(X)}\widehat\rho
    \right)
    &=
    \operatorname{Tr}
    \left(
        X_C\widehat\rho_C
    \right)
    \prod_{j\in\operatorname{supp}(O_i)}
    \operatorname{Tr}
    \left(
        O_{i,j}\widehat\rho_j
    \right),
    \\
    \operatorname{Tr}
    \left(
        Q_i^{(Y)}\widehat\rho
    \right)
    &=
    \operatorname{Tr}
    \left(
        Y_C\widehat\rho_C
    \right)
    \prod_{j\in\operatorname{supp}(O_i)}
    \operatorname{Tr}
    \left(
        O_{i,j}\widehat\rho_j
    \right).
    \label{eq:local-pauli-trace-factorization}
\end{align}
For any single-qubit Pauli $R\in\{X,Y,Z\}$, Eq.~\eqref{eq:local-pauli-snapshot} gives
\begin{equation}
    \operatorname{Tr}
    \left(
        R\widehat\rho_j
    \right)
    =
    3m_j\,\mathbf 1[P_j=R].
    \label{eq:local-pauli-single-factor}
\end{equation}
Thus, if $O_i$ has weight at most $k$ and is specified by its support and Pauli labels, computing either $\operatorname{Tr}(Q_i^{(X)}\widehat\rho)$ or $\operatorname{Tr}(Q_i^{(Y)}\widehat\rho)$ requires accessing only the control qubit and the at most $k$ system qubits in $\operatorname{supp}(O_i)$. Consequently,
\begin{equation}
    T_{\mathrm{eval}}(n,\mathcal O)
    =
    O(k),
\end{equation}
which proves Eq.~\eqref{eq:local-pauli-evaluation-cost}. 

The same local-Pauli structure also gives the shadow-norm bound used in Corollary~\ref{cor:local-pauli-specialization}.  Let $O_i$ be a Pauli string of weight $w_i\leq k$.  Then both $Q_i^{(X)}=X_C\otimes O_i$ and $Q_i^{(Y)}=Y_C\otimes O_i$ are Pauli strings of weight $w_i+1$.

For $\alpha\in\{X,Y\}$, $\operatorname{Tr}(Q_i^{(\alpha)}\widehat\rho)$ is nonzero only when the measurement basis on the control qubit is $\alpha$ and the measurement bases on $\operatorname{supp}(O_i)$ agree with the Pauli labels of $O_i$.  Since the local measurement bases are chosen independently and uniformly from $\{X,Y,Z\}$, this occurs with probability
\begin{equation}
    3^{-(w_i+1)}.
\end{equation}
Whenever the bases agree,
\begin{equation}
    \left|
        \operatorname{Tr}
        \left(
            Q_i^{(\alpha)}\widehat\rho
        \right)
    \right|
    =
    3^{w_i+1}.
\end{equation}
Therefore,
\begin{equation}
    \left\|
        Q_i^{(\alpha)}
    \right\|_{\mathrm{shadow}}^2
    \leq
    3^{w_i+1}
    \leq
    3^{k+1},
    \qquad
    \alpha\in\{X,Y\}.
\end{equation}
By the definition of $\nu_O^2$ in Eq.~\eqref{eq:nu-shadow},
\begin{equation}
    \nu_O^2
    \leq
    3^{k+1}.
\end{equation}
Substituting $T_{\mathrm{eval}}(n,\mathcal O)=O(k)$ and $\nu_O^2\leq3^{k+1}$ into Eq.~\eqref{eq:all-B-classical}, together with Eq.~\eqref{eq:all-B-pauli-sample}, gives
\begin{equation}
    T_B^{\mathrm{cl}}
    =
    O\!\left(
        \frac{
            Mkr^2 3^{k+1}
        }{
            \epsilon_S^2
        }
        \log\!\left(
            \frac{Mr^2}{\delta_B}
        \right)
        +
        Mr^2
    \right),
\end{equation}
which is Eq.~\eqref{eq:local-pauli-classical-complexity}. Finally, since each stored local-Pauli snapshot is specified by $O(n)$ basis-and-outcome data, the $N_B$ Krylov-shadow samples require $O(nN_B)$ classical storage in total.


\appendixsection{Technical Details for Observable-Aware S-QSVT}
\label{app:observable-aware-details-algo3}

This appendix collects the technical proofs supporting the Observable-Aware S-QSVT construction and its end-to-end error analysis. We first prove the auxiliary structural and stability results used in the main text, including the finite-history characterization, the history-space realization of the observable-relevant quotient, and the spectral and polynomial-stability properties of the effective dynamics. We then give the complete proof of Theorem~\ref{thm:observable-aware-end-to-end}, including the propagation of the history-matrix and observable-matrix estimation errors and the resulting sample, circuit-depth, and classical computational complexities.

\subsection{Auxiliary structural and stability results}
\label{app:observable-aware-auxiliary-results}

We collect here the auxiliary results used to justify the history-space construction and to control perturbations of the effective dynamics. The first two results establish the relation between the observable-relevant quotient and its finite history-space representation. The final two results provide the spectral and polynomial-stability bounds used in the end-to-end error analysis.

\begin{restatedlemma}[Restatement of Lemma~\ref{lem:finite-history-kernel}]
Suppose $\dim\mathcal Q_{\mathcal O}=s.$ Then $\ker F_s^{\mathcal O}=K^{\mathcal O}.$
Equivalently, for every $\ket{\phi}\in\mathcal K_r$,
\begin{equation}
    F_s^{\mathcal O}(\ket{\phi})=0
    \quad\Longleftrightarrow\quad
    \ket{\phi}\in K^{\mathcal O}.
    \label{eq:history-kernel-equivalence}
\end{equation}
\end{restatedlemma}

\begin{proof}[Proof of Lemma~\ref{lem:finite-history-kernel}]
We first show $K^{\mathcal O}\subseteq\ker F_s^{\mathcal O}.$ If $\ket{\phi}\in K^{\mathcal O}$, then $O_iH^\ell\ket{\phi}=0$ for every $i\in[M]$ and every $\ell\geq0$. In particular, this holds for $\ell=0,\ldots,s-1$, and therefore $F_s^{\mathcal O}(\ket{\phi})=0.$ For the reverse inclusion, suppose $F_s^{\mathcal O}(\ket{\phi})=0.$ By the definition of the history map,
\begin{equation}
    O_iH^\ell\ket{\phi}=0,
    \qquad
    i\in[M],
    \quad
    \ell=0,\ldots,s-1.
    \label{eq:first-s-history-zero}
\end{equation}
The induced operator $\overline H$ acts on the $s$-dimensional quotient space $\mathcal Q_{\mathcal O}$. By the Cayley--Hamilton theorem, there exist coefficients $c_0,\ldots,c_{s-1}$ such that
\begin{equation}
    \overline H^s
    =
    \sum_{j=0}^{s-1}
        c_j\overline H^j.
    \label{eq:cayley-hamilton-quotient}
\end{equation}
Multiplying by $\overline H^m$ gives
\begin{equation}
    \overline H^{s+m}
    =
    \sum_{j=0}^{s-1}
        c_j\overline H^{j+m},
    \qquad
    m\geq0.
    \label{eq:quotient-finite-recurrence}
\end{equation}
Applying this operator identity to $[\ket{\phi}]$ gives
\begin{equation}
    H^{s+m}\ket{\phi}
    -
    \sum_{j=0}^{s-1}
        c_jH^{j+m}\ket{\phi}
    \in
    K^{\mathcal O}.
    \label{eq:state-recurrence-mod-invisible}
\end{equation}
Since $K^{\mathcal O}\subseteq\ker O_i$, we obtain
\begin{equation}
    O_iH^{s+m}\ket{\phi}
    =
    \sum_{j=0}^{s-1}
        c_jO_iH^{j+m}\ket{\phi}.
    \label{eq:history-recurrence}
\end{equation}
Equation~\eqref{eq:first-s-history-zero} gives the base case. Equation~\eqref{eq:history-recurrence} then implies inductively that
\begin{equation}
    O_iH^\ell\ket{\phi}=0
\end{equation}
for every $i\in[M]$ and every $\ell\geq0$. Hence $\ket{\phi}\in K^{\mathcal O}.$ Therefore,
\begin{equation}
    \ker F_s^{\mathcal O}
    =
    K^{\mathcal O}.
\end{equation}
\end{proof}

\begin{restatedcorollary}[Restatement of Corollary~\ref{cor:history-space-isomorphism}]
Suppose $\dim\mathcal Q_{\mathcal O}=s.$ Then
\begin{equation}
    \mathcal Q_{\mathcal O}
    \cong
    \mathcal S_{\mathrm{hist}}^{\mathcal O},
    \label{eq:quotient-history-isomorphism}
\end{equation}
and therefore $\dim\mathcal S_{\mathrm{hist}}^{\mathcal O}=s.$ Moreover, $\{\ket{\widehat\phi_0},\ket{\widehat\phi_1},\ldots,\ket{\widehat\phi_{s-1}}\}$
form a basis of $\mathcal S_{\mathrm{hist}}^{\mathcal O}$.
\end{restatedcorollary}

\begin{proof}
By Lemma~\ref{lem:finite-history-kernel},
\begin{equation}
    \ker F_s^{\mathcal O}
    =
    K^{\mathcal O}.
\end{equation}
The first isomorphism theorem therefore gives
\begin{equation}
    \mathcal K_r/K^{\mathcal O}
    \cong
    F_s^{\mathcal O}(\mathcal K_r)
    =
    \mathcal S_{\mathrm{hist}}^{\mathcal O}.
\end{equation}
Hence
\begin{equation}
    \dim\mathcal S_{\mathrm{hist}}^{\mathcal O}=s.
\end{equation}
By Lemma~\ref{lem:observable-aware-basis}, the quotient vectors
\begin{equation}
    [\ket{\psi}],
    [H\ket{\psi}],
    \ldots,
    [H^{s-1}\ket{\psi}]
\end{equation}
form a basis of $\mathcal Q_{\mathcal O}$. Under the isomorphism induced by $F_s^{\mathcal O}$, these vectors are mapped to
\begin{equation}
    \ket{\widehat\phi_0},
    \ket{\widehat\phi_1},
    \ldots,
    \ket{\widehat\phi_{s-1}},
\end{equation}
which therefore form a basis of the history space.
\end{proof}

\begin{restatedlemma}[Restatement of Lemma~\ref{lem:observable-aware-spectral-structure}]
Under the standing assumptions of this section, suppose that $H=H^\dagger$ and $\|H\|_2\leq1$. Although $J_s^{\mathcal O}$ need not itself be Hermitian because the history basis is generally nonorthogonal, it is similar to a Hermitian contraction. Equivalently, there exist an invertible matrix $X_{\mathcal O}$ and a real diagonal matrix $\Lambda_{\mathcal O}$ such that
\begin{equation}
    J_s^{\mathcal O}
    =
    X_{\mathcal O}
    \Lambda_{\mathcal O}
    X_{\mathcal O}^{-1},
    \qquad
    \operatorname{spec}(\Lambda_{\mathcal O})
    \subseteq[-1,1].
    \label{eq:observable-aware-spectral-decomposition}
\end{equation}
Fix such a diagonalization and define its spectral conditioning by
\begin{equation}
    \kappa_{\mathcal O}
    :=
    \|X_{\mathcal O}\|_2
    \|X_{\mathcal O}^{-1}\|_2.
    \label{eq:observable-aware-spectral-conditioning}
\end{equation}
\end{restatedlemma}

\begin{proof}
Recall from Eq.~\eqref{eq:state-krylov-invariance-OA} that $\mathcal K_r(H,\ket{\psi})$ is $H$-invariant, and from Eq.~\eqref{eq:observable-invisible-invariance} that $K^{\mathcal O}$ is also $H$-invariant. Hence the induced quotient operator $\overline H$ in Eq.~\eqref{eq:quotient-H} is well-defined.

We first identify the spectral structure of $\overline H$. Consider the orthogonal complement of $K^{\mathcal O}$ inside the state-Krylov space,
\begin{equation}
    \mathcal K_r(H,\ket{\psi})
    \cap
    (K^{\mathcal O})^\perp.
    \label{eq:observable-visible-orthogonal-complement}
\end{equation}
This subspace is $H$-invariant. Indeed, if $\ket{\phi}$ belongs to Eq.~\eqref{eq:observable-visible-orthogonal-complement}, then Eq.~\eqref{eq:state-krylov-invariance-OA} gives $H\ket{\phi}\in\mathcal K_r(H,\ket{\psi})$. Moreover, for every $\ket{\eta}\in K^{\mathcal O}$,
\begin{equation}
    \bra{\eta}H\ket{\phi}
    =
    \bra{H\eta}\phi\rangle
    =0,
\end{equation}
because $H=H^\dagger$, Eq.~\eqref{eq:observable-invisible-invariance} implies $H\ket{\eta}\in K^{\mathcal O}$, and $\ket{\phi}\perp K^{\mathcal O}$.

The orthogonal decomposition
\begin{equation}
    \mathcal K_r(H,\ket{\psi})
    =
    K^{\mathcal O}
    \oplus
    \left[
        \mathcal K_r(H,\ket{\psi})
        \cap
        (K^{\mathcal O})^\perp
    \right]
\end{equation}
shows that every equivalence class in $\mathcal Q_{\mathcal O}$ has a unique representative in the orthogonal complement above. Under this identification, $\overline H$ is similar to the restriction of $H$ to that subspace. The restriction is Hermitian and has operator norm at most $\|H\|_2\leq1$; hence it is a Hermitian contraction and its spectrum lies in $[-1,1]$.

Finally, Corollary~\ref{cor:history-space-isomorphism} identifies the quotient dynamics with the history-space dynamics, and Eq.~\eqref{eq:history-effective-dynamics} represents that dynamics in the history basis by $J_s^{\mathcal O}$. Therefore $J_s^{\mathcal O}$ is similar to the Hermitian contraction above, which gives Eq.~\eqref{eq:observable-aware-spectral-decomposition}.
\end{proof}

\begin{restatedlemma}[Restatement of Lemma~\ref{lem:observable-aware-polynomial-stability}]
Let
\begin{equation}
    P(x)=\sum_{k=0}^{d}p_kx^k,
    \qquad
    B_P:=\max_{x\in[-1,1]}|P(x)|.
\end{equation}
There exist universal constants $c,C>0$ such that, for any perturbation $E$ satisfying
\begin{equation}
    \|E\|_2
    \leq
    \frac{c}{\kappa_{\mathcal O}d^2},
    \label{eq:observable-aware-degree-stability-condition}
\end{equation}
we have
\begin{equation}
    \left\|
        P(J_s^{\mathcal O}+E)
        -
        P(J_s^{\mathcal O})
    \right\|_2
    \leq
    C\kappa_{\mathcal O}^2
    B_P d^4
    \|E\|_2.
    \label{eq:observable-aware-degree-stability}
\end{equation}
\end{restatedlemma}

\begin{proof}
Let $\rho_d=1+1/d$ and let $\Gamma_d$ be the Bernstein ellipse with parameter $\rho_d$. Its distance from $[-1,1]$ is $\Omega(d^{-2})$. By the Bernstein--Walsh inequality,
\begin{equation}
    \max_{z\in\Gamma_d}|P(z)|
    \leq
    \rho_d^d B_P
    \leq
    e B_P.
\end{equation}
By Lemma~\ref{lem:observable-aware-spectral-structure},
\begin{equation}
    \|(zI-J_s^{\mathcal O})^{-1}\|_2
    =
    O(\kappa_{\mathcal O}d^2)
\end{equation}
uniformly for $z\in\Gamma_d$. If Eq.~\eqref{eq:observable-aware-degree-stability-condition} holds with sufficiently small universal $c$, the resolvent identity and a Neumann-series argument give
\begin{equation}
    \|(zI-J_s^{\mathcal O}-E)^{-1}\|_2
    =
    O(\kappa_{\mathcal O}d^2).
\end{equation}
Using the Cauchy integral representation and the resolvent identity,
\begin{align}
    P(J_s^{\mathcal O}+E)-P(J_s^{\mathcal O})
    &=
    \frac{1}{2\pi i}
    \oint_{\Gamma_d}
        P(z)
        (zI-J_s^{\mathcal O}-E)^{-1}
        E
        (zI-J_s^{\mathcal O})^{-1}
    \,dz .
\end{align}
Since the length of $\Gamma_d$ is $O(1)$, combining the above bounds proves Eq.~\eqref{eq:observable-aware-degree-stability}.
\end{proof}

\subsection{Proof of the end-to-end Observable-Aware S-QSVT theorem}
\label{app:observable-aware-end-to-end-proof}

We now prove Theorem~\ref{thm:observable-aware-end-to-end}. The proof begins by translating the estimation errors of $G_s^{\mathcal O}$ and $\widehat G_s^{\mathcal O}$ into a perturbation bound for $J_s^{\mathcal O}$. The stability results above then control the induced error in $P(J_s^{\mathcal O})\mathbf e_0$. Combining this bound with the estimation error of $B_{s,i}$ yields the final observable-estimation error. We then derive the stated sample, classical computational, and circuit-depth complexities.

\begin{proof}[Proof of Theorem~\ref{thm:observable-aware-end-to-end}]
\textbf{Accuracy of the history Gram matrices.}
For $0\leq p,q\leq 2s-1$, define
\begin{equation}
    c_{p,q}
    :=
    \bra{\psi}
        H^p
        \Omega_{\mathcal O}
        H^q
    \ket{\psi}.
\end{equation}
Since $\Omega_{\mathcal O}$ is Hermitian,
\begin{equation}
    c_{q,p}
    =
    c_{p,q}^{*}.
\end{equation}
There are
\begin{equation}
    C_{\mathrm{hist}}
    =
    \frac{(2s)(2s+1)}{2}
    =
    s(2s+1)
\end{equation}
pairs with $0\leq p\leq q\leq2s-1$. For each such pair, Algorithm~\ref{alg:observable-aware-qsvt} first generates an observable-independent Krylov-shadow dataset using Algorithm~\ref{alg:krylov-shadow-sampling} and then estimates the singleton observable $\Omega_{\mathcal O}$ using Algorithm~\ref{alg:krylov-shadow-postprocessing} with failure probability
\begin{equation}
    \delta_{\mathrm{hist}}^{\mathrm{pair}}
    =
    \frac{
        \delta_{\mathrm{hist}}
    }{
        C_{\mathrm{hist}}
    }.
\end{equation}
By Lemma~\ref{lem:observable-krylov-estimation} and a union bound, with probability at least $1-\delta_{\mathrm{hist}}$,
\begin{equation}
    \left|
        \widetilde c_{p,q}
        -
        c_{p,q}
    \right|
    \leq
    \eta_{\mathrm{hist}}
    \label{eq:history-cross-moment-error}
\end{equation}
simultaneously for every $0\leq p,q\leq2s-1$. Conditioned on this event, for every $a,b$,
\begin{align}
    \left|
        (\widetilde G_s^{\mathcal O})_{a,b}
        -
        (G_s^{\mathcal O})_{a,b}
    \right|
    &\leq
    \sum_{\ell=0}^{s-1}
    \left|
        \widetilde c_{a+\ell,b+\ell}
        -
        c_{a+\ell,b+\ell}
    \right|
    \leq
    s\eta_{\mathrm{hist}}
    \leq
    \frac{\epsilon_G}{s}.
\end{align}
Hence
\begin{align}
    \left\|
        \widetilde G_s^{\mathcal O}
        -
        G_s^{\mathcal O}
    \right\|_2
    \leq
    \left\|
        \widetilde G_s^{\mathcal O}
        -
        G_s^{\mathcal O}
    \right\|_F
    \leq
    \epsilon_G.
    \label{eq:observable-aware-history-G-error-proof}
\end{align}
Similarly,
\begin{align}
    \left|
        (\widetilde{\widehat G}_s^{\mathcal O})_{a,b}
        -
        (\widehat G_s^{\mathcal O})_{a,b}
    \right|
    &\leq
    \sum_{\ell=0}^{s-1}
    \left|
        \widetilde c_{a+\ell,b+\ell+1}
        -
        c_{a+\ell,b+\ell+1}
    \right|
    \leq
    s\eta_{\mathrm{hist}}
    \leq
    \frac{\epsilon_{\widehat G}}{s},
\end{align}
and therefore
\begin{equation}
    \left\|
        \widetilde{\widehat G}_s^{\mathcal O}
        -
        \widehat G_s^{\mathcal O}
    \right\|_2
    \leq
    \epsilon_{\widehat G}.
    \label{eq:observable-aware-history-Ghat-error-proof}
\end{equation}

\noindent\textbf{Accuracy of the observable Krylov matrices.}
There are
\begin{equation}
    C_B
    =
    \frac{s(s+1)}{2}
\end{equation}
distinct Krylov pairs
$0\leq a\leq b\leq s-1$.
For each pair, Algorithm~\ref{alg:observable-aware-qsvt} generates a Krylov-shadow dataset using Algorithm~\ref{alg:krylov-shadow-sampling} and then applies Algorithm~\ref{alg:krylov-shadow-postprocessing} to estimate the cross moments for all $M$ observables simultaneously with precision $\epsilon_S$ and failure probability
\begin{equation}
    \delta_B^{\mathrm{pair}}
    =
    \frac{\delta_B}{C_B}.
\end{equation}
Another union bound therefore implies that, with probability at least
$1-\delta_B$,
\begin{equation}
    \left|
        (\widetilde B_{s,i})_{a,b}
        -
        (B_{s,i})_{a,b}
    \right|
    \leq
    \epsilon_S.
    \label{eq:observable-aware-B-entrywise-error-proof}
\end{equation}
simultaneously for all
$i\in[M]$ and $a,b$.
Since
\begin{equation}
    \delta_{\mathrm{hist}}
    +
    \delta_B
    =
    \delta,
\end{equation}
the history-matrix and observable-matrix guarantees hold simultaneously with probability at least $1-\delta$.

\noindent\textbf{Error in the effective history-space dynamics.}
On the event above, Weyl's inequality gives
\begin{align}
    \lambda_{\min}
    \left(
        \widetilde G_s^{\mathcal O}
    \right)
    &\geq
    \lambda_{\min}
    \left(
        G_s^{\mathcal O}
    \right)
    -
    \left\|
        \widetilde G_s^{\mathcal O}
        -
        G_s^{\mathcal O}
    \right\|_2
    \geq
    \gamma_{\mathcal O}
    -
    \epsilon_G
    >
    0.
\end{align}
Thus $\widetilde G_s^{\mathcal O}$ is invertible and
\begin{equation}
    \left\|
        \left(
            \widetilde G_s^{\mathcal O}
        \right)^{-1}
    \right\|_2
    \leq
    \frac{1}{
        \gamma_{\mathcal O}
        -
        \epsilon_G
    }.
    \label{eq:observable-aware-inverse-bound}
\end{equation}
In particular, Algorithm~\ref{alg:observable-aware-qsvt} does not return failure on this event. Using
\begin{equation}
    G_s^{\mathcal O}J_s^{\mathcal O}
    =
    \widehat G_s^{\mathcal O}
\end{equation}
and
\begin{equation}
    \widetilde G_s^{\mathcal O}
    \widetilde J_s^{\mathcal O}
    =
    \widetilde{\widehat G}_s^{\mathcal O},
\end{equation}
we obtain
\begin{align}
    \widetilde G_s^{\mathcal O}
    \left(
        \widetilde J_s^{\mathcal O}
        -
        J_s^{\mathcal O}
    \right)
    &=
    \left(
        \widetilde{\widehat G}_s^{\mathcal O}
        -
        \widehat G_s^{\mathcal O}
    \right)
    -
    \left(
        \widetilde G_s^{\mathcal O}
        -
        G_s^{\mathcal O}
    \right)
    J_s^{\mathcal O}.
\end{align}
Therefore,
\begin{align}
    \left\|
        \widetilde J_s^{\mathcal O}
        -
        J_s^{\mathcal O}
    \right\|_2
    \leq
    \left\|
        \left(
            \widetilde G_s^{\mathcal O}
        \right)^{-1}
    \right\|_2
    \left[
        \epsilon_{\widehat G}
        +
        \epsilon_G
        \left\|
            J_s^{\mathcal O}
        \right\|_2
    \right]
    \leq
    \frac{
        \epsilon_{\widehat G}
        +
        \epsilon_G
        \left\|
            J_s^{\mathcal O}
        \right\|_2
    }{
        \gamma_{\mathcal O}
        -
        \epsilon_G
    }
    =
    \epsilon_J.
    \label{eq:observable-aware-J-bound}
\end{align}

\noindent\textbf{Error in the observable-aware Krylov coefficients.}
Let
\begin{equation}
    E_J
    :=
    \widetilde J_s^{\mathcal O}
    -
    J_s^{\mathcal O}.
\end{equation}
By Eq.~\eqref{eq:observable-aware-J-bound},
\begin{equation}
    \|E_J\|_2
    \leq
    \epsilon_J.
\end{equation}
Under
Eq.~\eqref{eq:observable-aware-polynomial-stability-condition},
Lemma~\ref{lem:observable-aware-polynomial-stability} therefore gives
\begin{equation}
    \left\|
        P(\widetilde J_s^{\mathcal O})
        -
        P(J_s^{\mathcal O})
    \right\|_2
    \leq
    C\kappa_{\mathcal O}^2
    B_Pd^4\epsilon_J
    =
    \epsilon_\lambda.
\end{equation}
Since $\|\mathbf e_0\|_2=1$,
\begin{align}
    \left\|
        \widetilde{\boldsymbol{\lambda}}^{\mathcal O}
        -
        \boldsymbol{\lambda}^{\mathcal O}
    \right\|_2
    &=
    \left\|
        \left[
            P(\widetilde J_s^{\mathcal O})
            -
            P(J_s^{\mathcal O})
        \right]
        \mathbf e_0
    \right\|_2
    \\
    &\leq
    \epsilon_\lambda.
\end{align}

\noindent\textbf{Error in the observable estimates.}
Let
\begin{equation}
    \Delta\boldsymbol{\lambda}
    :=
    \widetilde{\boldsymbol{\lambda}}^{\mathcal O}
    -
    \boldsymbol{\lambda}^{\mathcal O}.
\end{equation}
By Lemma~\ref{lem:observable-aware-reconstruction},
\begin{equation}
    y_i
    =
    \left(
        \boldsymbol{\lambda}^{\mathcal O}
    \right)^\dagger
    B_{s,i}
    \boldsymbol{\lambda}^{\mathcal O}.
\end{equation}
Hence
\begin{align}
    |\widetilde y_i-y_i|
    \leq
    \left|
        \left(
            \widetilde{\boldsymbol{\lambda}}^{\mathcal O}
        \right)^\dagger
        \left(
            \widetilde B_{s,i}
            -
            B_{s,i}
        \right)
        \widetilde{\boldsymbol{\lambda}}^{\mathcal O}
    \right|
    +
    \left|
        \left(
            \widetilde{\boldsymbol{\lambda}}^{\mathcal O}
        \right)^\dagger
        B_{s,i}
        \widetilde{\boldsymbol{\lambda}}^{\mathcal O}
        -
        \left(
            \boldsymbol{\lambda}^{\mathcal O}
        \right)^\dagger
        B_{s,i}
        \boldsymbol{\lambda}^{\mathcal O}
    \right|.
    \label{eq:observable-aware-error-decomposition}
\end{align}
Since every entry of $\widetilde B_{s,i}-B_{s,i}$ has magnitude at most $\epsilon_S$,
\begin{equation}
    \left\|
        \widetilde B_{s,i}
        -
        B_{s,i}
    \right\|_2
    \leq
    \left\|
        \widetilde B_{s,i}
        -
        B_{s,i}
    \right\|_F
    \leq
    s\epsilon_S.
\end{equation}
Moreover,
\begin{equation}
    \left\|
        \widetilde{\boldsymbol{\lambda}}^{\mathcal O}
    \right\|_2
    \leq
    \theta_{\mathcal O}
    +
    \epsilon_\lambda.
\end{equation}
Thus the first term in Eq.~\eqref{eq:observable-aware-error-decomposition} is at most
\begin{equation}
    s\epsilon_S
    \left(
        \theta_{\mathcal O}
        +
        \epsilon_\lambda
    \right)^2.
\end{equation}
For the second term,
\begin{align}
    &
    \left(
        \widetilde{\boldsymbol{\lambda}}^{\mathcal O}
    \right)^\dagger
    B_{s,i}
    \widetilde{\boldsymbol{\lambda}}^{\mathcal O}
    -
    \left(
        \boldsymbol{\lambda}^{\mathcal O}
    \right)^\dagger
    B_{s,i}
    \boldsymbol{\lambda}^{\mathcal O}
    \nonumber\\
    &=
    \left(
        \Delta\boldsymbol{\lambda}
    \right)^\dagger
    B_{s,i}
    \boldsymbol{\lambda}^{\mathcal O}
    +
    \left(
        \boldsymbol{\lambda}^{\mathcal O}
    \right)^\dagger
    B_{s,i}
    \Delta\boldsymbol{\lambda}
    +
    \left(
        \Delta\boldsymbol{\lambda}
    \right)^\dagger
    B_{s,i}
    \Delta\boldsymbol{\lambda}.
\end{align}
Using
\begin{equation}
    \left\|
        \Delta\boldsymbol{\lambda}
    \right\|_2
    \leq
    \epsilon_\lambda,
    \qquad
    \left\|
        \boldsymbol{\lambda}^{\mathcal O}
    \right\|_2
    =
    \theta_{\mathcal O},
    \qquad
    \left\|
        B_{s,i}
    \right\|_2
    \leq
    \beta_{\mathcal O},
\end{equation}
we obtain
\begin{equation}
    \left|
        \left(
            \widetilde{\boldsymbol{\lambda}}^{\mathcal O}
        \right)^\dagger
        B_{s,i}
        \widetilde{\boldsymbol{\lambda}}^{\mathcal O}
        -
        \left(
            \boldsymbol{\lambda}^{\mathcal O}
        \right)^\dagger
        B_{s,i}
        \boldsymbol{\lambda}^{\mathcal O}
    \right|
    \leq
    \beta_{\mathcal O}
    \left(
        2\theta_{\mathcal O}\epsilon_\lambda
        +
        \epsilon_\lambda^2
    \right).
\end{equation}
Combining the two contributions gives
\begin{equation}
    |\widetilde y_i-y_i|
    \leq
    s\epsilon_S
    \left(
        \theta_{\mathcal O}+\epsilon_\lambda
    \right)^2
    +
    \beta_{\mathcal O}
    \left(
        2\theta_{\mathcal O}\epsilon_\lambda
        +
        \epsilon_\lambda^2
    \right).
\end{equation}
By Eq.~\eqref{eq:observable-aware-target-condition},
\begin{equation}
    |\widetilde y_i-y_i|
    \leq
    \epsilon_O
    \qquad
    \text{for all }i\in[M].
\end{equation}

\medskip
\noindent
\textbf{Sample complexity.}
For the history matrices, the algorithm estimates
\begin{equation}
    C_{\mathrm{hist}}
    =
    s(2s+1)
    =
    O(s^2)
\end{equation}
distinct Krylov cross moments. Each one is estimated to precision
\begin{equation}
    \eta_{\mathrm{hist}}
    =
    \frac{
        \min\{\epsilon_G,\epsilon_{\widehat G}\}
    }{
        s^2
    }.
\end{equation}
By Lemma~\ref{lem:observable-krylov-estimation}, the number of samples per pair is
\begin{equation}
    O\!\left(
        \frac{
            s^4\nu_{\Omega}^2
        }{
            \min\{\epsilon_G,\epsilon_{\widehat G}\}^2
        }
        \log\!\left(
            \frac{s^2}{\delta}
        \right)
    \right).
\end{equation}
Multiplying by $O(s^2)$ pairs gives
\begin{equation}
    N_{\mathrm{hist}}
    =
    O\!\left(
        \frac{
            s^6\nu_{\Omega}^2
        }{
            \min\{\epsilon_G,\epsilon_{\widehat G}\}^2
        }
        \log\!\left(
            \frac{s^2}{\delta}
        \right)
    \right).
    \label{eq:observable-aware-history-samples}
\end{equation}
For the observable Krylov matrices, there are $C_B=O(s^2)$ distinct upper-triangular Krylov pairs. For each pair, all $M$ observables are estimated simultaneously using
\begin{equation}
    O\!\left(
        \frac{
            \nu_O^2
        }{
            \epsilon_S^2
        }
        \log\!\left(
            \frac{Ms^2}{\delta}
        \right)
    \right)
\end{equation}
samples. Thus
\begin{equation}
    N_B
    =
    O\!\left(
        \frac{
            s^2\nu_O^2
        }{
            \epsilon_S^2
        }
        \log\!\left(
            \frac{Ms^2}{\delta}
        \right)
    \right).
    \label{eq:observable-aware-B-samples}
\end{equation}
Adding the two contributions gives Eq.~\eqref{eq:observable-aware-total-samples}.

\medskip
\noindent
\textbf{Classical computational complexity.}
We first consider the history-matrix estimation in Step~1.  By
Eq.~\eqref{eq:history-evaluation-cost}, computing either
\(
\operatorname{Tr}(Q_{\Omega}^{(X)}\widehat\rho)
\)
or
\(
\operatorname{Tr}(Q_{\Omega}^{(Y)}\widehat\rho)
\)
from one stored snapshot requires at most
$T_{\mathrm{eval}}(n,\{\Omega_{\mathcal O}\})$ classical time.
The remaining operations performed by
Algorithm~\ref{alg:krylov-shadow-postprocessing} contribute $O(1)$
amortized work per stored sample.  Summing over the
$N_{\mathrm{hist}}$ history-stage samples therefore gives
\begin{equation}
    O\!\left(
        N_{\mathrm{hist}}
        \left[
            T_{\mathrm{eval}}(n,\{\Omega_{\mathcal O}\})+1
        \right]
    \right)
\end{equation}
classical time for estimating the cross moments
$\{\widetilde c_{p,q}\}$.  Constructing
$\widetilde G_s^{\mathcal O}$ and
$\widetilde{\widehat G}_s^{\mathcal O}$ from these moments requires
$O(s^3)$ additional arithmetic operations, since each of their
$O(s^2)$ entries is a sum of $s$ estimated moments.

Step~2 reconstructs
$\widetilde J_s^{\mathcal O}
=(\widetilde G_s^{\mathcal O})^{-1}
\widetilde{\widehat G}_s^{\mathcal O}$.
Using standard dense linear algebra, the singularity test and this
matrix reconstruction require $O(s^3)$ classical time.  In Step~3,
only the vector
$P(\widetilde J_s^{\mathcal O})\mathbf e_0$ is needed.  Writing
$P(x)=\sum_{k=0}^{d}p_kx^k$, this vector can be evaluated by Horner's
rule using $d$ dense $s\times s$ matrix-vector multiplications, for a
total cost $O(ds^2)$.

For Step~4, Lemma~\ref{lem:observable-krylov-resources}, applied with
$r=s$, gives
\begin{equation}
    O\!\left(
        M N_B
        \left[
            T_{\mathrm{eval}}(n,\mathcal O)+1
        \right]
        +
        Ms^2
    \right)
\end{equation}
classical time to estimate and assemble all
$\{\widetilde B_{s,i}\}_{i=1}^{M}$.  Finally, each quadratic form
\begin{equation}
    \widetilde y_i
    =
    (\widetilde{\boldsymbol{\lambda}}^{\mathcal O})^\dagger
    \widetilde B_{s,i}
    \widetilde{\boldsymbol{\lambda}}^{\mathcal O}
\end{equation}
requires $O(s^2)$ arithmetic operations, so Step~5 contributes
$O(Ms^2)$ in total.  Combining all classical steps and absorbing the
two $O(s^3)$ contributions and the two $O(Ms^2)$ contributions into
constant factors yields
\begin{align}
    T_{\mathrm{OA}}^{\mathrm{cl}}
    =
    O\!\Bigl(
        &N_{\mathrm{hist}}
        \bigl[
            T_{\mathrm{eval}}(n,\{\Omega_{\mathcal O}\})+1
        \bigr]
        +
        M N_B
        \bigl[
            T_{\mathrm{eval}}(n,\mathcal O)+1
        \bigr]
        \nonumber\\
        &\qquad
        +
        s^3
        +
        d s^2
        +
        M s^2
    \Bigr),
\end{align}
which proves Eq.~\eqref{eq:observable-aware-classical-complexity}.

\medskip
\noindent
\textbf{Circuit depth.}
The history-matrix estimation requires Krylov cross moments with $0\leq p,q\leq2s-1$. For a fixed pair $(p,q)$, one shadow sample uses $\max\{p,q\}=O(s)$ sequential applications of $U_H$ or controlled-$U_H$. Hence its circuit depth is
\[
    O\!\left(sD_H+D_{\mathrm{sh}}\right).
\]
The estimation of $B_{s,i}$ only requires $0\leq a,b\leq s-1$ and satisfies the same bound. All remaining steps are classical. Therefore the maximum quantum circuit depth is $O(sD_H+D_{\mathrm{sh}})$, proving Eq.~\eqref{eq:observable-aware-depth}.
\end{proof}

\end{document}